\documentclass[mnsc,nonblindrev]{informs3_for_WP} 

\OneAndAHalfSpacedXI 

\usepackage{verbatim} 	
\usepackage{dsfont} 		
\usepackage{nicefrac}       
\usepackage{xspace}
\usepackage{mathtools}
\usepackage{setspace}
\usepackage{lineno}

\usepackage{tikz}
  \usetikzlibrary{positioning}
  \usetikzlibrary{arrows}
  \usetikzlibrary{fit}

\usepackage{algorithm}
\usepackage[noend]{algpseudocode}

\newcommand{\LinesNumbered}{%
  \setboolean{algocf@linesnumbered}{true}%
  \renewcommand{\algocf@linesnumbered}{\everypar={\nl}}}%

\let\oldnl\nl
\newcommand{\nonl}{\renewcommand{\nl}{\let\nl\oldnl}}

\algblock{Input}{EndInput}
\algnotext{EndInput}
\algblock{Output}{EndOutput}
\algnotext{EndOutput}

\newlength{\maxwidth}
\newcommand{\algalign}[2]
{\makebox[\maxwidth][r]{$#1{}$}${}#2$}

\usepackage{enumitem}
\setlistdepth{15}

\graphicspath{{Figs/}} 

\usepackage{floatpag}
\usepackage[figuresleft]{rotating} 	
\usepackage{lscape}			
\usepackage{pdflscape}
\usepackage{siunitx} 			
\usepackage{booktabs} 
\usepackage{multirow} 			

\newlength\tbspace
\newcolumntype{L}{c<{\hspace{\tbspace}}}

\usepackage{setspace}

\usepackage{nicefrac}

\usepackage{subfigure}

\usepackage[hidelinks]{hyperref}
\hypersetup{
    linkcolor=blue,
    filecolor=blue,      
    urlcolor=blue,
    pdftitle={Overleaf Example},
    pdfpagemode=FullScreen
    }

\usepackage{shortcuts}	

\usepackage{xcolor}

\usepackage{natbib}
 \bibpunct[, ]{(}{)}{,}{a}{}{,}%
 \def\bibfont{\small}%
\TheoremsNumberedThrough     

\EquationsNumberedThrough    

\MANUSCRIPTNO{MS-MKG-2024-06769} 

\begin{document}



\RUNTITLE{Detecting and Mitigating Group Bias in Heterogeneous Treatment Effects}


\TITLE{Detecting and Mitigating Group Bias in Heterogeneous Treatment Effects}

\ARTICLEAUTHORS{%
\AUTHOR{Joel Persson\footnote{This research was initiated during the author's PhD at ETH Zurich. It was completed independently of the author’s current employment at Spotify.}\footnote{Corresponding author, email: joelpersson@spotify.com}
}
\AFF{Spotify} 
\AUTHOR{Jurri{\"e}n Bakker}
\AFF{Booking.com}
\AUTHOR{Dennis Bohle}
\AFF{Booking.com}
\AUTHOR{Stefan Feuerriegel}
\AFF{Munich Center for Machine Learning \& LMU Munich}
\AUTHOR{Florian von Wangenheim}
\AFF{ETH Zurich}
} 

\ABSTRACT{%
Heterogeneous treatment effects (HTEs) are increasingly estimated using machine learning models that produce highly personalized predictions of treatment effects. In practice, however, predicted treatment effects are rarely interpreted, reported, or audited at the individual level but, instead, are often aggregated to broader subgroups, such as demographic segments, risk strata, or markets. We show that such aggregation can induce systematic bias of the group-level causal effect: even when models for predicting the individual-level conditional average treatment effect (CATE) are correctly specified and trained on data from randomized experiments, aggregating the predicted CATEs up to the group level does \textit{not}, in general, recover the corresponding group average treatment effect (GATE). We develop a unified statistical framework to detect and mitigate this form of group bias in randomized experiments. We first define group bias as the discrepancy between the model-implied and experimentally identified GATEs, derive an asymptotically normal estimator, and then provide a simple-to-implement statistical test. For mitigation, we propose a shrinkage-based bias-correction, and show that the theoretically optimal and empirically feasible solutions have closed-form expressions. The framework is fully general, imposes minimal assumptions, and only requires computing sample moments. We analyze the economic implications of mitigating detected group bias for profit-maximizing personalized targeting, thereby characterizing when bias correction alters targeting decisions and profits, and the trade-offs involved. Applications to large-scale experimental data at major digital platforms validate our theoretical results and demonstrate empirical performance.
}%


\maketitle

%


\vspace{-1cm}

\section{Introduction}\label{sec:intro}

Heterogeneous treatment effects (HTEs)---the presence of systematic variation in treatment effects across individuals or subgroups---have become central to how organizations learn about and target interventions across a wide range of domains, including marketing \citep[e.g.,][]{Lemmens2025, Hitsch2024} and healthcare \citep[e.g.,][]{Feuerriegel2024, Kraus2024}. Recently developed methods based on machine learning (ML) \citep[e.g.,][]{Athey2016, Wager2018, Chernozhukov2018} have made it possible to estimate highly personalized HTEs by modeling conditional average treatment effects (CATEs) as flexible functions of rich, high-dimensional, and individual-level covariates.

While individualized CATE predictions have been adopted for automated decision-making related to personalization \citep{Lemmens2025}, treatment effect estimates themselves are typically not interpreted, reported, or audited at the individual level. Instead, they are \textit{aggregated} to subgroups of substantive interest to the business, organization, or research question at hand, such as demographic segments, risk strata, or markets. This practice reflects both an interest in more generalizable subgroup-level effects and the fact that ML-based CATEs are functions. For instance, digital platforms may personalize recommendations to individual users while inferring effects at the segment level \citep{Lemmens2025}; medical professionals may deploy individualized treatment rules while evaluating effects by clinically relevant strata (such as young vs. old or smokers vs. non-smokers) \citep{Hernan2023}; and firms and labor market programs may use ML systems for personalization while reporting outcomes for coarser groups defined by age, location, or other characteristics \citep{agrawal2025economics, athey2022effective}. More generally, whenever CATE estimates are plotted, tabulated, or summarized over bins or discrete covariates, they are implicitly being summarized to treatment effects for groups. 

Despite this widespread practice, it is less recognized that such subgroup-level summaries, in general, do \emph{not} recover a causal effect. For this to hold, the appropriately weighted group-average of the CATE must equal the estimand of the corresponding group-average treatment effect (GATE). This is a strong requirement and may fail even when the CATE is point-identified, correctly specified, and unbiased and consistently estimated. As a result, subgroup-level aggregates of CATEs learned even under ideal conditions may lack a causal interpretation. More generally, this problem reflects the distinction between unbiasedness, confounding, and collapsibility in causal inference \citep{Greenland1999, Huitfeldt2019, Didelez2022, Colnet2023}.\footnote{\SingleSpacedXI\footnotesize Collapsibility refers to the property that a conditional effect measure, when properly marginalized, recovers the corresponding marginal effect measure. An estimator may be unbiased and unconfounded yet not collapsible \citep[see, e.g.,][]{Greenland1999}.}

In this paper, we study the problem of detecting and mitigating this group bias in personalized CATE predictions from experimental data. Using a stylized example (Section \ref{sec:motivating_example}), we first show that even state-of-the-art CATE learners trained on experimental data can exhibit systematic group bias relative to the corresponding GATE. Perhaps surprisingly, including group indicators, estimating separate CATE models per group, or using more training data, does not necessarily remove the group bias and may even be practically infeasible. In many settings, groups of interest are defined \textit{ex~post} by different stakeholders, such as, managers, researchers, or auditors, and is often unavailable at training time due to privacy, legal, or operational constraints.\footnote{\SingleSpacedXI\footnotesize Even if group indicators are included, they can be dropped by regularization, effectively limiting the group-wise heterogeneity in treatment effects captured by the model. Moreover, re-training the model separately within each group amounts to re-estimating a nonparametric treatment response function at smaller and potentially noiser subsets of the data, thus requiring even stronger assumptions than for the original model that is presumed to be biased (such as sufficient overlap \emph{within} each group), and may therefore introduce new forms of bias. Assuming regularization does not drop the indicators and that the stronger within-group identification conditions are met, estimation is still subject to differential sample sizes and signal-to-noise across groups, which may still introduce group bias.} Another key challenge of the problem is that only randomized experiments provide reliable benchmarks for model-implied estimates, yet they cannot produce evidence of treatment effects at the same level of granularity as CATEs from ML. Because potential outcomes are never jointly observed \citep{Holland1986} and personalized CATE estimates condition on many individual-level covariates, model-free experimental estimators such as difference-in-means can only be applied at more aggregate levels, such as those that identify GATEs. Thus, rather than measuring or correcting bias at the individual level, we aim to reliably detect and mitigate bias in CATE predictions at the group level---for which experiments provide nonparametric identification and model-free estimation.

We aim to make three contributions:

\emph{First}, we introduce a general methodology for detecting group bias in CATE predictions from randomized experiments (Section~\ref{sec:framework}). We define group bias as a causal estimand capturing the discrepancy between model-implied and experimentally identified GATEs, that is, whether aggregating predicted CATEs recovers the corresponding subgroup treatment effects. We introduce a general, asymptotically normal estimator of this group bias and provide a formal statistical test for detection. The resulting procedure is nonparametric, model-agnostic, and applies to both binary and continuous outcomes, different scales of treatment effects (i.e., additive and relative effect scales), and arbitrary CATE learners under minimal assumptions. 

\emph{Second}, we propose a shrinkage-based bias-correction approach for mitigation. A na{\"i}ve correction by simply subtracting the estimated bias per group, akin to standard bias-correction, removes group bias only in expectation. In finite samples, this may \emph{increase} the dispersion of group bias among the groups, as it tends to overcorrect for smaller or noisier groups. As a solution, we formulate mitigation as choosing how much to debias (i.e., how much to shrink the correction) in order to minimize the expected loss (risk) of residual group bias. For natural loss functions, we derive the closed-form oracle minimizer as well as feasible estimators, both of which automatically adjust the debiasing to the signal-to-noise ratio of the estimated group bias. We apply our framework to large-scale A/B test data from a leading online travel platform to demonstrate empirical performance and validate our theoretical results  (Sections~\ref{sec:application} and ~\ref{sec:targeting}). 

\emph{Third}, we analyze the implications of mitigating detected group bias in personalized CATE predictions for decision-making. By focusing on profit-maximizing personalized targeting, we characterize how bias correction alters the optimal decision rule, targeting decisions, and expected profits, and the factors that contribute to this (Section~\ref{sec:targeting}). We illustrate our theoretical insights using counterfactual off-policy evaluation on the Criteo Uplift Prediction Dataset \citep{Diemert2018}. We then discuss the resulting trade-offs and provide practical guidance for implementation (Section~\ref{sec:discussion}).

\section{Related Work}\label{sec:related_work}

Our study connects to three strands of research: (1) algorithmic bias in machine learning, (2) bias in causal inference, and (3) bias mitigation and its implications.

\subsection{Algorithmic Bias in Machine Learning}

Our work is related to, but distinct from, the literature on algorithmic bias and algorithmic fairness. In that literature, bias is typically operationalized as a prediction or decision problem: bias is said to exist if there are systematic disparities in predictive accuracy, error rates, or decision outcomes across groups defined by protected attributes \citep{Chouldechova2020, barocas2023, castelnovo2022clarification, Corbett2017, De2022}. Related work on algorithmic fairness, often motivated by legal or ethical considerations, studies how such disparities should be mitigated, typically by modifying prediction models or downstream decision rules \citep{Kleinberg2018}. Different from this literature, we do not study bias in individual-level CATE predictions or downstream decisions per se, but bias that arises when personalized treatment effect predictions are aggregated to recover group-level causal estimands.

While causal or counterfactual notions are sometimes used to define measures of bias or fairness \citep[see, e.g.,][]{Carey2022, nilforoshan2022causal}, the object being predicted is typically not itself a causal effect, and the causal estimand obtained after mitigation is often left implicit. Our work differs in that we focus on causal inference, with both the individual-level prediction, the group-level estimand, and the bias obtained once aggregating the former to the latter, all relate to treatment effects.

Methodologically, several papers in the algorithmic bias literature frame bias detection as a statistical testing problem, focusing on bias in outcome predictions or classification decisions \citep{Taskesen2021, Diciccio2020, Yik2022}. Although we also employ statistical testing, the bias we study is fundamentally different. Because individual-level treatment effects are unobserved \citep{Holland1986}, bias assessment for CATE models necessarily proceeds through aggregation.\footnote{\SingleSpacedXI\footnotesize Here, we refer to settings in which the CATE model depends on many and/or continuous covariates, as this is precisely the regime that motivates ML-based estimation for personalization. To assess bias in such settings, training a new CATE model is not a viable solution, as the original model is presumed to be biased, and model-free estimates from experiments, such as difference-in-means, identify treatment effects only at more aggregate levels.} In the lens of this literature, our contribution is thus to show that ML models of personalized CATEs will tend to be biased when summarizing the CATES into the causal effects at broader groups of intrinsic interest.

\subsection{Bias in Causal Inference}

Our work also connects to research on bias in causal inference. Here, we distinguish three different issues: bias from identification failures, calibration, and bias from ambiguity in the causal estimand.

First, \citet{Gordon2019, Gordon2023} show that ML estimators for causal inference often fail to recover experimental estimates of average treatment effects, essentially revisiting the LaLonde critique \citep{Lalonde1986} with modern methods \citep{Imbens2024} and marketing data. In these works, the discrepancy reflects an \emph{identification problem}: observational data do not support the assumptions required to recover the target estimand. Our work highlights a different issue. Even when CATE models are correctly identified, specified, and estimated, aggregating their predictions can fail to recover intended subgroup-level treatment effects. This can arise because models are trained to optimize global rather than group-specific objectives, or because regularization shrinks modeled heterogeneity \citep{Chernozhukov2018, Melnychuk2024}.

Our work is also related to methods that summarize or evaluate estimated CATEs, but differs in both aims and approach. \citet{Chernozhukov2018b} aim to summarize the predictive content of CATE for more aggregate groups. They do so by sorting predicted CATES into bins and then using regressions to perform inference on the average CATE per bin. \citet{Leng2024} aim at calibrating CATE across its distribution, by regressing the average CATE within calibration bins onto experimental subgroup effects. In these works, ``groups'' are technical devices endogenously defined by quantiles of the CATE distribution. In contrast, in our work, groups are defined by a manager, policy-maker, or researcher (e.g., markets, segments, or demographic attributes), and the question is whether aggregating individual-level CATE predictions within groups recovers the group-level causal effect, and how to correct for the potential bias. As such, our work is motivated by a practical use of CATE models, rather than by their properties. Moreover, while those work focuses on additive effect measures, our framework also applies to relative treatment effects, such as relative risk and lift factors, commonly used in health sciences and marketing. 

In doing so, our research connects to recent econometric literature highlighting when commonly used estimation practices fail to recover the causal estimands researchers intend to estimate \citep[e.g.,][]{Goodman2021, Goldsmith2024}. Among these works, the issue is that the average treatment effect is not identified by the estimation strategy in question when individual-level effects are heterogeneous. This issue is related to the concept of \emph{collapsibility} from the causal inference literature, which characterizes when population-marginal causal effects (such as the ATE) can be recovered from weighted averages of conditional effects \citep{Huitfeldt2019, Didelez2022, Colnet2023}. We extend this logic from population-level effects to subgroup-level effects, showing when aggregation of individual-level CATEs fails to recover group-average treatment effects (GATEs), and propose methods to detect and mitigate the resulting group bias.

\subsection{Bias Mitigation and Its Implications}

Finally, our work relates to research on bias mitigation and its implications. One line of work in marketing studies bias in ML models for CATE specifically. \citet{Ascarza2022} introduce fairness constraints in the training of causal trees to regulate downstream targeting across groups, while \citet{Huang2023} propose how to calibrate away bias in CATE estimates induced by privacy-preserving noise in the training data. While both papers focus on personalized CATEs, they target different bias objects and are tied to specific learners, constraints, or data-generating processes. By contrast, we focus on errors that arise when \emph{aggregating} CATE into group-level treatment effects, while remaining agnostic to the model.

Another stream of research studies the consequences of constraining how predictive models are used. \citet{Rambachan2020} develop an economic model of how group-level constraints on predictive allocations affect optimal decision rules in hiring, while related work examines applications in business analytics \citep{De2022} and public policy \citep{Corbett2017, Chohlas2023}. A recurring insight in this literature is that correcting group-level disparities in individual-level predictions does not necessarily yield the intended economic or allocative outcomes. For example, \citet{De2022} provide counterpoints to the personalization--fairness ``trade-off'', while \citet{Chohlas2023} show that allocation rules designed to promote equity can be Pareto-dominated by purely efficiency-oriented policies.

Our analysis of targeting implications is conceptually related but differs in focus. Rather than studying fairness or equity, we examine the trade-off that arises when personalized treatment effect predictions are used both for aggregate causal inference and for individualized treatment decision-making. In this setting, correcting aggregated CATE predictions to recover GATEs improves group-level causal inference but can alter targeting decisions, resulting in profit loss relative to using the raw CATE predictions. Our analysis clarifies when such trade-offs arise and how they depend on the chosen bias-mitigation strategy. We distill these insights into practical guidance for how managers and firms can navigate this trade-off, helping to rationalize why coarser targeting policies may be preferred even when more granular personalization is technically feasible \citep{Lemmens2025}.

\section{Motivating Example}
\label{sec:motivating_example}

To build intuition for why group bias can arise in ML models of heterogeneous treatment effects, we consider a motivating example in which a personalized CATE model is trained on experimental data and then evaluated by aggregating predictions to ex post-defined groups. The example shows that systematic group bias can arise even under correct identification and specification, and is not specific to any particular estimation approach. The same logic extends to multiple groups and, as shown later, to relative effect measures.

Suppose we observe data on $N$ individuals indexed by $i=1,\ldots,N$. For each individual, let $X_i$ denote a vector of pre-treatment covariates, $T_i\in\{0,1\}$ a binary treatment assigned uniformly at random, and let potential outcomes be given by
\begin{equation}
\label{eq:y0_y1}
Y_i(0) = \mu_0(X_i) + \varepsilon_{i0}, \qquad
Y_i(1) = \mu_1(X_i) + \varepsilon_{i1},
\end{equation}
where $\mu_t(x)=\mathbb{E}[Y(t)\mid X=x]$ and the errors $\varepsilon_{i0}, \varepsilon_{i1}$ have mean zero and finite variance. The CATE is
\begin{equation}
\tau(x) = \mu_1(x) - \mu_0(x),
\end{equation}
which we assume is a smooth, learnable function of $X$. In this data-generating process, the standard identification assumptions hold: There is no interference between units (stable unit treatment value assumption), and the treatment is randomly assigned, implying no confounding and overlap. As a result, $\tau(x)$ is point-identified, and we can use a model $f$ to obtain a CATE estimate $\widehat{\tau}^f(x)$ from the data $(X_i,T_i,Y_i)_{i=1}^N$.

One way to represent this estimation task is empirical risk minimization, where the objective is to solve
\begin{equation}
\label{eq:erm}
\min_{\tau}
\mathbb{E}\Big[\big(Y-\mu_0(X)-T\cdot\tau(X)\big)^2\Big],
\end{equation}
with finite-sample analogs implemented via meta-learners \citep{Kunzel2019}, causal forests \citep{Wager2018, athey2019generalized}, and related methods. Another representation is estimating equations: modern causal ML estimators solve moment conditions of the form
\begin{equation}
\label{eq:moment}
\mathbb{E}\big[\psi( \{X,T,Y\};\tau,\eta_0)\big]=0,
\end{equation}
for some estimating function $\psi(\cdot)$, possibly after orthogonalization with respect to nuisance parameters $\eta_0$ \citep{Chernozhukov2018}. This encompasses double/debiased ML, doubly robust (DR) learners \citep{Kennedy2023}, R-loss minimization \citep{nie2021quasi}, and other orthogonal learners \citep{Foster2023}. In all cases, the CATE model is trained to optimize an objective over the data, under regularization and hyperparameter tuning to control model complexity.

Now suppose that, after the CATE model has been trained, a binary group variable $G\in\{1,2\}$ is introduced to analyze the data at different segments of managerial, organizational, or research relevance. Essentially, the variable $G$ partitions the population into two groups of unequal size, with $N_1>N_2$. The group indicator may be included among the covariates $X$ or correlated with them, but it does not affect treatment assignment per its randomization. The induced groups may therefore differ arbitrarily in their covariate distributions and outcome distributions, even though treatment remains randomized within each group. Our interest lies in using the trained CATE model to infer GATEs $\tau_g = \mathbb{E}[\tau(X)\mid G=g]$, $g\in\{1,2\}$ by \textit{aggregating} the individual-level CATE predictions via subgroup sample averages, $\widehat{\tau}^f_g = \mathbb{E}_N[\widehat{\tau}_f(X)\mid G=g]$. 

We illustrate the resulting tension using a simple simulation. We set sample sizes to $N_1=150$ and $N_2=100$, generate $T_i \sim \text{Bernoulli}(1/2)$, and draw two continuous covariates from group-specific bivariate normal distributions.\footnote{\SingleSpacedXI\footnotesize The sample sizes are chosen for visual clarity; increasing them narrows the precision of the group bias of CATE in GATE without changing its sign or magnitude.}  The true CATE $\tau(X)$ is constructed as a nonlinear function of the covariates with different group means, and thus generating different GATEs $\tau_g$. Outcomes are generated as $Y_i = Y_i(0) + T_i\tau(X_i)$. We then construct a generic CATE predictor by perturbing the true CATE with group-specific mean shifts and mean-zero noise, chosen so that the predictor is unbiased in the population but systematically biased within groups. We also estimate CATEs using a range of representative and widely used causal ML models, namely, causal forests, metalearners (S-, T-, and X-learners), and the DR-learner, which we implemented using the \texttt{grf}, \texttt{glmnet}, and \texttt{randomForest} packages in \texttt{R} with default hyperparameter values.\footnote{\SingleSpacedXI\footnotesize We implement causal forests using \texttt{grf}, the metalearners with random forests for the outcome regressions, and the DR-learner via 5-fold cross-fitted estimation of the propensity score and outcome regressions using regularized generalized linear models (\texttt{glmnet}), followed by a second-stage regression of the resulting orthogonalized (AIPW) pseudo-outcome on covariates using a random forest.} We include the group indicator $G_i$ when training the models to allow them to capture the  group-wise heterogeneity. All models then satisfy point-identification, and are correctly specified.

Figure~\ref{fig:motivating_example} summarizes the results. In the left panel, the generic CATE predictions are centered around the 45-degree line, indicating no overall bias. However, conditional on group membership, predictions systematically understate treatment effects for the larger group and overstate them for the smaller group. When these predictions are aggregated to the group level, the resulting GATE estimates differ systematically from the true GATEs (center panel). The bias is statistically significant at the 5\% level for both groups, and the sign and magnitude vary. The same qualitative pattern appears for the causal ML estimators (right panel), with the sign and magnitude of the group bias now also varying across models. This demonstrates that even under ideal identification and correct specification, ML models that are unbiased for individual-level CATE can exhibit systematic and uncontrolled bias of group-level effects.\footnote{\SingleSpacedXI\footnotesize An underlying mechanism is that the models are trained to optimize global fit over the full distribution of the data, whether by empirical risk minimization, solving a moment condition, or some orthogonalized version thereof \citep[as in][]{Chernozhukov2018, Foster2023, nie2021quasi}. Under regularization and heterogeneous group sizes, covariate distributions, and signal-to-noise ratios, prediction errors then need not cancel symmetrically within the groups chosen ex~post.}

\begin{figure}
\FIGURE
{\includegraphics[width=\linewidth]{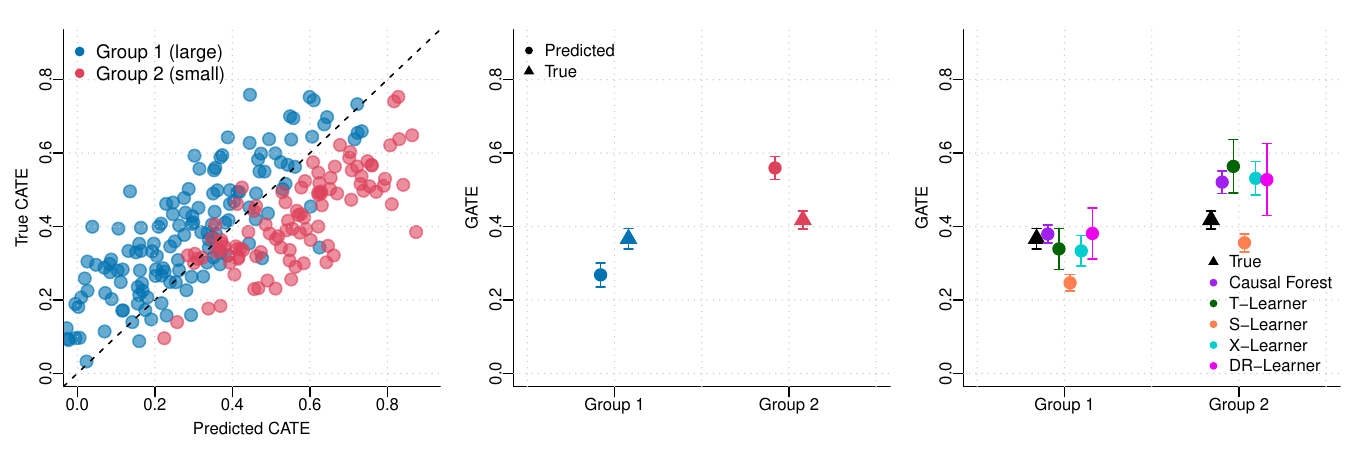}}
{Simulation illustration of group bias in CATE predictions.\label{fig:motivating_example}}
{\emph{(Left)}: A generic CATE model, given by the true CATE plus errors with group=wise heteroskedastic variance, overestimates CATE for the larger group ($N=150$) and underestimates CATE for the smaller group ($N=100$). The dashed 45° line signifies no estimation errors. 
\emph{(Center)}: Averaging the CATE predictions per group yields GATEs that differ in bias.
\emph{(Right)}: Causal ML estimates of CATE are also biased of GATE, despite satisfying all identification assumptions, being correctly specified, and trained on data with with randomized treatment.
Error bars are 95\% confidence intervals based on normal approximation.}
\end{figure}


\section{Framework}\label{sec:framework}

This section provides our framework. We first introduce the relevant estimands and then presenting methods and theory for the detection and mitigation of group bias from experimental data.

\subsection{Treatment Effect Estimands}\label{sec:setup}

We consider individual-level experimental data on pre-treatment covariates $X$, a randomly assigned treatment $T \in \{0,1\}$, and an outcome $Y$ that may be binary or continuous. The random variables $(X,T,Y)$ are drawn i.i.d. from a distribution $P$. The data are partitioned into groups indexed by $G \in \mathcal G$, where $G$ may be contained in $X$. Groups may represent demographic attributes, market segments, or some other one-dimensional partitioning of interest, or they may be defined as a function of observed variables.\footnote{\SingleSpacedXI\footnotesize Groups may also be defined or data-driven manner, provided they are defined independently of the bias detection and mitigation procedures. Examples of the latter include quantiles of predicted outcomes (risk or churn quantiles), or clusters constructed using pre-treatment information} For each group $g \in \mathcal G$, let $P_g$ denote the conditional distribution of $(X,T,Y)$ given $G=g$, which may differ arbitrarily across groups. In a sample of size $N=\sum_{g \in \mathcal G} N_g$ from $P$, group sizes $N_g$ may also vary.

There is a pre-existing ML model $f \colon X \mapsto \widehat \tau^f(X)$, where $\widehat \tau^f(X)$ is its prediction of the (true) CATE, defined on an additive or relative scale as
\begin{equation}\label{eq:cate}
    \tau(X) = \E[Y(1) - Y(0) \mid X]
    \quad\text{and}\quad
    \tau(X) =\frac{\E[Y(1) \mid X]}{\E[Y(0) \mid X]},
\end{equation}
respectively. The model was trained on some prior data using a possibly unknown objective and procedure. No assumptions are imposed on it beyond standard regularity conditions, such as that its predictions have finite variance, and that we can obtain its CATE prediction for each observation.

We are interested in whether the model's CATE predictions recover the GATE, defined as
\begin{equation}\label{eq:gate}
    \tau_g = \E[Y(1) - Y(0) \mid G=g]
    \quad\text{and}\quad
    \tau_g = \frac{\E[Y(1) \mid G=g]}{\E[Y(0) \mid G=g]},
\end{equation}
respectively. Either estimand $\tau_g$ can be estimated from the treated and controls observations per group, also via appropriate regression adjustment estimators (see Appendix \ref{app:gate_estimation} for details). 

Additive effects for continuous outcomes are standard in the statistics and econometrics literature based on the potential outcomes framework. They directly answer questions such as the dollar impact of a promotion on sales or the effect of a medical treatment on blood pressure. In applied research and practice, however, relative effects are also widely used due to their scale-free, percentage-change interpretation and are commonly used when binary outcomes are of primary interest.\footnote{\SingleSpacedXI\footnotesize For relative effects to be well-defined, denominators must be strictly positive. A common ad-hoc fix is to add a small constant to the numerator and denominator, analogous to practices for fitting regression models with zero-valued log-transformed outcomes. However, this can introduce bias. More principled approaches include modeling on appropriate scales (e.g., via log-link GLMs).} For example, in the health sciences, interest often lies in effects on adverse events such as disease onset or death, while in marketing and the tech sector, common binary outcome is conversion. In these settings, $\E[Y(t)] = \Pr(Y(t)=1)$, and so the treatment effect estimands correspond to a causal measure of the risk difference and relative risk, which are standard in the health sciences \citep{Greenland1999, Hernan2006}, or magnitude and relative measures for ``lift'', which are standard in advertising measurement and incrementality testing \citep[see e.g., the measures in Section 4.1. and 5.2 in][]{Gordon2023}. It is worth noting that, although additive and relative effects are expressed on different scales, they are equivalent via simple transformations: relative effects are additive effects scaled by the mean baseline potential outcome or, equivalently, they are additive effects measured on the log scale and then transformed back.\footnote{\SingleSpacedXI\footnotesize Any ratio $\E[Y(1)]/\E[Y(0)]$ can be written as $\E[Y(1)-Y(0)]/\E[Y(0)] + 1$ or equivalently as $\exp(\log\E[Y(1)] - \log\E[Y( 0)])$.}

Our framework applies to both additive and relative effect measures. The distinction matters because aggregation from CATEs to GATEs behaves differently across the scales, as we show in the next section.

\subsection{Group Bias}\label{sec:estimands}

To formalize group bias of CATE predictions relative to GATEs, we first distinguish between two sources: estimation error in the CATE itself and bias induced by aggregation. This distinction is crucial, as only the former represents the underlying CATE error of interest, while the latter must be controlled to ensure that such errors propagate into an error in the GATE rather than an ill-defined or misaligned estimand. Separating the two allows detected group bias to be interpreted as error in the GATE and ensures that statistical inference targets precisely this object.

We begin by contrasting what holds at the population level with what occurs in finite samples for a given group. When treatment effects are defined on the additive scale, the population ATE satisfies
\begin{equation}
    \E[Y(1)-Y(0)] = \E_X\{\E[Y(1)-Y(0)\mid X]\} = \E_X[\tau(X)] .
\end{equation}
This equality explains why most ML models of the ATE first estimate $\tau(X)$ and then averages them over the covariate distribution used for training. In finite samples, estimation errors in CATE tend to cancel symmetrically when averaged over the same distribution, yielding an accurate estimate of the ATE.

However, this aggregation property does \emph{not} generally hold over arbitrary subgroups. The reason is not a failure of identification, but of estimation combined with aggregation: CATE is either estimated via a single model trained to minimize error over the population, or separately for subgroups that vary in size and distribution. Then, estimation errors need not cancel within groups, nor be equal in magnitude across groups. This issue is exacerbated by that causal ML models typically rely on asymptotic guarantees and regularization, which can penalize heterogeneity unevenly when sample sizes or signal-to-noise ratios differ.

When treatment effects are instead measured on a relative scale via ratios, an additional complication is introduced. Because ratios aggregate nonlinearly, the expectation of relative CATEs generally does not point-identify the corresponding group-level estimand. By Jensen’s inequality, we have
\begin{equation}\label{eq:jensen}
\E[\tau(X)] 
= \E\!\left\{\frac{\E[Y(1)\mid X]}{\E[Y(0)\mid X]}\right\}
\ge
\frac{\E\{\E[Y(1)\mid X]\}}{\E\{\E[Y(0)\mid X]\}}
=\tau.
\end{equation}
Hence, the simple average of relative CATEs recovers an estimand that is \emph{not} a GATE but instead (weakly) upward biased relative to it. Because Eq.~\eqref{eq:jensen} is derived with respect to the true CATE, this issue arises even when $\tau(X)$ is known or estimated perfectly.

To address bias from aggregation, and ensure group-level errors reflects the CATE model itself, we draw on the concept of \emph{collapsibility} \citep{Colnet2023}.

\begin{definition}[Collapsibility]\label{def:collapse}
\emph{Let $P\{X,Y(0)\}$ denote the joint distribution of covariates and baseline potential outcomes. A treatment effect measure $\tau$ is \emph{collapsible} if there exists a weight function $W=w(X,P(X,Y(0)))$ such that}
\begin{equation}
\E_P[W\tau(X)] = \tau, \qquad W\ge0,\ \E_P[W]=1.
\end{equation}
\end{definition}

\noindent Collapsibility formalizes when a conditional treatment effect estimand can be aggregated to its marginal estimand. It is a property of the estimand and \textit{not} just of its estimate. Additive effects are collapsible with uniform weights ($W=1$), recovering the familiar linear aggregation result mentioned earlier. For relative effects, the appropriate weights are $W(X)=\E[Y(0)\mid X]/\E[Y(0)]$ \citep{Greenland1999,Colnet2023}, reflecting that the marginal relative effect is a ratio of means rather than a mean of ratios (cf. Eq. \eqref{eq:jensen}), and therefore requires weighting conditional ratios by their baseline potential outcomes.\footnote{\SingleSpacedXI\footnotesize Not all effect estimands are collapsible; e.g., conditional (log) odds ratios are not.}

We can now define a causal estimand of group bias in CATE.

\begin{definition}[Group Bias]\label{def:bias}
\emph{Let $\widehat{\tau}^f(X)$ be a prediction of $\tau(X)$ from a model $f$, and let $W$ be a weight function satisfying Def.~\ref{def:collapse} within group $g\in\mathcal{G}$.
We define the \emph{group bias} of the model as
\begin{equation}\label{eq:bias}
b_g = \E_{P_g}\!\big[W(\widehat{\tau}^{f}(X)-\tau(X))\big]
     .
\end{equation}}
\end{definition}

\noindent The expression mirrors the canonical definition of bias in statistics, but conditional on a group and appropriately weighted so error in CATE propagates to bias with respect to the causal estimand of GATE. 

As shown in Section \ref{sec:motivating_example}, the group bias will generally differ across groups. To assess whether bias is unevenly distributed, one can consider contrasts in the group bias. A natural measure is the \emph{cross-group bias} $b_g - b_{-g}$, where $b_{-g} = \sum_{k \in \mathcal{G} \setminus g} b_{k} \Pr[G=k]$ is the average bias across all other groups. This estimand captures whether the bias induced by a CATE model for group $g$ differs systematically from that of the remaining population, and provides a basis for across-group comparisons.\footnote{\SingleSpacedXI\footnotesize One may also consider pairwise contrasts $b_g - b_h$ for any $h \neq g$ in $\mathcal G$, scalar dispersion measures such as $\Var_G[\widehat B_g]$ for estimates $(\widehat B_g)$, or empirical densities over $(\widehat B_g)$. We later provide results based on the latter two (cf. Appendix \ref{app:distribution} to the empirical results in Section \ref{sec:booking_results} and Appendix \ref{app:simulation} on a simulation study.}


Two examples illustrate why detecting and mitigating the group bias can matter in practice.

\begin{example} 
Suppose a firm fits a CATE model on its customer data and then wants to target promotional offers towards segments that differ in size and demand elasticity. One segment has particularly strong demand elasticity—for example, due to lower average income or higher price sensitivity—but is made up of comparatively few customers. Aggregating the customer-level CATE predictions to segment-level GATEs would then tend to underestimate the GATE for this high-response segment. If the firm has a budget constraint on the number of promotions or thresholds the GATEs to decide which segment to target, then this segment would receive fewer promotions than in optimum, and total targeting profits and customer surplus would be lower than attainable.
\end{example} 

\begin{example} 
Suppose a public health agency wants to promote COVID-19 vaccination for demographic groups according to their COVID risk. To do so, the agency conducts large-scale randomized controlled trial to learn the CATE of individual vaccine response, which can then inform an outreach targeting policy. Consider that African Americans were underrepresented in such trials despite facing above-average relative risk from COVID \citep{Warren2020, KFF2020}. Their estimated risk by the model-implied GATE would then be downward biased towards the population-marginal ATE, potentially causing inefficient targeting by the public health agency and suboptimal health outcomes for this high-risk group.
\end{example} 

\subsection{Detection}\label{sec:detection}

\subsubsection{Estimation via Decomposition of Group Bias.}

We now describe a general approach for estimating and detecting group bias in CATE predictions. The main idea is that, although the CATE $\tau(X)$ is unobservable, the group bias of a CATE model can be decomposed as simply the difference between two group-level aggregates that are directly estimable:
\begin{align}\label{eq:bias_decomp}
b_g
   &= 
   \E_{P_g}\!\left[W\widehat{\tau}^{f}(X)\right]
    - 
    \E_{P_g}\!\left[W\tau(X)\right]
    = \underbrace{\tau^f_g}_{\text{model-implied GATE }} - \underbrace{\tau_g}_{\text{true GATE}}
    ,
\end{align}
where the first equality follows by applying linearity of expectations to Eq. \eqref{eq:bias} and the second by Definition~\ref{def:collapse}. Here, $\tau^f_g = \E_{P_g}[W\widehat{\tau}^f(X)]$ is the potentially-biased GATE parameter implied by the CATE model using the true weights, and $\tau_g$ is the corresponding GATE estimand defined in Eq.~\eqref{eq:gate}. This decomposition holds for any model of CATE used and irrespective of whether effects are defined as differences or ratios. Therefore, a general estimator can be written as
\begin{align}\label{eq:bhat}
    \widehat B_g = \widehat\tau_g^f - \widehat\tau_g,
\end{align}
where $\widehat\tau^{f}_g = \E_{N_g}[\widehat W \widehat{\tau}^{f}(X)]$, and $\widehat W$ are estimates of the weights $W=\E[Y(0)\mid X, G]/\E[Y(0)\mid G]$ if effects are measured on a relative scale (cf. the discussion after Definition \ref{def:bias}) and $\widehat W=1$ if not, and $\widehat\tau_g$ is an estimate of GATE for the same group. Hence, estimating the group bias in CATE does not require the actual CATE: it suffices to compare properly aggregated  predictions to a direct estimate of GATE.

In randomized experiments, the weights required to collapse ratio CATEs are nonparametrically identified and easily estimated. In particular, randomization implies that $\E[Y(0)\mid X, G]/\E[Y(0)\mid G]$ is identified by $\E[Y\mid X, G, T=0]/\E[Y\mid G, T=0]$, which can be estimated by fitting regression models to the non-treated using a group indicator and covariates as controls, and then taking the ratio of fitted values per observation.\footnote{\SingleSpacedXI\footnotesize It is useful to consider when weighting for collapsing ratio CATEs is necessary. Because the weights are identified by $\E[Y\mid X, G, T=0]/\E[Y \mid G, T=0]$, their estimates approach a value of one when outcome heterogeneity with respect to covariates is negligible among the non-treated observations in a group. In that case, the weights may, in principle, be ignored for ratio CATEs.} The estimated GATE $\widehat\tau_g$, in turn, is obtained via the appropriate contrast-in-means (difference or ratio) or regression-adjustment estimator; see Appendix~\ref{app:gate_estimation}.\footnote{\SingleSpacedXI\footnotesize This type of bias assessment of comparing a model-based treatment effect estimate to an experimental one parallels the classical evaluation in \citet{Lalonde1986} and its revisits for ML methods in marketing settings \citep{Gordon2019,Gordon2023}. Here, however, we assess bias in subgroup treatment effects implied by individual-level predictions, i.e., across levels of aggregation in HTEs.}

\subsubsection{Inference and Statistical Test.}

In practice, a point estimate $\widehat B_g$ rarely suffices as evidence of the group bias; instead, we want to perform inference using a statistical test. The following proposition provides the sufficient conditions for this. 

\begin{proposition}\label{thm:conv_in_dist}
Let $b_g = \tau_g^f - \tau_g$ be the group bias and let $\widehat B_g = \widehat\tau_g^f - \widehat\tau_g$, where $\widehat\tau_g^f$ and $\widehat\tau_g$ are sample estimators of $\tau_g^f$ and $\tau_g$, respectively, computed on a sample of effective size $N_g$. If $\widehat\tau_g^f$ and $\widehat\tau_g$ admit a joint $\sqrt{N_g}$-asymptotic normal distribution with finite second moments, then, as $N_g\to\infty$,
\begin{equation}
    \sqrt{N_g}\big(\widehat B_g - b_g\big)
    \overset{d}{\longrightarrow}
    \mathcal N(0,\sigma_g^2).
\end{equation}
\end{proposition}

\noindent \emph{Proof.} See Appendix~\ref{proof:conv_in_dist}. \hfill$\square$

Proposition~\ref{thm:conv_in_dist} establishes valid inference for the group bias $b_g$ under standard regularity conditions.\footnote{\SingleSpacedXI\footnotesize Here and throughout, asymptotic arguments are with respect to the observations used to form the group-level aggregates $\widehat \tau^f_g$ and $\widehat \tau_g$, and we use $N_g$ to denote the effective sample size determining their joint convergence.} If the collapsibility weights are consistently estimated, then $\E_{N_g}[\widehat W\widehat\tau^f(X)]$ consistently estimates its parameter $\tau_g^f$ as the number of observations for the average grows, while the direct estimator $\widehat\tau_g$ is $\sqrt{N_g}$-consistent and asymptotically normal for $\tau_g$. As such, the difference $\widehat\tau_g^f-\widehat\tau_g$ isolates $b_g$ increasingly well with more data. In that sense, $\widehat\tau_g^f$ and $\widehat\tau_g$ are nuisance components: they are not of interest in themselves but must satisfy consistency and regularity conditions to enable inference on $b_g$.

The assumption of joint asymptotic normality is mild. Conditional on the fitted CATE model, both nuisance components are averages (or smooth functions thereof) of random variables with finite second moments, so standard central limit arguments apply. Sample splitting can be used to eliminate covariance between the two components if desired, at the cost of a reduced effective sample size.

Under the null hypothesis $H_0 \colon b_g = 0$ of no group bias, Proposition~\ref{thm:conv_in_dist} implies that $Z_g = (\widehat B_g - b_g)/\sigma_g$ is approximately distributed as a standard normal. Replacing the unknown standard error with an estimate thereof thereby allows us to use a standard Wald test. As such, we reject the null hypothesis against its two-sided alternative and say that we have detected group bias if 
\begin{equation}\label{eq:test} 
   \left|\frac{ \widehat B_g }{\widehat \sigma_g}  \right| \geq z_{1-\alpha/2},
\end{equation} 
where $z_{1-\alpha/2}$ is the $1-\alpha/2$-quantile of a standard normal and $\widehat \sigma_g$ the standard error. 

Closed-form expressions for the standard error exist but are cumbersome to derive, as they depend on the weights and the covariance between the model-implied and experimental GATEs (unless independent samples are used), all of which themselves depend on whether effects are measured as differences or ratios. As a simple and general solution, one can use an appropriate bootstrap scheme. Appendix \ref{app:evaluation} provides an implementation with the usual nonparametric bootstrap.

For testing for cross-group bias, the null hypothesis of interest is $H_0 \colon b_g = b_{-g}$, which we evaluate with estimated difference $\widehat B_g - \widehat B_{-g}$. Here, $\widehat B_{-g}$ is the estimated bias on complement of group $g$, obtained analogously. Valid inference follows by an application of Slutsky's lemma to Proposition \ref{thm:conv_in_dist}, provided the focal group is not asymptotically small relative to the remainder. For either test, family-wise error across multiple groups can be controlled using a Bonferroni correction.

\subsection{Mitigation}\label{sec:mitigation}

\subsubsection{Problem and Objective.}

Our next step is to mitigate the group bias out-of-sample. A natural but naïve strategy is to simply subtract the estimated group bias $\widehat B_g$ from each new prediction, yielding adjusted CATE predictions $\widehat\tau^{f}(X)-\widehat B_g$, analogous to classical bias correction. If $\widehat B_g$ estimated $b_g$ without error, this adjustment would eliminate group bias exactly. 

However, in practice, $\widehat B_g$ is subject to estimation error that depends on the sample size and signal-to-noise ratio for the group. Here, the problem is not the estimation procedure itself, as $\widehat B_g$ is unbiased and consistent per Proposition \ref{thm:conv_in_dist}, but rather its sampling variability---we may get an ``unlucky draw'' of detection data that makes $\widehat B_g$ a poor approximation of $b_g$, even though the estimator has the desired properties. This issue will be more problematic for smaller or noisier groups, which will be subject to greater variance and sensitivity to outlier observations. As a result, the naïve strategy may latch onto noise and amplify the group bias unevenly across groups, leading to \emph{greater} cross-group bias compared to before the mitigation. 

As a solution, we introduce a group-specific shrinkage factor $\gamma_g\in[0,1]$ and instead correct new predictions as
\begin{equation}\label{eq:debiased_cate}
\widehat\tau^{f}(X) - \gamma_g\widehat B_g.
\end{equation}
This formulation nests the range of possible debiasing strategies: $\gamma_g=0$ implies no correction, $\gamma_g=1$ implies the naïve strategy of a full correction, and any value of $\gamma_g$ between zero and one implies an intermediate strategy with a different bias-variance trade-off.

To make the trade-off explicit, consider the residual group bias left after debiasing, $b_g - \gamma_g \widehat{B}_g$. By Proposition~\ref{thm:conv_in_dist} and the standard rules for linear combinations of random variables, this residual bias has an ex~ante expected value of $(1-\gamma_g)b_g$ and variance $\gamma_g^2\sigma^2_{g}$. Less shrinkage (i.e., $\gamma_g$ closer to one, meaning more debiasing) therefore trades off greater bias reduction against variance inflation, while more shrinkage (i.e., $\gamma_g$ closer to zero, implying less debiasing) achieves the opposite. Since both bias and variance contribute to the finite-sample deviation of $\widehat B_g$ from $b_g$, any effective method for choosing $\gamma_g$ must balance this trade-off according to the accuracy and precision of the estimated group bias per group.

We balance this trade-off by framing mitigation as a statistical decision problem---incorporating a loss function and aiming for risk minimization under uncertainty. Let $L(b_g-\gamma_g\widehat B_g)$ be the loss from a debiasing error. We want to choose $\gamma_g$ to minimize the (Bayes) risk of such an error, that is,
\begin{equation}\label{eq:mitigation_setup}
    \gamma_g^{L^*} \in \argmin_{\gamma_g \in [0,1]} \
   \E_{\widehat B_g \sim P_g}
   \left[L\left(b_g - \gamma_g \widehat B_g\right)\right],
\end{equation}
where the expectation is over the sampling distribution of $\widehat B_g$ from the detection stage. For a given loss function $L(\cdot)$, the minimizer $\gamma_g^{L^*}$ determines how strongly to act on estimated group bias, thereby representing a particular debiasing strategy.

\subsubsection{Optimal Debiasing.}\label{sec:mitigation_factors}

We next derive optimal debiasing strategies for mitigating group bias in terms of the shrinkage parameter $\gamma_g$. We consider two natural loss functions: \emph{mean debiasing error} (signed linear loss), and \emph{mean-squared debiasing error} (corresponding to MSE). We obtain three debiasing strategies: an \emph{mean error strategy}, which yields a binary decision of whether to debias, an \emph{MSE$-$ strategy}, which yields a continuous correction that trades off bias reduction against estimation variance, and an \textit{MSE$+$ strategy}, which is a simpler approximation of the latter. Our main result is that the oracle solutions, as well as their feasible estimators, admit closed-form solutions that automatically adapt to the statistical uncertainty in bias detection. We start with the optimal solution under mean-error loss.

\begin{proposition}\label{prop:opt_me} 
The oracle minimizer for the mean debiasing error is
\begin{equation}
    \gamma^{\mathrm{ME}}_g
    =
    \mathds{1}\{b_g \neq 0\}.
\end{equation}
\end{proposition}

\noindent \emph{Proof.} This result is straightforward, and we therefore provide a proof sketch. By unbiasedness of $\widehat B_g$ per Proposition~\ref{thm:conv_in_dist}, the mean debiasing error is
\begin{equation}
    \E[b_g - \gamma_g \widehat B_g]
    =
    (1 - \gamma_g) b_g.
\end{equation}
It follows that one should debias fully ($\gamma_g = 1$) if true bias is present and not debias otherwise ($\gamma_g = 0$). \hfill$\square$

Because $b_g$ is unknown, we implement this rule by substituting the test decision from the detection stage. The feasible estimator of $\gamma^{\mathrm{ME}}_g$, which we call the \emph{mean-error strategy}, is thus
\begin{equation}\label{eq:ME_gamma_est}
    \widehat \gamma^{\mathrm{ME}}_g(\alpha)
    =
    \mathds{1}\!\left(
        \left|
        \frac{ \widehat B_g }{
            \widehat \sigma_g 
        }
        \right|
        \ge z_{1-\alpha/2}
    \right).
\end{equation}
Optimal debiasing under mean-error loss is thus determined by the signal-to-noise ratio of the estimated bias. Smaller or noisier groups tend to exhibit larger estimated bias but also greater estimation variance; the rule accounts for this by debiasing only when evidence is sufficiently strong. In this sense, the significance level $\alpha$ reflects our risk tolerance: smaller values of $\alpha$ (e.g., 0.01) imply that we debias only when evidence is strong (risk-averse), while larger values of $\alpha$ (e.g., 0.10) imply we debias more liberally (risk-tolerant).

We now turn to optimal debiasing under squared loss. The mean-squared debiasing error is
\begin{equation} 
    \E_{\widehat B_g \sim P_g} \big[ (b_g - \gamma_g\widehat B_g)^2 \big] = (1 - \gamma_g)^2 b_g^2 + \gamma_g^2 \sigma^2_g, 
\end{equation}
which recovers the familiar bias-variance trade-off, but in terms of how much to debias as controlled by $\gamma_g$.

\begin{proposition}\label{prop:opt_mse} 
The oracle minimizer for the mean-squared debiasing error is
\begin{align}\label{eq:correction_opt} 
\gamma^{\textrm{\emph{MSE}}}_g = \frac{b^2_g}{\sigma^2_g + b^2_g} . 
\end{align} 
\end{proposition}

\noindent \emph{Proof.} See Appendix~\ref{proof:opt_alpha}. \hfill$\square$

While the mean error strategy yields a binary decision of whether to debias, the MSE minimizer prescribes how much to debias. Holding $b_g^2$ fixed, $\gamma^{\mathrm{MSE}}_g \to 1$ as $\sigma_g^2 \to 0$, and $\gamma^{\mathrm{MSE}}_g \to 0$ as $\sigma_g^2 \to \infty$, meaning that more precise estimates during detection lead to stronger debiasing.\footnote{\SingleSpacedXI\footnotesize Equivalently, a rational decision-maker hedges against estimation noise and debiases less when evidence is imprecise. Because of this, the MSE-minimizer can also be interpreted as a form of classical empirical Bayes shrinkage.}

Note that $\E[\widehat B_g^2] = \sigma_g^2 + b_g^2$ for any $\widehat B_g$ satisfying Proposition~\ref{thm:conv_in_dist}. Therefore, Eq.~\eqref{eq:correction_opt} can be rewritten as
\begin{equation}\label{eq:correction_opt_alt} 
    \gamma^{\textrm{MSE}}_g = \frac{\E[\widehat B^2_g] - \sigma^2_g}{\E[\widehat B^2_g]} 
    = \frac{b^2_g}{\E[\widehat B^2_g]}.
\end{equation} 
This representation suggests two feasible estimators of the MSE-optimal shrinkage, which we can call the \emph{MSE$-$ strategy} and the \emph{MSE$+$ strategy}:
\begin{equation}\label{eq:MSE_corrections_est} 
    \widehat{\gamma}^{\textrm{MSE}-}_g 
    = \frac{\widehat \E[\widehat B^2_g] - \widehat \sigma^2_g}{\widehat \E[\widehat B^2_g]} \quad \text{and} \quad
    \widehat{\gamma}^{\textrm{MSE}+}_g 
    = \frac{\widehat B^2_g}{\widehat \E[\widehat B^2_g]}
    .
\end{equation} 
The empirical moments $\widehat\E[\widehat B_g^2]$ and $\widehat\sigma_g^2$ can be obtained during the detection stage, for example via bootstrap resampling, as part of constructing the test statistic in Eq.~\eqref{eq:test}.  

It is straightforward to see that $\widehat\gamma^{\mathrm{MSE}-}_g$ is unbiased for $\gamma^{\mathrm{MSE}}_g$. Therefore, it will shrink the debiasing optimally given estimated moments. The other estimator, $\widehat\gamma^{\mathrm{MSE}+}_g$, is easier to implement, as it involves one less moment, but will tend to overcorrect. This is because it replaces $b_g^2$ in the numerator by $\widehat B_g^2$. Since $b_g^2 = \E[\widehat B_g^2]-\sigma_g^2$, this substitution inflates the numerator and thus pulls $\widehat\gamma^{\mathrm{MSE}+}_g$ upward relative to $\gamma^{\mathrm{MSE}}_g$. Nonetheless, it remains useful as a simpler approximation.

\subsubsection{Evaluation.}\label{sec:evaluation}

The mitigation procedure minimizes the expected risk of debiasing errors, but does not by itself show how well the correction worked. A natural and tempting approach to evaluate mitigation is to compare the debiased CATE predictions to the GATE estimate $\widehat \tau_g$ obtained during detection. We now explain why this reuse is invalid and, instead, how to obtain valid inference.

Reusing the detection-stage GATE to evaluate mitigation is inappropriate because the same estimate is then used both to select the amount of debiasing and to assess its effectiveness. This induces \emph{post-selection bias}: the estimated residual bias is artificially shrunk toward zero because it effectively assesses in-sample performance, and inference is invalid since the estimate (and its test statistic) is derived from the data used to optimize the mitigation. As a remedy, one can use sample splitting or hold-out data. whereby, one re-estimates the GATE $\tau_g$ on an independent evaluation sample, for example, a held-out split of the original experimental or from a new experiment, and compare this estimate $\widetilde \tau_g$ to the properly collapsed, debiased CATE predictions. Specifically, the estimator for the residual group bias is
\begin{equation}\label{eq:est_bias_diff_post}
    \widehat B^{\widehat \gamma}_g
    = \underbrace{\E_{P_g} \big[ \widehat W \widehat \tau^{f}(X) - \widehat \gamma_g \widehat B_g \big]}_{\text{Collapsed debiased CATE predictions}}
    - \underbrace{\widetilde \tau_g}_{\text{Hold-out GATE estimate}},
\end{equation}
where, just like for the initial detection, the weights $\widehat W$ and GATE estimator $\widetilde \tau_g$ are selected for scale of the treatment effect (i.e., the appropriate contrast-in-means or regression adjustment; see Appendix~\ref{app:gate_estimation}).

We now show that this is an unbiased estimator of the residual group bias. If we ignore $\widehat\gamma_g\widehat B_g$ in Eq. \eqref{eq:est_bias_diff_post}, then the remainder is simply a group bias estimate $\widetilde B_g$ on the new data, and Eq. \eqref{eq:est_bias_diff_post} can be written as
\begin{equation}\label{eq:debias_error_est}
    \widehat B^{\widehat \gamma}_g
    = \widetilde B_g - \gamma_g \widehat B_g.
\end{equation}
Because $\widehat\gamma_g \widehat B_g$ is already chosen, we are interested in the residual bias conditional on it. Hence,
\begin{align}\label{eq:debias_error_true}
    \E\left[ \widehat B^{\widehat \gamma}_g \mid \widehat \gamma_g \widehat B_g \right] 
    &=
    \E
    \left[\widetilde B_g - \gamma_g \widehat B_g \mid \widehat \gamma_g \widehat B_g \right] \\
    &=
    \E
    \left[\widetilde B_g \right] - \widehat \gamma_g \widehat B_g \\
    &= b_g - \widehat \gamma_g \widehat B_g
\end{align}
and so $\widehat B^{\widehat\gamma}_g$ is an unbiased estimator of the residual bias $b_g - \widehat \gamma_g \widehat B_g$ introduced in Section \ref{sec:mitigation}.

For inference, the statistics literature on post-selection inference tells us that, when the data for selection and inference are independent, normal approximations or bootstrap procedures yield valid inference under minimal assumptions
\citep[e.g.,][]{Rinaldo2019,Kuchibhotla2022,Rasines2023}. Accordingly, testing proceeds as prescribed by Proposition~\ref{thm:conv_in_dist} and the initial detection: testing $H_0\colon b_g - \widehat \gamma_g \widehat B_g=0$ evaluates whether mitigation removed bias for group $g$, and testing differences in this residual bias for a given group compared to that of the rest evaluates whether the mitigation equalized the bias across groups. Appendix~\ref{app:evaluation} provides a pseudo-algorithm for evaluating the detection and mitigation on historical data.

\section{Empirical Application: Group Bias in a Large-Scale A/B Test}\label{sec:application}

\subsection{Setting}

We apply our framework to data from a large-scale internal A/B test at \emph{Booking.com}, a leading online travel platform, in which an ML model of CATE was used to infer the lift of a marketing intervention (described later). The company routinely uses ML models of CATE to understand how effects of interventions vary across individuals and subgroups of their customer base. While estimated CATEs are intended to capture meaningful behavioral heterogeneity, they should not exhibit systematic bias across stable user segments (e.g., by geography or demographics).

For this A/B test, an internal team selected country of origin as the group variable of interest for assessing GATEs. This choice reflects both managerial relevance (countries naturally define markets and segments in the travel industry) and internal evidence that treatment effects vary substantially across countries, whether due to underlying conditions or because other determinants of treatment effect heterogeneity correlate with country of origin.\footnote{\SingleSpacedXI\footnotesize Booking.com does not collect sensitive personal attributes such as race or nationality and does not use protected information in its ML systems. Country of origin is not a protected attribute under the General Act on Equal Treatment in the Netherlands, where the company is headquartered.} The CATE model evaluated in this application was an existing, general-purpose ML model trained prior to this analysis. This reflects a common organizational structure in platform companies, where models are developed and productionized by engineering teams, while evaluation, reporting, and decision-making are handled ex post by data science and product teams.

There are several reasons why a user-level model of CATE may fail to recover the GATE defined with respect to country of origin. In digital platform companies, CATE models are typically trained on pooled data, either because this is believed to optimize overall predictive performance or because the customer base in some countries has insufficient sample size.\footnote{\SingleSpacedXI\footnotesize More generally, group-level bias can arise for additional reasons, including limited availability of group attributes at training time due to legal, privacy, or operational constraints. For example, to obtain user-level data on country of origin, an online platform must either ask each user to share it or infer it algorithmically. Such data would then need to be stored and governed in a privacy-preserving and secure way for millions of users, which may not be feasible reputationally, technically, or economically. Importantly, as shown in Section~\ref{sec:motivating_example}, group bias can arise even when group indicators are included in the CATE model and treatment assignment is randomized.} As a result, users from different countries are inevitably unevenly represented in the training data, both in terms of sample size and the information their covariates provide about treatment effects. Even when CATE predictions capture meaningful individual-level heterogeneity, aggregating these predictions to the country level can yield systematically biased GATE estimates that obscure inference about within-market and cross-market effectiveness. 

Detecting such bias is therefore important. An operational advantage of the proposed framework is that it allows the CATE model to remain unchanged and continue operating in production: testing for group bias requires only the model’s CATE predictions and A/B test data, and any bias correction is applied ex post. This makes the approach particularly attractive for large-scale platforms, where models operate in real time and retraining or redeployment is costly and time-consuming.

\subsection{Data}

The A/B test ran between January and March 2020. The treatment was a free benefit to encourage users to complete their hotel bookings. The exact nature of the benefit cannot be disclosed, but similar incentives in the travel industry include free breakfast, late check-out, or room upgrades. The offer was shown to users who met eligibility criteria (described below), were randomly assigned to the treatment arm, and who navigated to hotels included in the campaign.

Incoming user sessions were randomly assigned to treatment (offer displayed) or control (no offer), until both arms had 18.5~million observations. Eligibility required that: (i) the session was on a computer; (ii) the user selected a hotel that was part of the campaign; (iii) their search met a minimum spend threshold; and (iv) the user's planned stay involved at most six people. User sessions that did not meet these criteria were excluded, ensuring that treatment and control groups were comparable on baseline factors.

The unit of observation is a user-session on the desktop website, identified by an anonymized session ID, country-of-origin indicator variable, and eight behavioral, pre-treatment covariates capturing browsing, search, and purchase history. Due to confidentiality, the covariates cannot be disclosed but are known to be predictive of CATE on the platform. Each observation also has the treatment status, a binary booking indicator, and a XGBoost prediction of ratio-CATE. The CATE model was trained on data from identical A/B test that ran one year earlier, using the eight covariates as features. Details on the model and data are provided in \citet{Goldenberg2020}.

The company decided to restrict the data to countries with at least 10{,}000 observations to ensure reliable results. This is motivated by that the inference guarantees for detection are asymptotic (Proposition \ref{thm:conv_in_dist}) and because simulations (Appendix \ref{app:simulation}) show better performance at this sample size. As a result, the empirical findings are conservative and less likely to be driven by estimation noise.

\subsection{Application of Framework}

We use ratio-of-means to estimate GATEs and the estimator in \citet{Goldenberg2020} for collapsing CATEs, which is standard at the company. See Appendix \ref{app:retrospective_estimator} for details on the estimator, and Appendix \ref{app:simulation} for simulation evidence of its performance when used for detection and mitigation. We follow our procedure in Section~\ref{sec:evaluation} to evaluate the mitigation. We run 10-fold cross-validation with random splits into detection and mitigation sets, each with fifty bootstrap resamples per group to obtain the sampling distribution statistics required to implement the debiasing strategies.

\subsection{Empirical Results}
\label{sec:booking_results}

Figure~\ref{fig:bias_scatter} plots the experimentally estimated GATEs against model-predicted GATEs before and after mitigation, providing a direct visualization of group bias and its correction. Before mitigation, the predicted GATEs exhibit a substantially narrower range than the experimental GATEs, as seen by comparing the $x$-axis to the $y$-axis in Figure \ref{fig:bias_scatter_before}. This reflects that the CATE model was trained on all data across groups to minimize empirical loss under regularization, thereby shrinking heterogeneity in treatment effects toward the global mean, particularly for smaller or less predictable groups. When CATE predictions are aggregated to the group level, this shrinkage carries over to the predicted GATEs, resulting in systematic bias relative to the experimental estimates.

The corresponding scatters after mitigation are shown in Figure~\ref{fig:bias_scatter_debiased}. After correction, the predicted GATEs span a range comparable to that of the experimental GATEs, reflecting the purpose of debiasing. What matters, however, is not only recovering the range, but how closely the corrected predicted and experimental GATE align, as indicated by their proximity to the 45-degree line of no error. The mitigation strategies differ markedly in how they achieve this. Naïve debiasing produces the widest spread around the diagonal, reflecting amplification of noise when bias estimates are imprecise. The risk-minimizing strategies account for estimation variance and can therefore address this. Among them, the MSE$-$ strategy produces the tightest overall alignment with the experimental GATEs, consistent with its optimality under mean-squared loss over the debiasing error. The mean error strategy applies full correction only when evidence is sufficiently strong, leading to more conservative adjustments for smaller groups, resulting in that the predicted GATEs still have a somewhat narrower range that the experimental estimates. The heuristic MSE$+$ strategy exhibits a somewhat intermediate behavior: it has a tighter spread than naïve but more outlier points than the mean error and MSE$-$ strategy.

\begin{figure}[htbp]
  \FIGURE
  {
  \centering
  \subfigure[Without mitigation \label{fig:bias_scatter_before}]{
    \includegraphics[height=2.5in]{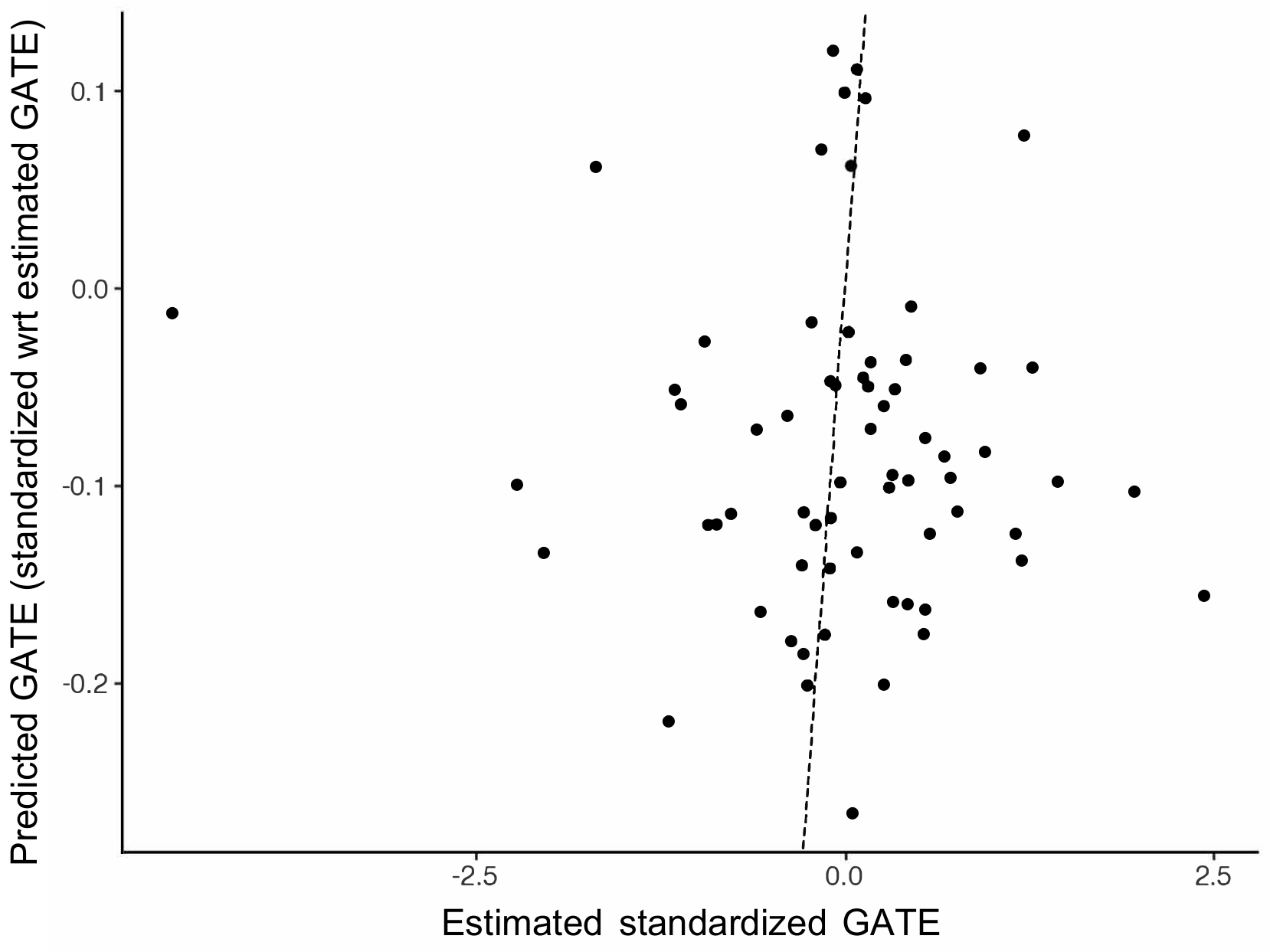}
  }\hfill
  \subfigure[With mitigation \label{fig:bias_scatter_debiased}]{
    \includegraphics[height=2.5in]{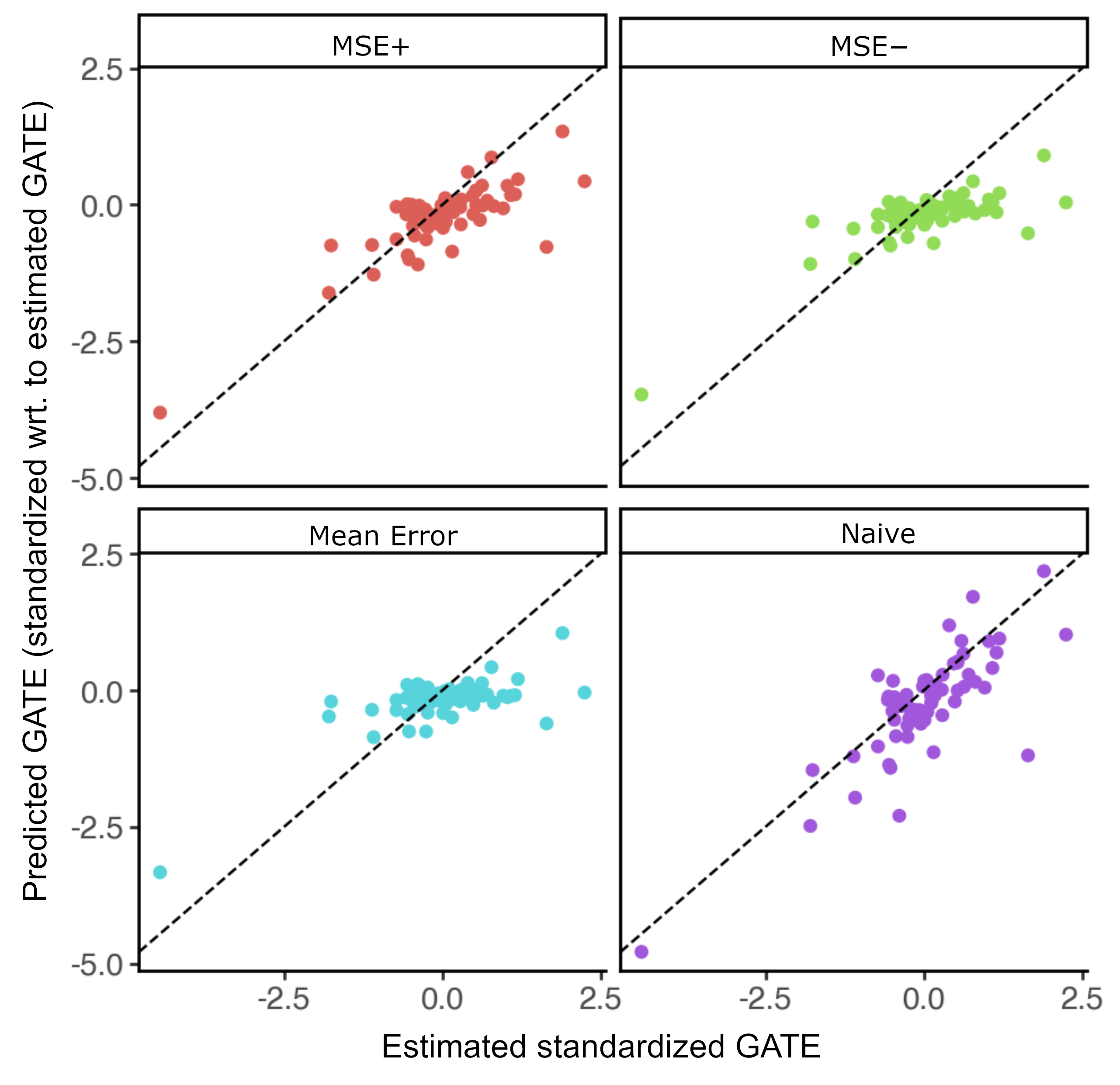}
  }
  }
  {
  Experimentally estimated GATEs versus model-predicted GATEs in the hold-out data
  \label{fig:bias_scatter}  
  }
  {
  Each dot represents a country of origin. Estimates are standardized by the mean and standard deviation of the experimental estimated GATEs to obtain an interpretable scale while preserving confidentiality. The dashed line indicates perfect alignment.}
\end{figure}

We finally examine how group bias relates to group sample size. Intuitively, groups that make up a larger share of the data are expected to exhibit smaller bias, since they contribute more to the model’s training objective (assuming the user base is somewhat stable) and therefore tend to be predicted more accurately. Figure~3(a) supports this pattern. Most countries account for only about 1\% of the total sample, while just five countries (fewer than one tenth of all groups) make up approximately 5, 8, 10, 22, and 33\%---and over 75\% collectively. Among the least represented groups, the estimated biases range from about $-2$ to over $+4$ standard deviations. For the five largest groups, by contrast, the estimates do not exceed $\pm 1$.


Figure~3(b) shows that all debiasing strategies improve this pattern, though in different ways. All methods substantially reduce the most extreme upward deviation, from more than four standard deviations away from zero to just above one. However, the naïve strategy also increases the most extreme downward deviation by nearly a full standard deviation, a behavior not observed for the risk-minimizing strategies. Moreover, the naïve strategy disproportionately corrects the single largest group: the right-most purple point in Fig.~3(b) lies exactly at zero. This can be explained by that naïve debiasing optimizes for bias reduction in expectation. The largest group will always have the most precise bias estimate, and so the correction is most effective for that one group. The same logic also explains why naïve debiasing increases bias for the smaller groups: it ignores heterogeneity in estimation variance and instead applies an overconfident, full correction uniformly. The risk-minimizing strategies do not have this weakness; a larger fraction of the smaller groups’ bias estimates fall close to zero under their debiasing.

\begin{figure}[htbp]
  \FIGURE
  {
  \centering
  \subfigure[Without mitigation\label{fig:bias_vs_samplesize_no_debiasing}]{
    \includegraphics[height=2.5in]{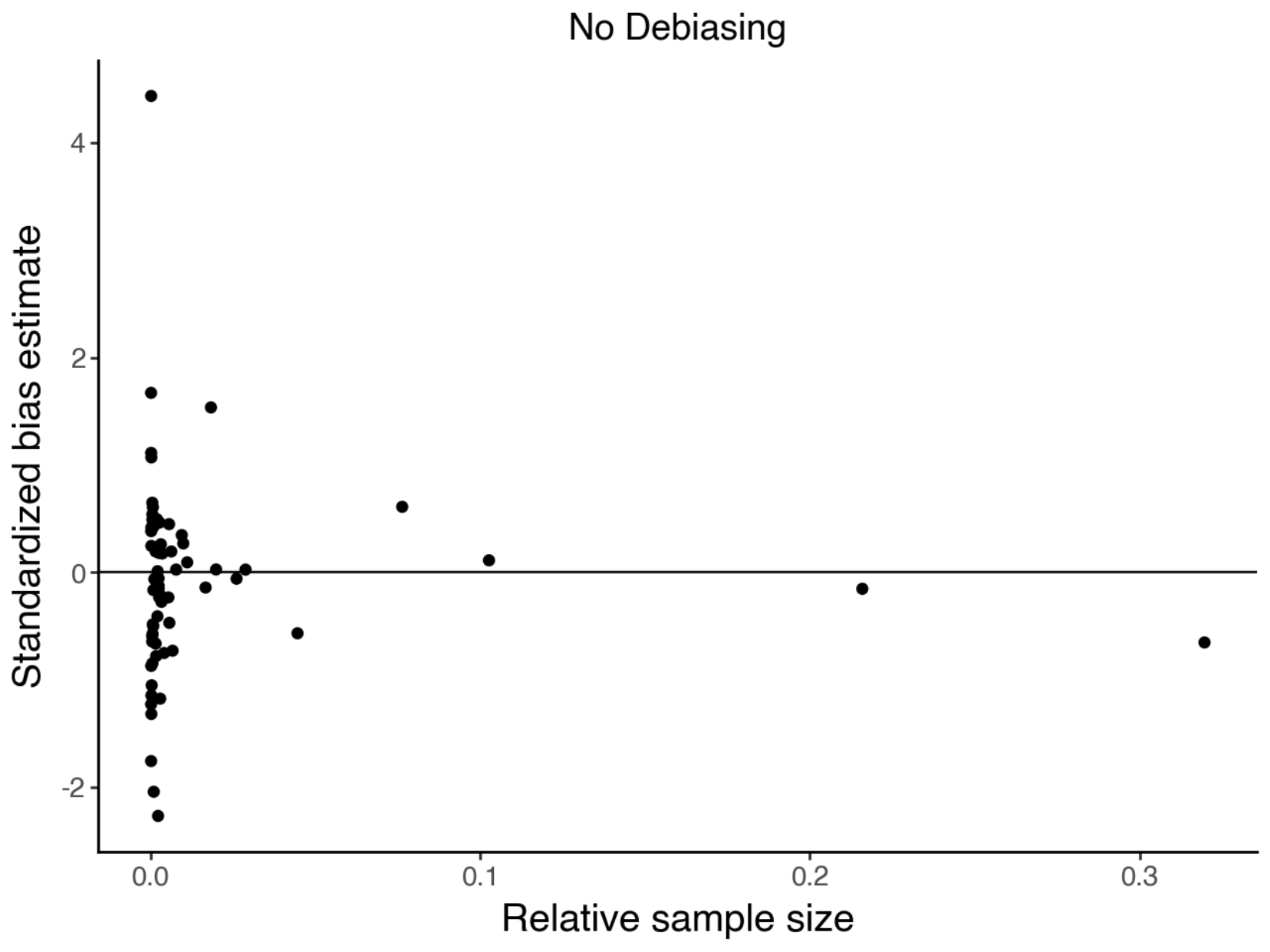}
  }\hfill
  \subfigure[With mitigation\label{fig:bias_vs_samplesize_debiased}]{
    \includegraphics[height=2.5in]{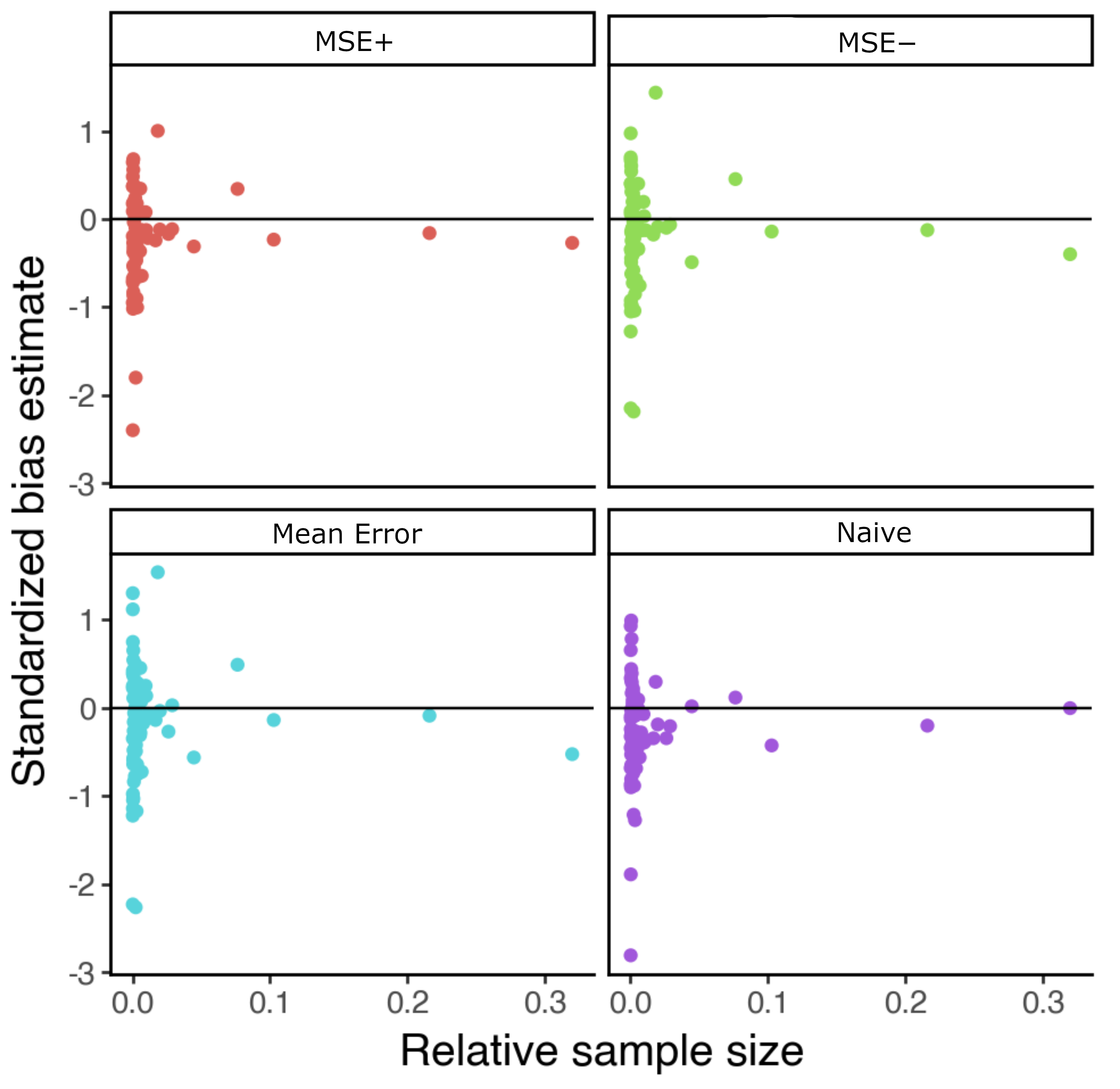}
  }
  }
  {
  Group bias estimates $\widehat{B}_g$ against relative group sample sizes $N_g / N$ in the hold-out data.
  \label{fig:bias_vs_samplesize}
  }
  {
  Each dot represents a country of origin. Estimates are standardized by the mean and standard deviation of the experimental estimated GATEs to obtain an interpretable scale while preserving confidentiality.
  }
\end{figure}

\section{Implications for Targeting}\label{sec:targeting}

We now analyze the implications for targeting. We focus on the canonical application of using CATE models for profit-maximizing personalized targeting, where treatment carries a cost and positive binary outcomes yield revenue. Our analysis focuses on the empirically relevant case in which organizations correct personalized CATE models to recover GATEs using experimental evidence, for instance, for purposes of insights, reporting, or model auditing, while continuing to target based on the personalized predictions to exploit all heterogeneity. For this setting, we show that group debiasing induces trade-offs that depend on the accuracy of the CATE model, the uncertainty in estimated group bias, and the profit margins underlying the targeting problem. We discuss the implications of our results and how to navigate them in Section~\ref{sec:discussion}.

\subsection{Profit-Maximizing Personalized Targeting} 

Consider a firm that assigns a binary treatment $T \in \{0,1\}$ to users based on covariates $X \in \mathcal X$. Let $Y(T) \in \{0,1\}$ denote the potential conversion outcome under treatment $T$, and let $\Pr[Y(T)=1 \mid X] = \E[Y(T) \mid X]$ denote the corresponding conversion probability. A conversion yields revenue $R>0$, while assigning $T=1$ incurs cost $C \in (0,R)$ regardless of the outcome. The firm seeks to learn a targeting policy $\pi \colon \mathcal X \to \{0,1\}$ to maximize expected profit. The expected profit from assigning $T$ is
\begin{equation}\label{eq:profit_function}
    \psi(T)
    = \E[Y(1)\mid X](R-C)T
    + \E[Y(0)\mid X]R(1-T).
\end{equation}
Assigning treatment is therefore optimal whenever
\begin{equation}
    \E[Y(1)\mid X](R-C) > \E[Y(0)\mid X]R
    \quad\Longleftrightarrow\quad
    \tau(X) = \frac{\E[Y(1)\mid X]}{\E[Y(0)\mid X]} > M = \frac{R}{R-C}.
\end{equation}
Hence, the \emph{optimal policy} $\pi^* \in \argmax_{\pi\in\Pi} \mathbb{E}[\psi(\pi)]$ assigns treatment whenever the relative CATE exceeds the inverse profit-margin $M$, i.e., the incremental lift required for treatment to break even.

The firm can approximate the optimal policy by replacing $\tau(X)$ with an ML estimate $\widehat\tau^{f}(X)$. By the plug-in principle, this yields the empirical profit-maximizing policy
\begin{equation}\label{eq:optimal_policy}
\widehat\pi(X) = \mathds{1}\{\widehat\tau^{f}(X) > M\}.
\end{equation}
Suppose the firm applies one of the mitigation procedures. The profit-maximizing policy then becomes
\begin{equation}\label{eq:d_fair}
    \widehat\pi_g(X;\gamma_g\widehat B_g)
=   \mathds{1} \big\{\widehat\tau^{f}(X)-\gamma_g\widehat B_g>M \big\}.
\end{equation}
\noindent The debiased policy continues to maximize expected profit, but using CATE estimates that are corrected to exhibit no statistically detectable group bias. Hence, Eq.~\eqref{eq:d_fair} can be interpreted as solving the program
\begin{equation}
\label{eq:max_pi}
    \max_{\pi \in \Pi} 
    \frac{1}{N}\sum^{N}_{j=1} \psi(\pi(X_j))
    \quad \text{subject to} \quad
    \sum_{g\in \mathcal G} 
    \mathds 1 \left(
        \Bigg |
        \frac{ \widehat B^{\gamma}_g - b^{\gamma}_g}{\sqrt{\widehat{\textrm{Var}}(\widehat B^{\gamma}_g)}} 
        \Bigg | 
        > z_{1-\alpha/2}
    \right) = 0,
\end{equation}
where the constraint restricts the policy space to those for which all group-level bias tests fail to reject at the pre-specified significance level $\alpha$.\footnote{\SingleSpacedXI\footnotesize Eq.~\eqref{eq:max_pi} admits an alternative interpretation as a stochastic program with chance constraints. We focus on the statistical-testing interpretation, which directly corresponds to the framework's implementation.} 

A key implication is that debiasing induces \emph{group-specific thresholds} for targeting. This follows from that $\widehat\tau^{f}(X) - \gamma_g \widehat B_g > M$ is equivalent to $\widehat\tau^{f}(X) > M + \gamma_g \widehat B_g$, and thus debiasing shifts the threshold by an amount that varies across groups. As shown in Section~\ref{sec:mitigation_factors}, the magnitude of this shift depends on the size and precision of the estimated group bias, so that groups with larger or more precisely estimated bias receive larger adjustments. By contrast, the unconstrained policy based on unadjusted CATE predictions applies a common threshold to all individuals, reflecting the principle that treatment allocations are based on a common standard for all individuals \citep[see, e.g., the discussions in][]{Corbett2017}.

\subsection{Expected Profit Loss and Probability of Altered Targeting}\label{sec:profit_loss}

We now consider the consequences of debiasing on targeting decisions and profits. Let
\begin{equation}\label{eq:disagreement}
    \varphi(X) \coloneqq \mathds{1}\!\left\{
        \widehat{\pi}(X) \neq \widehat{\pi}_g(X; \gamma_g \widehat{B}_g)
    \right\},
\end{equation}
which equals one when mitigation alters targeting. The expected \emph{per-group} profit change from debiasing is
\begin{align}
\label{eq:profit_loss}
    \Delta \psi_g
    &\coloneqq \E_{P_g}\left[\psi(\widehat{\pi}_g;\gamma_g\widehat{B}_g)
           - \psi(\widehat{\pi})\right] \\
    &= \E_{P_g} \left[ \varphi(X)
       \Big(\psi\{\widehat{\pi}_g(X;\gamma_g\widehat{B}_g)\}
           - \psi\{\widehat{\pi}(X)\} \Big)\right] \\
    \label{eq:diff_targeting}
    &= \E_{P_g} \Big[
    \big(\underbrace{\widehat{\pi}_g(X;\gamma_g\widehat{B}_g)
            - \widehat{\pi}(X)}_{\text{Equals $\pm 1$ when targeting differs}}\big)
            \big(\underbrace{\psi(1)-\psi(0)}_{\text{Profit lift from treatment}}\big)\Big]
            .
\end{align}
Eq.'s~\eqref{eq:profit_loss}--\eqref{eq:diff_targeting} show that profit is impacted only when debiasing alters targeting, weighted by profit lift from treatment. It is therefore useful to characterize the probability of this event, as done in the following.

\begin{proposition}\label{thm:differ}
Let $\widehat{\gamma}_g\widehat{B}_g$ denote a bias correction, and assume that $\widehat \tau^{f}(X) = \tau(X) + b_g(X) + \varepsilon(X)$, where $b_g(X)=\E[\widehat{\tau}^{f}(X)-\tau(X)\mid X]$, $\varepsilon(X)\mid X\sim\mathcal{N}\!\big(0,\sigma^2(X)\big)$, and $\sigma^2(X)=\Var[\widehat{\tau}^{f}(X)-\tau(X)\mid X]$. Let $\Phi$ denote the standard normal cumulative distribution function. Define the policies $\widehat{\pi}(X)$ and $\widehat{\pi}(X;\widehat{\gamma}_g\widehat{B}_g)$ as in
Eq.~\eqref{eq:optimal_policy} and Eq.~\eqref{eq:d_fair}. Then, the probability that their targeting decisions differ is
\begin{equation}\label{eq:disagree}
\Pr\big(\widehat{\pi}(X;\widehat{\gamma}_g\widehat{B}_g)\neq \widehat{\pi}(X)\mid X\big)
=
\Phi\left(\frac{M+|\widehat{\gamma}_g\widehat{B}_g|-[\tau(X)+b_g(X)]}{\sigma(X)}\right)
-
\Phi\!\left(\frac{M-|\widehat{\gamma}_g\widehat{B}_g|-[\tau(X)+b_g(X)]}{\sigma(X)}\right).
\end{equation}
\end{proposition}

\noindent \emph{Proof.} See Appendix~\ref{app:proof_differ}. \hfill$\square$

Proposition~\ref{thm:differ} characterizes the probability that debiasing alters a targeting decision---that is, the probability that the constrained and unconstrained policies disagree for a given user.\footnote{\SingleSpacedXI\footnotesize In the proposition, $b_g(X)$ represents systematic estimation error at covariate value $X$, $\varepsilon(X)$ captures random estimation noise, and $\sigma^2(X)$ denotes its conditional variance. These imply that the CATE prediction error is approximately normal, which holds asymptotically for many modern CATE learners under standard regularity conditions. Here, normality is invoked for analytical convenience, and the qualitative result does not rely on it: any continuous error distribution with no point mass at the decision threshold yields an analogous expression, that differs only CDF being evaluated.} The marginal probability for a given group, or for the population as a whole, is obtained by integrating this conditional probability over the corresponding covariate distribution.\footnote{\SingleSpacedXI\footnotesize The marginal probability for group $g$ is given by $\Pr(\widehat{\pi}(X;\widehat{\gamma}_g\widehat{B}_g)\neq \widehat{\pi}(X)\mid G=g)
= \int \Pr(\widehat{\pi}(X;\widehat{\gamma}_g\widehat{B}_g)\neq \widehat{\pi}(X)\mid X)\, \diff P_g(X)$, whereas the for the population, it is given by $\Pr(\widehat{\pi}(X;\widehat{\gamma}_g\widehat{B}_g)\neq \widehat{\pi}(X))
= \int \Pr(\widehat{\pi}(X;\widehat{\gamma}_g\widehat{B}_g)\neq \widehat{\pi}(X)\mid X)\, \diff P(X)$.} Empirically, these marginal probabilities correspond to the share of individuals whose treatment assignment would change after debiasing.

Proposition~\ref{thm:differ} shows that four factors determine whether debiasing changes targeting decisions:
(a) the magnitude of the correction $|\widehat{\gamma}_g\widehat{B}_g|$,
(b) the proximity of the predicted CATE $\widehat{\tau}^{f}(X)$ to the profit-lift threshold $M$,
(c) the level of the threshold $M$ itself, and
(d) the CATE residual variance $\sigma^2(X)$.

For (a), the correction $\widehat{\gamma}_g \widehat{B}_g$ shifts the effective decision boundary inside the cumulative distribution function (CDF) from $M$ to $M + \widehat{\gamma}_g \widehat{B}_g$. A larger and more precisely estimated group bias $\widehat B_g$, which implies that $\widehat \gamma_g$ approaches 1 (cf. Section \ref{sec:mitigation_factors}), therefore increases the probability that a CATE prediction crosses the threshold; the opposite holds for smaller or noisier estimates. Hence, more precise detection (and therefore larger bias corrections) raise the chance of altered targeting with profit losses. For (b), units whose predicted CATEs lie close to the boundary $M$ under the unconstrained policy are most likely to be targeted differently under the constrained one, since even small corrections $\widehat{\gamma}_g \widehat{B}_g$ can change whether $\widehat{\tau}^{f}(X) - \gamma_g \widehat{B}_g > M$ holds. For (c), the level of $M$ determines how many units lie near the treatment threshold. Lower thresholds (e.g., when conversion profits are small or treatment costs are high) concentrate more units near $M$, making the policy globally more sensitive to debiasing. Environments with low break-even lifts therefore face inherently higher risk of profit loss from correcting group bias. Finally, for (d), the residual variance $\sigma^2(X)$ appears in the denominator of the CDF. Greater variance thereby flattens the normal CDF and reduces the chance that a CATE estimate crosses the decision boundary, whereas lower variance sharpens the distribution and makes targeting decisions more sensitive to small shifts induced by debiasing.

\subsection{Illustration on the Criteo Uplift Prediction Dataset}
\label{sec:criteo_illustration}

\subsubsection{Data and Methods.}\label{sec:citeo_data_methods}

We illustrate the targeting implications using the \emph{Criteo Uplift Prediction Dataset}, also used in related work \citep{Leng2024}.\footnote{\SingleSpacedXI\footnotesize A description of the dataset and a link to download it is available at \url{https://ailab.criteo.com/criteo-uplift-prediction-dataset/}.}
The dataset is a large-scale benchmark for HTE estimation and policy learning, constructed by combining data from multiple advertising campaigns in which users were randomly assigned to receive or not receive display ads \citep{Diemert2018}.

Each observation corresponds to a user--advertisement impression pair with twelve anonymized pre-treatment covariates $X$, a binary treatment $T$ indicating whether the ad was shown to a user, and a binary outcome $Y$ indicating whether a user converted.\footnote{\SingleSpacedXI\footnotesize All covariates are anonymized and numerically projected to preserve privacy while retaining predictive signal. Randomization checks confirm treatment assignment is independent of covariates. 84.6\% of impressions are assigned to treatment, reflecting industry practice of maintaining a small control population to minimize opportunity costs from withholding ad impressions if they have an effect. See Section~3 of \citet{Diemert2018} for details on the data.} The dataset contains 14 million observations.

We randomly split the data into three independent subsets: (i) 40\% as a \emph{training split} for fitting the CATE model, (ii) 30\% as a \emph{detection split} for estimating and testing group bias, and (iii) 30\% as a \emph{targeting split} for estimating policy outcomes and profits via off-policy evaluation. Following Eq. \eqref{eq:optimal_policy}, we fit a relative-scale CATE model as a T-learner instantiated with XGBoost, which separately estimates $\E[Y(1)\mid X]$ and $\E[Y(0)\mid X]$ on the treated and control units in the training split. On the detection and targeting splits, we obtain predicted relative CATEs by evaluating the two outcome models per value of $X$ and taking the ratio.

Because the dataset lacks predefined groups $\mathcal G$, treatment costs $C$, and revenues $R$, we construct these to emulate a realistic empirical setting. We define five ``baseline-conversion'' groups by fitting a regularized logistic regression model on the control units to predict $\E[Y \mid X]$ and taking quintiles of its fitted values, thereby segmenting users by ex-ante purchase propensity. We set conversion revenue to $R=1$ and base treatment cost to $C=0.005$, implying a profit-lift threshold $M=R/(R-C)=1.005$, or conversely, a profit margin of $0.995$, consistent with small per-impression costs in digital advertising. To examine sensitivity to less favorable unit economics, we increase $C$ up to $0.75$, yielding $M \in [1.005, 4]$.

On the detection split, we estimate experimental GATEs via ratio-of-means and compare them to the properly collapsed predicted GATEs from the CATE model using the procedure in Section~\ref{sec:detection}. This yields estimated group biases $\widehat B_g$ and corresponding mitigation factors $\widehat\gamma_g$, which we use to debias CATE predictions per Eq.~\eqref{eq:d_fair}. We evaluate four mitigation strategies: the naïve correction, the mean error strategy, the MSE$+$ strategy, and the MSE$-$ strategy (see Section~\ref{sec:mitigation_factors}).

To assess the impact on targeting decisions, we compute the share of observations whose targeting changes after applying a given mitigation strategy. Specifically, for each observation $i=1,\ldots,N_g$, group $g$, and mitigation strategy, we evaluate the disagreement indicator $\varphi$ in Eq.~\eqref{eq:disagreement} and take sample averages either within groups or over the full population. These averages correspond to marginal analogs of the conditional probability characterized in Proposition~\ref{thm:differ}. To study how this share (probability) varies with the determinants in that proposition, we additionally compute these averages across ranges of each determinant.

We then counterfactually evaluate the profit impact using the Horvitz--Thompson estimator \citep{horvitz1952}, also known as inverse-propensity weighting (IPW). In randomized experiments, IPW is an unbiased and consistent estimator of counterfactual outcomes and thus allows us to quantify the expected profit effect of debiasing. In this setting, the IPW estimator for an arbitrary policy $\widehat{\pi}_0$ and group $g$ is
\begin{equation}\label{eq:ipw_est} 
    \widehat \psi_g(\widehat \pi_0)
    =
    \frac{1}{N_g} \sum_{j=1}^{N_g}
    \omega_j(\widehat{\pi}_0)\,
    \big(Y_j R - \widehat \pi_0(X_j) C\big),
    \quad\text{with}\quad
    \omega_j(\widehat{\pi}_0)
    =
    \frac{T_j\,\widehat \pi_0(X_j)}{p}
    +
    \frac{(1-T_j)\,[1-\widehat \pi_0(X_j)]}{1-p}
    ,
\end{equation}
where the treatment propensity $p \defeq \Pr(T=1)$ is known and constant by the uniform random assignment, and is computed as the empirical treatment rate in the targeting split.\footnote{\SingleSpacedXI\footnotesize If treatment assignment is random conditional on covariates, $p$ may be replaced by an estimate of the conditional propensity score, for example from logistic regression.} By standard results from semiparametric theory, a heteroskedasticity-robust variance estimator is
\begin{equation}\label{eq:ipw_var}
\widehat \Var[\widehat \psi_g(\widehat \pi_0)]
=
N^{-2}_g
\sum_{j=1}^{N_g}
\left(
\omega_j(\widehat \pi_0)\,
\big(Y_j R - C\,\widehat \pi_0(X_j)\big)
-
\widehat \psi_g(\widehat \pi_0)
\right)^2
    .
\end{equation}
Taking the difference in $\widehat \psi_g$ between the unconstrained policy $\widehat{\pi}(X)$ and a debiased policy $\widehat{\pi}_g(X;\widehat\gamma_g\widehat B_g)$ yields the estimated profit change $\Delta \psi_g$ in Eq. \eqref{eq:profit_loss}. We aggregate these per-group point and variance estimates to the population using weights $N_g/N$, via standard rules for linear combinations of random variables. We compute these profit differences over the same ranges of the determinants as for our analysis on targeting decisions, thereby assessing the corresponding implications for profits. The asymptotic normality of the IPW estimator allows us to construct confidence intervals in the usual manner.

\subsubsection{Empirical results.}\label{sec:criteo_results}

Figure~\ref{fig:margin_policy_criteo} shows the magnitude of the bias corrections per group and their economic consequences. The naïve strategy yields the largest corrections and profit losses. Both the MSE$+$ and MSE$-$ strategies yield similar corrections yet only one of them produces a statistically significant profit loss, which can only be explained by that it assigns more units to suboptimal treatment status than the other strategy. The mean-error strategy concentrates its corrections almost entirely on the highest baseline-conversion group (group 5) and likewise does not incur a statistically significant profit loss.

\begin{figure}[htbp]
\FIGURE
{
  \centering
  {
    \includegraphics[width=0.47\textwidth]{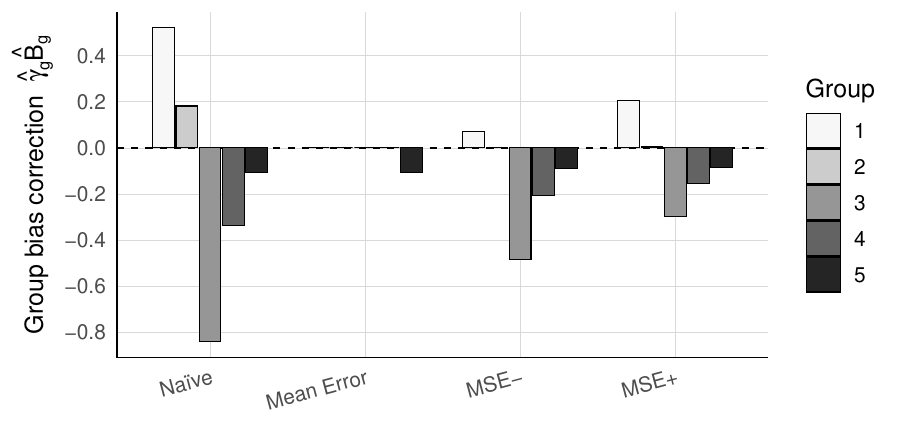}
  }\hfill
  {
    \includegraphics[width=0.47\textwidth]{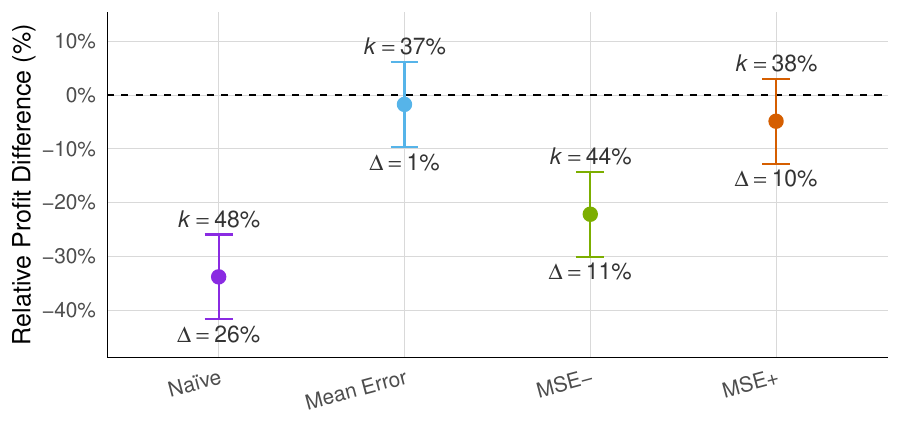}
  }
  }
  {  
  Bias Corrections and Profit Differences
  \label{fig:margin_policy_criteo}
  }
  {
    Left panel: bias corrections per group, $|\widehat{\gamma}_g \widehat{B}_g|$, across the mitigation strategies. 
    Right panel: Corresponding profit differences from debiasing for each mitigation strategy, reported as percentage change relative to the original policy (i.e., targeting based on the unadjusted CATE estimates), with 95\% confidence intervals. The annotated $k$ is the share of units treated, and $\Delta$ denotes the share of units whose targeting decisions differ from the original policy. 
    The relative change in profit is computed with the IPW estimator in Eq.~\eqref{eq:ipw_est} with variance from Eq.~\eqref{eq:ipw_var}, applied per group for the original and debiased policy, then aggregated up to the population-level using sample-share weights $N_g/N$, and finally taking relative differences. Error bars are 95\% confidence intervals. 
  }
\end{figure}

We next examine implications for targeting decisions. Figure~\ref{fig:targeting_determinants_probability} plots the share of users whose targeting assignment changes after debiasing as a function of each determinant entering Proposition~\ref{thm:differ}, following the procedure described in Section~\ref{sec:citeo_data_methods}. We observe that changes in targeting increase proportionally with the magnitude of the bias corrections. They are most likely for users whose predicted CATEs lie close to the profit-lift threshold, with the naïve strategy additionally altering decisions at larger distances. Increasing the profit-lift threshold (via higher treatment costs) reduces the share of targeting decisions that change substantially, consistent with that targeting decisions in environments with low break-even-thresholds are more sensitive to debiasing. Finally, greater noise in CATE predictions quickly attenuates the impact of debiasing. Overall, the empirical patterns align with the implications of determinants in Proposition~\ref{thm:differ}.\footnote{\SingleSpacedXI\footnotesize We also compute 95\% Wilson score confidence intervals, but they are not visually distinguishable from the point estimates due to the large sample size. For reference, each point in panel (a) is based on approximately 840{,}000 observations, each point along the curves in panel (b) on about 210{,}000 observations, and each point along the curves in panels (c) and (d) on the full targeting split of 4.2 million observations.} 


\begin{figure}[htbp]
\FIGURE
{%
  \centering
  \begin{tabular}{cc}
    \includegraphics[width=0.48\textwidth]{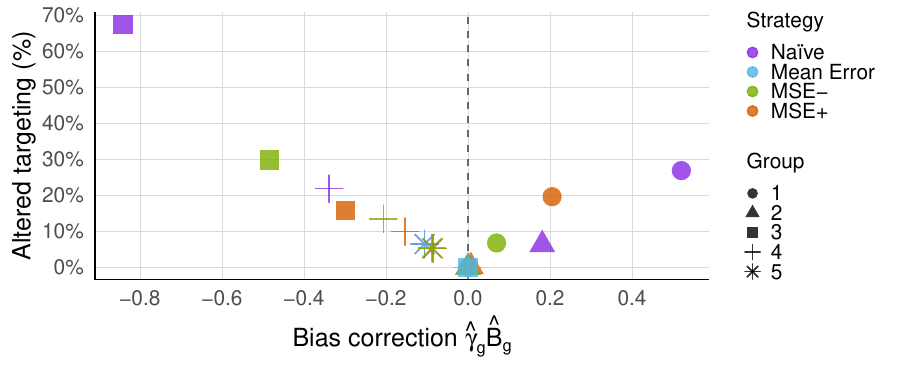}
    &
    \includegraphics[width=0.48\textwidth]{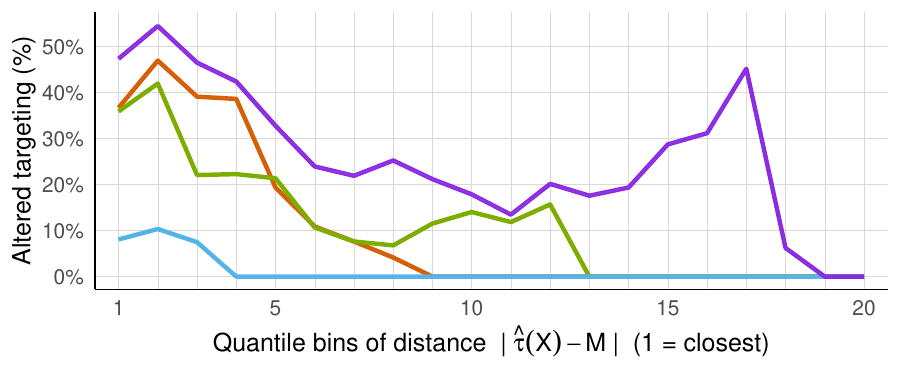}
    \\
    \includegraphics[width=0.48\textwidth]{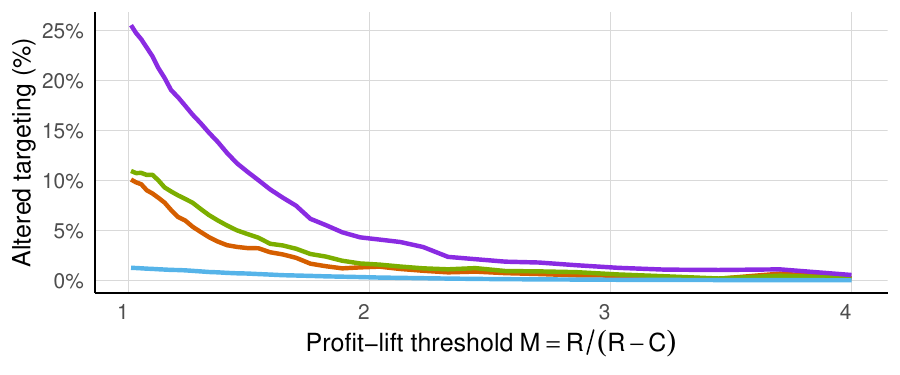}
    &
    \includegraphics[width=0.48\textwidth]{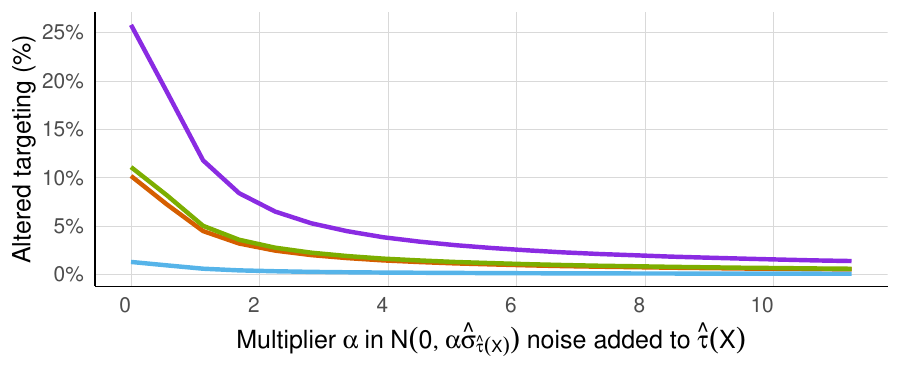}
  \end{tabular}
}
{%
  Share of targeting decisions changed by debiasing as functions of its determinants.
  \label{fig:targeting_determinants_probability}
}
{%
Each panel shows the share (percentage) of targeting decisions from the profit-maximizing policy that change with debiasing, plotted against one determinant of the probability in Proposition~\ref{thm:differ}. The share is computed as the sample average of the disagreement indicator in Eq.~\eqref{eq:disagreement}, evaluated on the targeting split, and corresponds to the marginal probability obtained by integrating the conditional probability in the proposition over the covariate distribution.
\emph{Upper left}: Share plotted against the bias corrections per group. 
\emph{Upper right}: Share computed within 20 quantile bins of the observed distance between the CATE estimate and the targeting threshold.
\emph{Lower left}: Share as a function of the profit-lift threshold, computed over a fine grid of threshold values.
\emph{Lower right}: Share as a function of normal noise added to the CATE estimates, computed over a fine grid of increasing noise variance.
}
\end{figure}

Finally, we examine profit implications of debiasing. Figure~\ref{fig:targeting_determinants_profit} similarly shows how changes in profit vary with the four determinants of whether debiasing changes targeting decisions. The patterns largely mirror those in Figure~\ref{fig:targeting_determinants_probability} but inversely. This follows from that the original empirical targeting is a better approximation of the optimal policy, so any changes induced by debiasing tend to 
reduce profit.\footnote{\SingleSpacedXI\footnotesize Because some profit differences are close to zero, results are shown in dollar units rather than percentages; these magnitudes are illustrative given the normalization of $R=1$ and $C\in[0.005,0.75]$. What matters is the relative ordering and shape of the curves.} Consistent with this logic, we see that profit losses are approx. linear in the bias correction and concentrate among users near the profit-lift threshold, and they diminish when the profit-lift threshold or the variance of the CATE estimates increases. Across debiasing strategies, the naïve one produces large but also more fluctuating profit reductions, whereas the other debiasing strategies have smaller and more stable impact.

Taken together, these results indicate that the economic impact of debiasing is second-order relative to the extent it alters targeting. When biases and margins are moderate, correcting CATE estimates to recover GATEs may not incur substantial profit losses if one follows a risk-minimizing shrinkage strategy, especially compared to naïve approaches that ignore the estimation uncertainty.

\begin{figure}[htbp]
  \FIGURE
  {%
  \centering
  \begin{tabular}{cc}
    \vspace{0.25cm}
    \includegraphics[width=0.48\textwidth]{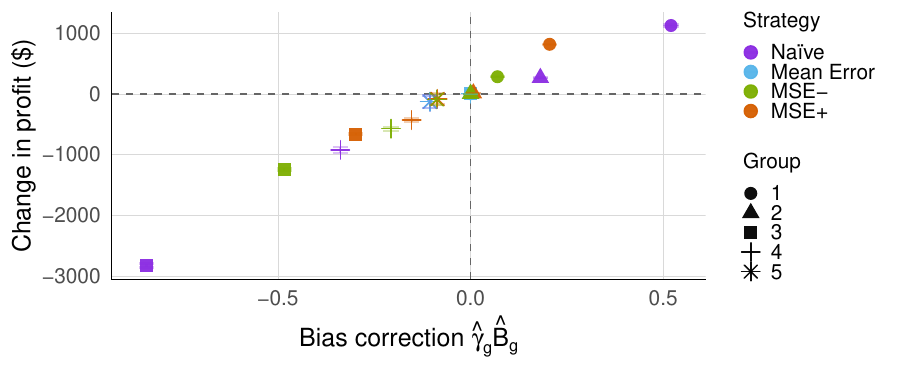}
    &
    \includegraphics[width=0.48\textwidth]{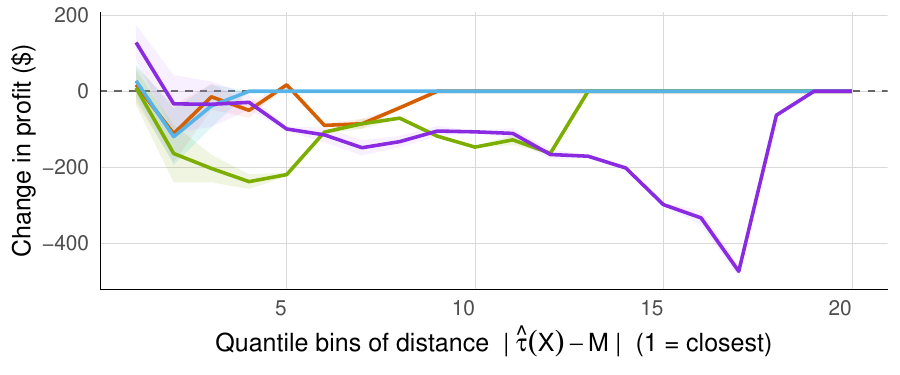}
    \\
    \includegraphics[width=0.48\textwidth]{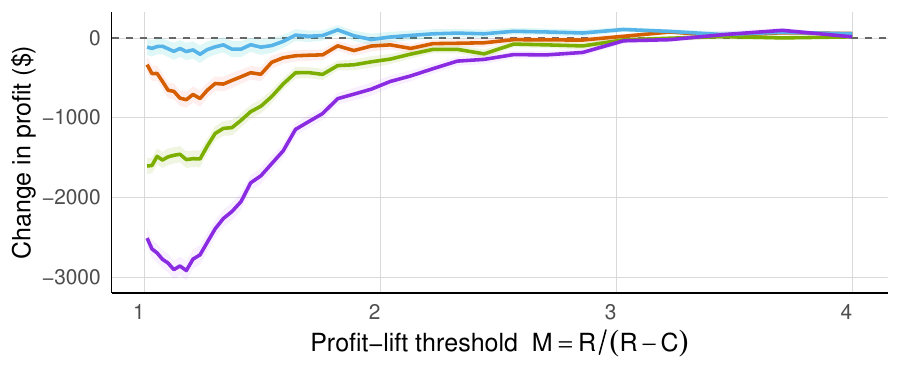}
    &
    \includegraphics[width=0.48\textwidth]{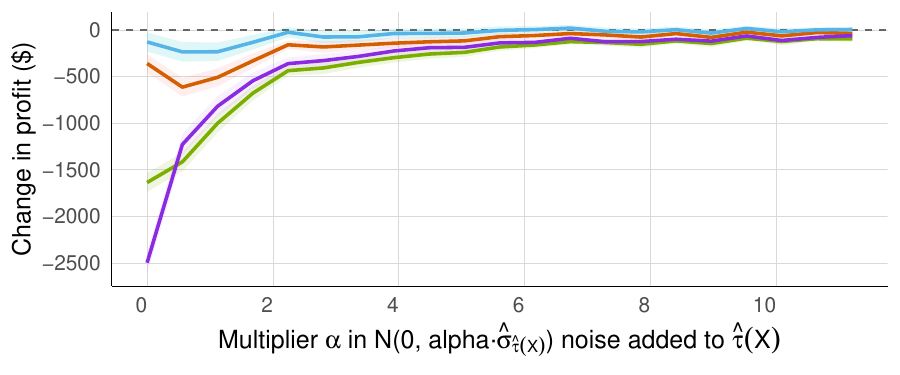}
  \end{tabular}
}
   {
    Change in targeting profits by debiasing as functions of its determinants
\label{fig:targeting_determinants_profit}
    }
    {
    Each panel shows how total profits from the profit-maximizing targeting policy change with debiasing, plotted against one determinant of the probability in Proposition~\ref{thm:differ}. Profit changes are computed per group using the IPW estimator in Eq.~\eqref{eq:ipw_est} with variance given by Eq.~\eqref{eq:ipw_var}, and then aggregated to the population level using weights $N_g/N$, thereby identifying the theoretical profit change expression in Eq.~\eqref{eq:profit_loss}.    
    \emph{Upper left}: Profit difference plotted against the bias corrections per group. 
    \emph{Upper right}: Profit difference computed within 20 quantile bins of the observed distance between the CATE estimate and the targeting threshold.
    \emph{Lower left}: Profit difference as a function of the profit-lift threshold, computed over a fine grid of threshold values.
    \emph{Lower right}: Profit difference as a function of normal noise added to the CATE estimates, computed over a fine grid of increasing noise variance.
  }
\end{figure}

\section{Practical Takeaways and Guidance}\label{sec:discussion}

We now distill our theoretical and empirical results into main takeaways and practical guidance.

\textbf{When does group bias arise?} Our work point to conditions under which group bias is most pronounced. Group bias is most likely when CATE models are trained on pooled data with heterogeneous group representation, high-dimensional or continuous covariates, and strong regularization, and when groups of interest are defined ex post to model training. In such settings, CATE predictions can have systematic group bias even when correctly identified and estimated on randomized experimental data using a model that is consistent and unbiased of CATE; see the stylized example in Section \ref{sec:motivating_example} and the empirical evidence in Section \ref{sec:booking_results}. By contrast, group bias may be negligible when groups are large, well represented, and the definition of groups is closely aligned the covariate-strata within which CATEs are estimated.

\textbf{When does debiasing matter?} In our framework, the mitigation objective is not to correct downstream decisions, but to correct CATE predictions for group bias when that bias is estimated with heterogeneous precision. Whether this correction has economic consequences depends on how often it alters decisions together with the treatment effect on the economic outcome in question; see Eq. \eqref{eq:diff_targeting}. Our results show that the economic impact is determined by how large the corrections are relative to decision thresholds and how much mass of predicted CATEs lies near those thresholds (cf. Proposition \ref{thm:differ} and Section \ref{sec:criteo_results}). Debiasing has only a second-order impact when biases are modest and profit margins are not tight, particularly when the bias correction is optimized using risk-minimizing shrinkage rather than naïvely applied (Figure \ref{fig:targeting_determinants_profit}). As such, shrinkage-based bias correction can be preferred when CATE predictions are used both for group-level causal inference and personalized targeting.

\textbf{Choosing covariates versus choosing groups.} As suggested by the discussion above, a central tension concerns the interaction between the richness of covariates used to estimate CATEs and the granularity of groups used to estimate GATEs. When CATE models rely on a small number of categorical covariates, experimental data may support estimation of GATEs at comparable levels of granularity to the CATE, so debiasing can improve the CATE models recovery of GATE with little economic cost when used for targeting. By contrast, when CATE models use high-dimensional or continuous covariates, overlap is sparse and experimentally estimating GATEs at the same granularity becomes infeasible. Practitioners must then choose between coarsening the covariates for CATE or re-defining the groups used to estimate GATE. Both choices improve statistical feasibility but entail costs: coarsening covariates reduces predictive accuracy, while broad groups attenuate meaningful heterogeneity. Either choice can diminish the profitability of personalized targeting, as implied by our characterization in Section \ref{sec:profit_loss}.

\textbf{Coarse personalization has a place.} Despite these trade-offs, organizations often report treatment effects and use them for targeting at coarse levels, even when richer data and models are available. Coarser policies can be easier to deploy, scale, and communicate \citep{Lemmens2020}, and are typically more stable in practice, as flexible CATE models can be sensitive to hyperparameters, sample noise, and random seeds \citep{Wager2018}. With this in mind, a pragmatic design principle supported by our analysis is to use the coarsest group partition that captures actionable heterogeneity for decision-making, and the richest covariate set that preserves overlap for estimating GATEs (cf. Section \ref{sec:detection} and Appendix \ref{app:gate_estimation}). At the extreme, when targeting policies are based directly on estimated GATEs rather than individual CATEs, mitigating detected group bias improves targeting performance by construction relative to no debiasing. Whether such coarse targeting outperforms personalized policies based on CATE predictions, before or after bias correction, depends on the data-generating process in each particular application.

\textbf{Within-group rankings are preserved.} 
Because the bias correction is constant across individuals within groups (cf. Eq. \eqref{eq:debiased_cate}), within-group rankings of individual CATE estimates are preserved. As a result, any policy based on within-group treatment prioritization (e.g., a top-$K$ rule for some $K < N_g$) is mathematically invariant to the bias correction. Therefore, for this class of policies, mitigating detected group bias carries no trade-off, in that group-level causal inferences from the model are improved without altering the economic returns from its personalized targeting decisions.


\section{Conclusion}\label{sec:conclusion}

Recent advances in causal machine learning have enabled researchers and practitioners to estimate highly personalized heterogeneous treatment effects. Such estimates are often aggregated to broader groups with the objective to uncover more generalizable or stable insights. However, these aggregates need not recover a causal estimand and may instead be biased relative to the corresponding group-average treatment effect. This paper studies the detection, mitigation, and implications of such group bias.

We have developed a unified statistical framework for defining, estimating, and testing for the group bias, along with a shrinkage-based bias-correction methodology that adaptively optimizes the debiasing. We have provided practical guidance for implementing and evaluating these methods using historical data from randomized experiments. The framework is agnostic to the choice of treatment effect estimand and CATE learner, imposes minimal assumptions, and requires no additional model training, but only computing sample moments. Finally, we characterized the resulting trade-offs in the context of profit-maximizing targeting and validate our theoretical results using large-scale randomized experiments at digital platforms. 

We conclude with limitations and future directions. First, our framework presumes that subgroup GATEs can be reliably estimated. This may be challenging in small or imbalanced samples but can be mitigated using regression-adjustment methods or by redefining groups. Second, while our framework is deliberately agnostic to the source of bias, understanding the underlying mechanisms remains important. Integrating in-processing methods may be useful when the source of bias is known, but risks exacerbating bias when it is not. Third, extending the framework to non-experimental settings with unmeasured confounding would be valuable but poses conceptual and technical challenges. Quasi-experimental designs such as instrumental variables, regression discontinuity, or difference-in-differences can be used to estimate the benchmark group effects, but these designs identify local average treatment effects. Consequently, the CATE estimates must be collapsed to the corresponding local estimand to enable a valid comparison. How to do so, and whether the resulting local bias estimand remains meaningful, are open questions.

\subsection*{Funding and Competing Interests}

Authors 1 was an intern at Booking.com at the start of the project that led to this research. Authors 2 and 3 are employed at Booking.com. Authors 4 and 5 have no competing interests.


\ACKNOWLEDGMENT{%
We would also like to thank participants of the following conferences and workshop for valuable feedback and comments: CODE@MIT 2025, Workshop on Platform Analytics 2023, ISMS Marketing Science Conference 2023, Annual Theory + Practice in Marketing Conference 2023, The 2023 American Causal Inference Conference, and the 2023 Marketing Science conference on Diversity, Equity And Inclusion. Finally, we would like to thank everyone at \emph{Booking.com} that enabled the research collaboration.
}

%
%
%



\interlinepenalty=10000
\SingleSpacedXI

\bibliographystyle{informs2014}
\bibliography{references}

@article{Lemmens2025,
  title={{Personalization and Targeting: How to Experiment, Learn \& Optimize}},
  author={Lemmens, Aur{\'e}lie and Roos, Jason and Gabel, Sebastian and Ascarza, Eva and Bruno, Hern{\'a}n and Gordon, Brett and Israeli, Ayelet and Feit, Elea McDonnell and Mela, Carl and Netzer, Oded},
  journal={{International Journal of Research in Marketing}},
  year={2025},
}

@article{Rasines2023,
  title={{Splitting Strategies for Post-Selection Inference}},
  author={Rasines, Daniel G and Young, G Alastair},
  journal={Biometrika},
  volume={10},
  number={3},
  pages={597-–614},
  year={2023},
}

@article{Kuchibhotla2022,
  title={{Post-Selection Inference}},
  author={Kuchibhotla, Arun K and Kolassa, John E and Kuffner, Todd A},
  journal={Annual Review of Statistics and Its Application},
  volume={9},
  pages={505--527},
  year={2022},
  publisher={Annual Reviews}
}

@article{Rinaldo2019,
  title={{Bootstrapping and Sample Splitting For High-Dimensional, Assumption-Lean Inference}},
  author={Rinaldo, Alessandro and Wasserman, Larry and G’Sell, Max},
  journal={The Annals of Statistics},
  volume={47},
  number={6},
  pages={3438--3469},
  year={2019},
  publisher={JSTOR}
}

@article{Leng2024,
  title={{Calibration of Heterogeneous Treatment Effects in Randomized Experiments}},
  author={Leng, Yan and Dimmery, Drew},
  journal={Information Systems Research},
  volume={35},
  number={4},
  pages={1721--1742},
  year={2024},
  publisher={INFORMS}
}

@article{Zhang1998,
  title={{What's the Relative Risk?: A Method of Correcting the Odds Ratio in Cohort Studies of Common Outcomes}},
  author={Zhang, Jun and Kai, F Yu},
  journal={JAMA},
  volume={280},
  number={19},
  pages={1690--1691},
  year={1998},
}

@article{Marschner2012,
  title={{Relative Risk Regression: Reliable and Flexible Methods for Log-Binomial Models}},
  author={Marschner, Ian C and Gillett, Alexandra C},
  journal={Biostatistics},
  volume={13},
  number={1},
  pages={179--192},
  year={2012},
  publisher={xford University Press}
}

@article{Hernan2006,
  title={{Estimating Causal Effects from Epidemiological Data}},
  author={Hern{\'a}n, Miguel A and Robins, James M},
  journal={Journal of Epidemiology \& Community Health},
  volume={60},
  number={7},
  pages={578--586},
  year={2006},
}

@book{Wilcox2010,
  title={{Fundamentals of Modern Statistical Methods: Substantially Improving Power and Accuracy}},
  author={Wilcox, Rand R},
  year={2010},
  publisher={Springer},
  address={New York, NY}
}

@article{Carey2022,
  title={{The Causal Fairness Field Guide: Perspectives from Social and Formal Sciences}},
  author={Carey, Alycia N and Wu, Xintao},
  journal={Frontiers in Big Data},
  volume={5},
  pages={892837},
  year={2022},
}

@article{Kleinberg2018,
  title={{Discrimination in the Age of Algorithms}},
  author={Kleinberg, Jon and Ludwig, Jens and Mullainathan, Sendhil and Sunstein, Cass R},
  journal={Journal of Legal Analysis},
  volume={10},
  pages={113--174},
  year={2018},
}

@article{Melnychuk2024,
  title={{Bounds on Representation-Induced Confounding Bias for Treatment Effect Estimation}},
  author={Melnychuk, Valentyn and Frauen, Dennis and Feuerriegel, Stefan},
  journal={International Conference on Learning Representations (ICLR)},
  year={2024}
}

@inproceedings{Corbett2017,
  title={{Algorithmic Decision Making and the Cost of Fairness}},
  author={Corbett-Davies, Sam and Pierson, Emma and Feller, Avi and Goel, Sharad and Huq, Aziz},
  booktitle={ACM SIGKDD International Conference on Knowledge Discovery and Data Mining (KDD)},
  year={2017}
}

@inproceedings{nilforoshan2022causal,
  title={{Causal Conceptions of Fairness and Their Consequences}    },
  author={Nilforoshan, Hamed and Gaebler, Johann D and Shroff, Ravi and Goel, Sharad},
  booktitle={International Conference on Machine Learning},
  pages={16848--16887},
  year={2022},
  organization={PMLR}
}

@article{Imbens2024,
  title={{LaLonde (1986) After Nearly Four Decades: Lessons Learned}},
  author={Imbens, Guido and Xu, Yiqing},
  journal={arXiv preprint arXiv:2406.00827},
  year={2024}
}

@article{Lalonde1986,
  title={{Evaluating the Econometric Evaluations of Training Programs with Experimental Data}},
  author={LaLonde, Robert J},
  journal={The American Economic Review},
  pages={604--620},
  year={1986},
}

@article{Greenland1999,
  title={{Confounding and Collapsibility in Causal Inference}},
  author={Greenland, Sander and Pearl, Judea and Robins, James M},
  journal={Statistical Science},
  volume={14},
  number={1},
  pages={29--46},
  year={1999},
}

@article{Huitfeldt2019,
  title={{On the Collapsibility of Measures of Effect in the Counterfactual Causal Framework}},
  author={Huitfeldt, Anders and Stensrud, Mats J and Suzuki, Etsuji},
  journal={Emerging Themes in Epidemiology},
  volume={16},
  pages={1--5},
  year={2019},
}

@article{Colnet2023,
  title={{Risk Ratio, Odds Ratio, Risk Difference ... Which Causal Measure is Easier to Generalize?}},
  author={Colnet, B{\'e}n{\'e}dicte and Josse, Julie and Varoquaux, Ga{\"e}l and Scornet, Erwan},
  journal={arXiv:2303.16008},
  year={2023}
}

@article{Didelez2022,
  title={{On the Logic of Collapsibility for Causal Effect Measures}},
  author={Didelez, Vanessa and Stensrud, Mats Julius},
  journal={Biometrical Journal},
  volume={64},
  number={2},
  pages={235--242},
  year={2022},
}

@article{Goodman2021,
  title={{Difference-in-Differences with Variation in Treatment Timing}},
  author={Goodman-Bacon, Andrew},
  journal={Journal of Econometrics},
  volume={225},
  number={2},
  pages={254--277},
  year={2021},
  publisher={Elsevier}
}

@article{Goldsmith2024,
  title={{Contamination Bias in Linear Regressions}},
  author={Goldsmith-Pinkham, Paul and Hull, Peter and Koles{\'a}r, Michal},
  journal={American Economic Review},
  volume={114},
  number={12},
  pages={4015--4051},
  year={2024},
}

@article{Foster2023,
  title={{Orthogonal Statistical Learning}},
  author={Foster, Dylan J and Syrgkanis, Vasilis},
  journal={The Annals of Statistics},
  volume={51},
  number={3},
  pages={879--908},
  year={2023},
}

@article{Hitsch2024,
  title = {{Heterogeneous Treatment Effects and Optimal Targeting Policy Evaluation}},
  author={Hitsch, G{\"u}nter J and Misra, Sanjog and Zhang, Walter W},
  journal={Quantitative Marketing and Economics},
  volume={22},
  number={2},
  pages={115--168},
  year={2024},
}

@article{Lemmens2020,
author = {Lemmens, Aurelie and Gupta, Sunil},
journal = {Marketing Science},
number = {5},
pages = {956--973},
title = {{Managing Churn to Maximize Profits}},
volume = {39},
year = {2020}
}

@article{Athey2016,
  title={{Recursive Partitioning for Heterogeneous Causal Effects}},
  author={Athey, Susan and Imbens, Guido},
  journal={Proceedings of the National Academy of Sciences},
  volume={113},
  number={27},
  pages={7353--7360},
  year={2016},
}

@article{Feuerriegel2024,
  title={{Causal Machine Learning for Predicting Treatment Outcomes}},
  author={Feuerriegel, Stefan and Frauen, Dennis and Melnychuk, Valentyn and Schweisthal, Jonas and Hess, Konstantin and Curth, Alicia and Bauer, Stefan and Kilbertus, Niki and Kohane, Isaac S and van der Schaar, Mihaela},
  journal={Nature Medicine},
  volume={30},
  number={4},
  pages={958--968},
  year={2024},
}

@inproceedings{Taskesen2021,
  title={{A Statistical Test for Probabilistic Fairness}},
  author={Taskesen, Bahar and Blanchet, Jose and Kuhn, Daniel and Nguyen, Viet Anh},
  booktitle={ACM Conference on Fairness, Accountability, and Transparency},
  year={2021}
}

@inproceedings{Yik2022,
  title={{Identifying Bias in Data using Two-Distribution Hypothesis Tests}},
  author={Yik, William and Serafini, Limnanthes and Lindsey, Timothy and Monta{\~n}ez, George D},
  booktitle={AAAI/ACM Conference on AI, Ethics, and Society},
  year={2022}
}

@inproceedings{Diciccio2020,
  title={{Evaluating Fairness using Permutation Tests}},
  author={DiCiccio, Cyrus and Vasudevan, Sriram and Basu, Kinjal and Kenthapadi, Krishnaram and Agarwal, Deepak},
  booktitle={ACM SIGKDD International Conference on Knowledge Discovery and Data Mining (KDD)},
  year={2020}
}

@inproceedings{Diemert2018,
title={{A Large Scale Benchmark for Uplift Modeling}},
author = {Diemert, Eustache and Betlei, Artem and Renaudin, Christophe and Massih-Reza, Amini},
booktitle={ACM SIGKDD International Conference on Knowledge Discovery and Data Mining (KDD)},
year = {2018}
}

@article{nie2021quasi,
  title={{Quasi-Oracle Estimation of Heterogeneous Treatment Effects}},
  author={Nie, Xinkun and Wager, Stefan},
  journal={Biometrika},
  volume={108},
  number={2},
  pages={299--319},
  year={2021},
  publisher={Oxford University Press}
}

@article{Wager2018,
  title={{Estimation and Inference of Heterogeneous Treatment Effects using Random Forests}},
  author={Wager, Stefan and Athey, Susan},
  journal={Journal of the American Statistical Association},
  volume={113},
  number={523},
  pages={1228--1242},
  year={2018},
}

@article{Gordon2019,
  title={{A Comparison of Approaches to Advertising Measurement: Evidence from Big Field Experiments at Facebook}},
  author={Gordon, Brett R and Zettelmeyer, Florian and Bhargava, Neha and Chapsky, Dan},
  journal={Marketing Science},
  volume={38},
  number={2},
  pages={193--225},
  year={2019},
}

@article{Gordon2023,
  title={{Close Enough? A Large-Scale Exploration of Non-experimental Approaches to Advertising Measurement}},
  author={Gordon, Brett R and Moakler, Robert and Zettelmeyer, Florian},
  journal={Marketing Science},
  volume={42},
  number={4},
  pages={768--793},
  year={2023},
}

@article{castelnovo2022clarification,
  title={{A Clarification of the Nuances in the Fairness Metrics Landscape}},
  author={Castelnovo, Alessandro and Crupi, Riccardo and Greco, Greta and Regoli, Daniele and Penco, Ilaria Giuseppina and Cosentini, Andrea Claudio},
  journal={Scientific Reports},
  volume={12},
  number={1},
  pages={4209},
  year={2022},
  publisher={Nature Publishing Group UK London}
}

@article{Chernozhukov2018,
author = {Chernozhukov, Victor and Chetverikov, Denis and Demirer, Mert and Duflo, Esther and Hansen, Christian and Newey, Whitney and Robins, James},
journal = {Econometrics Journal},
number = {1},
pages = {C1--C68},
title = {{Double/Debiased Machine Learning for Treatment and Structural Parameters}},
volume = {21},
year = {2018}
}

@article{athey2019generalized,
  title={{Generalized Random Forests}},
  author={Athey, Susan and Tibshirani, Julie and Wager, Stefan},
  journal={The Annals of Statistics},
  volume={47},
  number={2},
  pages={1148--1178},
  year={2019}
}

@article{athey2022effective,
  title={{Effective and Scalable Programs to Facilitate Labor Market Transitions for Women in Technology}},
  author={Athey, Susan and Palikot, Emil},
  journal={arXiv preprint arXiv:2211.09968},
  year={2022}
}

@article{agrawal2025economics,
  title={{The Economics of Algorithmic Personalization: Evidence from an Educational Technology Platform}},
  author={Agrawal, Keshav and Carleton Athey, Susan and Kanodia, Ayush and Nath, Shanjukta and Palikot, Emil},
  journal={Available at SSRN 5996014},
  year={2025}
}

@techreport{Chernozhukov2018b,
  title={{Generic Machine Learning Inference on Heterogeneous Treatment Effects in Randomized Experiments, with an Application to Immunization in India}},
  author={Chernozhukov, Victor and Demirer, Mert and Duflo, Esther and Fernandez-Val, Ivan},
  year={2018},
  institution={National Bureau of Economic Research}
}

@article{Ascarza2022,
  title={{Eliminating Unintended Bias in Personalized Policies using Bias-Eliminating Adapted Trees (BEAT)}},
  author={Ascarza, Eva and Israeli, Ayelet},
  journal={Proceedings of the National Academy of Sciences},
  volume={119},
  number={11},
  pages={e2115293119},
  year={2022},
}

@article{Chouldechova2020,
  title={{A Snapshot of the Frontiers of Fairness in Machine Learning}},
  author={Chouldechova, Alexandra and Roth, Aaron},
  journal={Communications of the ACM},
  volume={63},
  number={5},
  pages={82--89},
  year={2020},
  publisher={ACM New York, NY, USA}
}

@article{Doob1935,
  title={{The Limiting Distributions of Certain Statistics}},
  author={Doob, Joseph L},
  journal={The Annals of Mathematical Statistics},
  volume={6},
  number={3},
  pages={160--169},
  year={1935},
}

@article{horvitz1952,
  title={{A Generalization of Sampling Without Replacement from a Finite Universe}},
  author={Horvitz, Daniel G and Thompson, Donovan J},
  journal={Journal of the American Statistical Association},
  volume={47},
  number={260},
  pages={663--685},
  year={1952}
}

@article{lin2013agnostic,
  title={{Agnostic Notes on Regression Adjustments to Experimental Data: Reexamining Freedman's Critique}},
  author={Lin, Winston},
  journal={The Annals of Applied Statistics},
  pages={295--318},
  year={2013},
  publisher={JSTOR}
}

@inproceedings{deng2013improving,
  title={{Improving the Sensitivity of Online Controlled Experiments by Utilizing Pre-Experiment Data}},
  author={Deng, Alex and Xu, Ya and Kohavi, Ron and Walker, Toby},
  booktitle={Proceedings of the sixth ACM international conference on Web search and data mining},
  pages={123--132},
  year={2013}
}

@article{deng2023augmentation,
  title={{From Augmentation to Decomposition: A New Look at CUPED in 2023}},
  author={Deng, Alex and Hagar, Luke and Stevens, Nathaniel and Xifara, Tatiana and Yuan, Lo-Hua and Gandhi, Amit},
  journal={arXiv preprint arXiv:2312.02935},
  year={2023}
}

@article{Warren2020,
 title={{Trustworthiness Before Trust -- Covid-19 Vaccine Trials and the Black Community}},
  author={Warren, Rueben C and Forrow, Lachlan and Hodge Sr, David Augustin and Truog, Robert D},
  journal={New England Journal of Medicine},
  volume={383},
  number={22},
  pages={e121},
  year={2020},
}

@misc{KFF2020,
  author = {Artiga, Samantha and Kates, Jennifer and Michaud, Josh and Hill, Latoya}, 
  publisher    = {Kaiser Family Foundation},
  title     = {{Racial Diversity within COVID-19 Vaccine Clinical Trials: Key Questions and Answers}},
  year      = {2020},
  url       = {https://www.kff.org/racial-equity-and-health-policy/issue-brief/racial-diversity-within-covid-19-vaccine-clinical-trials-key-questions-and-answers/}
}

@book{van2000asymptotic,
  title={{Asymptotic Statistics}},
  author={Van der Vaart, Aad W},
  volume={3},
  year={2000},
  publisher={Cambridge University Press}
}

@article{Huang2023,
  title={{Debiasing Treatment Effect Estimation for Privacy-Protected Data: A Model Audition and Calibration Approach}},
  author={Huang, Ta-Wei and Ascarza, Eva},
  journal={Available at SSRN 4575240},
  year={2023}
}

@book{barocas2023,
  title={{Fairness and Machine Learning: Limitations and Opportunities}},
  author={Barocas, Solon and Hardt, Moritz and Narayanan, Arvind},
  year={2023},
  publisher={MIT press}
}

@article{Chohlas2023,
  title={{Designing Equitable Algorithms}},
  author={Chohlas-Wood, Alex and Coots, Madison and Goel, Sharad and Nyarko, Julian},
  journal={Nature Computational Science},
  volume={3},
  number={7},
  pages={601--610},
  year={2023},
}

@book{Hernan2023,
  title={{Causal Inference: What If}},
  author={Hernan, Miguel and Robins, James},
  year={2023},
  publisher={Boca Raton: Chapman \& Hall/CRC}
}

@techreport{Rambachan2020,
  title={{An Economic Approach to Regulating Algorithms}},
  author={Rambachan, Ashesh and Kleinberg, Jon and Mullainathan, Sendhil and Ludwig, Jens},
  year={2020},
  institution={National Bureau of Economic Research}
}

@inproceedings{Goldenberg2020,
  title={{Free Lunch! Retrospective Uplift Modeling for Dynamic Promotions Recommendation within ROI Constraints}},
  author={Goldenberg, Dmitri and Albert, Javier and Bernardi, Lucas and Estevez, Pablo},
  booktitle={ACM Conference on Recommender Systems (RecSys)},
  year={2020}
}

@article{Holland1986,
  title={{Statistics and Causal Inference}},
  author={Holland, Paul W},
  journal={Journal of the American Statistical Association},
  volume={81},
  number={396},
  pages={945--960},
  year={1986},
}

@article{Kennedy2023,
  title={{Towards Optimal Doubly Robust Estimation of Heterogeneous Causal Effects}},
  author={Kennedy, Edward H},
  journal={Electronic Journal of Statistics}, 
  volume={17},
  number={2},
  pages={3008--3049},
  year={2023}
}

@article{Kunzel2019,
  title={{Metalearners for Estimating Heterogeneous Treatment Effects using Machine Learning}},
  author={K{\"u}nzel, S{\"o}ren R and Sekhon, Jasjeet S and Bickel, Peter J and Yu, Bin},
  journal={Proceedings of the National Academy of Sciences},
  volume={116},
  number={10},
  pages={4156--4165},
  year={2019},
}

@article{De2022,
  title={{Algorithmic Fairness in Business Analytics: Directions for Research and Practice}},
  author={De-Arteaga, Maria and Feuerriegel, Stefan and Saar-Tsechansky, Maytal},
  journal={Production and Operations Management},
  volume={31},
  number={10},
  pages={3749--3770},
  year={2022},
}

@article{Kraus2024,
  title={Data-driven allocation of preventive care with application to diabetes mellitus type II},
  author={Kraus, Mathias and Feuerriegel, Stefan and Saar-Tsechansky, Maytal},
  journal={Manufacturing \& Service Operations Management},
  volume={26},
  number={1},
  pages={137--153},
  year={2024}
}


\newpage

\def\AppendixFontSize{\OneAndAHalfSpacedXI}
\begin{APPENDICES}

\begin{center}
~~

\Huge\vspace{3cm}
\textbf{Online Appendix}
\vspace{1cm}
\end{center}

\OneAndAHalfSpacedXI

\newpage
\section{Proofs}
\label{app:proofs}

\subsection{Proof of Proposition~\ref{thm:conv_in_dist}}
\label{proof:conv_in_dist}

We prove the result for the relative effect estimand. The argument for the additive estimand is analogous. Throughout, we allow the model-implied and experimentally estimated GATEs to be arbitrarily dependent. Without loss of generality, we suppress the group index $g$. Throughout, 
$\widehat\tau^f(X)$ is treated as fixed, since we take the CATE model to already be fitted. 

Let $h \colon \mathbb{R}^2 \to \mathbb{R}$ be defined by $h(x,y)=x/y$, with $y\neq 0$. The function $h$ is continuously differentiable, with gradient
\begin{equation}
    \nabla h(x,y)=\left(\frac{1}{y},\,-\frac{x}{y^2}\right).
\end{equation}
Let $\bar Y_1 \coloneqq N_1^{-1}\sum_{i \colon T_i=1}Y_i$ and $\bar Y_0\coloneqq N_0^{-1}\sum_{i \colon T_i=0}Y_i$, and let $\mu_1\coloneqq\E[Y(1)]$ and $\mu_0\coloneqq\E[Y(0)]$, so that the relative GATE is $\tau=\mu_1/\mu_0$. We have 
\begin{equation}
\sqrt{N}
\begin{pmatrix}
\bar Y_1 -\mu_1\\
\bar Y_0 -\mu_0
\end{pmatrix}
\;\overset{d}{\longrightarrow}\;
\mathcal N(0,\Sigma),
\end{equation}
where $\Sigma$ is a finite, positive semidefinite covariance matrix. This follows from standard central limit theorems for sample means under random assignment and finite second moments.

Let us define the experimental estimator of the relative GATE as 
$\widehat\tau = h(\bar Y_1,\bar Y_0)$. By the multivariate delta method \citep{van2000asymptotic, Doob1935},
\begin{equation}
\sqrt{N}(\widehat\tau-\tau)
\;\overset{d}{\longrightarrow}\;
\mathcal N\!\left(0,\;\sigma_\tau^2\right),
\end{equation}
where
\begin{equation}
\sigma_\tau^2
=
\nabla h(\mu_1,\mu_0)^\top
\Sigma
\nabla h(\mu_1,\mu_0).
\end{equation}

Now, let $\widehat\tau^{f} = \E_N[\widehat W \widehat \tau^f(X)]$ denote the model-implied GATE estimator, and let $\tau^{f} = \E_P[W \widehat \tau^f(X)]$ be its population counterpart. Here, $\widehat\tau^{f}$ is an empirical mean of a fixed function with estimated weights, and therefore admits a central limit under standard regularity conditions. In particular, by the assumptions of 
Proposition~\ref{thm:conv_in_dist},
\begin{equation}
\sqrt{N}
\begin{pmatrix}
\widehat\tau^{f}-\tau^{f}\\
\widehat\tau-\tau
\end{pmatrix}
\;\overset{d}{\longrightarrow}\;
\mathcal N(0,\Omega),
\end{equation}
for some finite covariance matrix $\Omega$, allowing for arbitrary dependence. 

Recall that the group bias is $b=\tau^{f}-\tau$ and its estimator is $\widehat B=\widehat\tau^{f}-\widehat\tau$. Let $a=(1,-1)^\top$. By the continuous mapping theorem, we yield
\begin{equation}
\sqrt{N}\big(\widehat B-b\big)
=
\sqrt{N}\,a^\top
\begin{pmatrix}
\widehat\tau^{f}-\tau^{f}\\
\widehat\tau-\tau
\end{pmatrix}
\;\overset{d}{\longrightarrow}\;
\mathcal N\!\left(0,\;a^\top\Omega a\right),
\end{equation}
where $\sigma^2=a^\top\Omega a$. This concludes the proof.
\hfill$\square$

\subsection{Proof of Proposition~\ref{prop:opt_mse}}
\label{proof:opt_alpha}

We first derive how the MSE in terms of the debiasing error depends on the shrinkage factor. For ease of notation, we omit the group index $g$. 

Recall that the debiasing error is $b_{\gamma}= b - \gamma \widehat B$. Let $b_{\gamma}^2  = (b - \gamma \widehat B)^2$. For any consistent estimator, $\widehat B$ is a normal random variable. Using the properties of squared normal random variables, we have
\begin{align}
    \E [b_{\gamma}^2] &= \gamma^2 \sigma^2 + b^2(\gamma - 1)^2.
\end{align}
We find the MSE-minimizing value $\gamma^*$ by solving for the first-order conditions. We have
\begin{align*}
    \frac{\partial \E b_{\gamma}^2}{\partial \gamma} 
    &= \frac{\partial}{\partial \gamma}\left(\gamma^2 \sigma^2 + b^2(\gamma - 1)^2\right) 
    = 2\gamma \sigma^2 + 2b^2(\gamma - 1) .
\end{align*}
Setting the derivative to 0 and solving for $\gamma$ yields
\begin{align}
    & 2\gamma \sigma^2 + 2b^2(\gamma - 1) = 0 \\
    \iff & \gamma\sigma^2 = b^2 - \gamma b^2 \\
    \iff&  \gamma(\sigma^2 + b^2) = b^2 .
\end{align}
Thus,
\begin{align}\label{eq:expression_1}
    \gamma^* &= \frac{b^2}{\sigma^2 + b^2} .
\end{align}
This concludes the proof.
$\hfill \square$

\subsection{Proof to Proposition \ref{thm:differ}}\label{app:proof_differ}

Since $\widehat{\gamma}_g$ and $\widehat{B}_g$ are estimated prior to debiasing, their product is treated as fixed. 

The unconstrained policy treats if $\widehat{\pi}(X)=\mathds{1}\{\widehat{\tau}^{f}(X)>M\}$, while the constrained policy treats if $\widehat{\pi}(X;\widehat{\gamma}_g\widehat{B}_g)=\mathds{1}\{\widehat{\tau}^{f}(X)-\widehat{\gamma}_g\widehat{B}_g>M\}$. The  event $\widehat{\pi}(X;\widehat{\gamma}_g\widehat{B}_g) \neq \widehat{\pi}(X)$ thus occurs if and only if
\begin{equation}\label{eq:inequality}
\big|\widehat{\gamma}_g\widehat{B}_g\big|>\big|\widehat{\tau}^{f}(X)-M\big|,
\end{equation}
or equivalently,
\begin{equation}
M-\big|\widehat{\gamma}_g\widehat{B}_g\big|<\widehat{\tau}^{f}(X)<M+\big|\widehat{\gamma}_g\widehat{B}_g\big|.
\end{equation}
Under the assumed condition that $\widehat \tau^f(X) = \tau(X) + b(X) + \epsilon(X)$, it follows that 
$\widehat{\tau}^{f}(X)\mid X \sim \mathcal{N}\!\big(\tau(X)+b(X),\,\sigma^2(X)\big)$.
Then, conditional on $X$, we have
\begin{align}
    \Pr\big(\widehat{\pi}(X;\widehat{\gamma}_g\widehat{B}_g)\neq\widehat{\pi}(X)\mid X\big)
&=\Pr\big(M-|\widehat{\gamma}_g\widehat{B}_g|<\widehat{\tau}^{f}(X)<M+|\widehat{\gamma}_g\widehat{B}_g|\mid X\big)\\
&=\Phi\!\left(\frac{M+|\widehat{\gamma}_g\widehat{B}_g|-[\tau(X)+b_g(X)]}{\sigma(X)}\right)
-\Phi\!\left(\frac{M-|\widehat{\gamma}_g\widehat{B}_g|-[\tau(X)+b_g(X)]}{\sigma(X)}\right).
\end{align}
This concludes the proof.
\hfill $\square$

\newpage
\section{Estimating the Group Bias}
\label{app:gate_estimation}

We describe how to estimate the group bias using data from a randomized experiment. Following our decomposition in Section \ref{sec:detection}, any estimator of the group bias admits the decomposition
\begin{align}
    \widehat B_g = \widehat\tau_g^f - \widehat\tau_g.
\end{align}
so estimation reduces to obtaining the model-implied GATE $\widehat\tau_g^f$ and an experimental estimator $\widehat\tau_g$.

Suppose we have $N^{\text{pred}}_g$ observations to estimate $\widehat\tau_g^f$ and $N^{\text{exp}}_g$ observations to estimate $\widehat\tau_g$. These samples may coincide, although using independent data simplifies inference by avoiding covariance between the two estimators. Following the collapsibility discussions in Section~\ref{sec:estimands}, the model-implied GATE is estimated as
\begin{equation}
\widehat\tau_g^f
    =     
    \frac{1}{N^{\text{pred}}_g} \sum^{N^{\text{pred}}_g}_{j=1}\widehat W_j \widehat \tau^{f}(X_j),
\end{equation}
with weights $\widehat W_j$ appropriately chosen and estimated for the scale of the treatment effects as described in Section \ref{sec:detection}. 

To estimate $\tau_g$, any unbiased and $\sqrt{N_g}$-consistent estimator may be used (cf. Proposition~\ref{thm:conv_in_dist}). We describe several such estimators below, beginning with contrast-in-means and ratio-of-means, and then turning to variance-reducing regression adjustments. Identification relies on the standard potential outcomes assumptions of unconfoundedness, overlap, and SUTVA, which holds for the randomized experiments we develop our framework for.

\subsection{Unadjusted Contrast-in-Means}

The simplest estimator is simply contrast-in-means. For additive treatment effects,
\begin{equation}
    \widehat \tau_g
    =
    \frac{1}{N^{\text{exp}}_{1g}} \sum_{i:T_i=1} Y_i - \frac{1}{N^{\text{exp}}_{0g}} \sum_{i:T_i=0} Y_i
\end{equation}
whereas for ratio effects,
\begin{equation}
    \widehat \tau_g
    =  
    \frac{(N^{\text{exp}}_{1g})^{-1} \sum_{i:T_i=1} Y_i}{(N^{\text{exp}}_{0g})^{-1}\sum_{i:T_i=0} Y_i}
    .
\end{equation}
Here, $N^{\text{exp}}_{1g}$ and $N^{\text{exp}}_{0g}$ denote treatment and control sample sizes, $N^{\text{exp}}_{1g} + N^{\text{exp}}_{0g}=N^{\text{exp}}_{g}$. These estimators are the direct sample analogs of the definition of additive- and ratio-GATEs in Eq. \eqref{eq:gate}.

Putting things together, the group bias estimator is then
\begin{equation}\label{eq:est_bias_diff}
    \widehat B_g
    =     
    \underbrace{\frac{1}{N^{\text{pred}}_g} \sum^{N^{\text{pred}}_g}_{j=1} \widehat \tau^{f}(X_j)}_{\textrm{average predicted CATE}} 
    - 
    \underbrace{\left(
        \frac{1}{N^{\text{exp}}_{1g}} \sum_{i:T_i=1} Y_i - \frac{1}{N^{\text{exp}}_{0g}} \sum_{i:T_i=0} Y_i
    \right)}_{\textrm{estimated additive GATE}},
\end{equation}
or
\begin{equation}\label{eq:est_bias_ratio}
    \widehat B_g
    =     
    \underbrace{\frac{1}{N^{\text{pred}}_g} \sum^{N^{\text{pred}}_g}_{j=1}\widehat W_j \widehat \tau^{f}(X_j)}_{\textrm{weighted-average predicted CATE}} 
    - 
    \underbrace{
        \frac{(N^{\text{exp}}_{1g})^{-1} \sum_{i:T_i=1} Y_i}{(N^{\text{exp}}_{0g})^{-1}\sum_{i:T_i=0} Y_i}
    }_{\textrm{estimated relative GATE}}
    .
\end{equation}

The expressions above use contrast-in-means to estimate the GATE, as this is the default estimator in randomized experiments and provides a simple, unbiased, and model-free estimator that meets the conditions for Proposition\ref{thm:conv_in_dist}. However, variance reduction can be achieved through covariate adjustment via regression, such as Lin’s interactive model \citep{lin2013agnostic}, CUPED \citep{deng2013improving,deng2023augmentation}, or an appropriate ML estimator of the ATE. These estimators control for variance in outcomes explained by pre-treatment covariates and will increase the power of the test, especially for small or noisy groups, while remaining nonparametric in identification under no additional assumptions. We next explain how those methods for our setting.

\subsection{Covariate Adjustment via Regression}

We now explain regression estimators for the GATE $\tau_g$ that can be used in place of the unadjusted contrast-in-means. Regression adjustment has long been used to improve the efficiency of ATE estimators in randomized experiments, and is valid because randomization ensures unbiasedness even after adjusting for baseline covariates. By applying these methods separately within each group $g \in \mathcal{G}$, they yield unbiased and variance-reduced estimators of GATEs. These more precise GATE estimates $\widehat \tau_g$ can be plugged into our bias estimator in Eq.~\eqref{eq:bias_decomp}, increasing the power of the statistical test for detection (Eq. \eqref{eq:test}), particularly for small or noisy groups.

We begin with regression-adjustment methods for the absolute (difference) scale, then turn to relative (ratio) scales.

\subsubsection{Absolute-Scale Regression Adjustment}
 
\citet{lin2013agnostic} shows that regression of outcomes on treatment, centered covariates, and treatment--covariate interactions yields an estimator of the ATE that is design-unbiased under randomization and never asymptotically less efficient than the simple difference in means. Applying this to observations $i=1,\ldots,N_g$ within a group $g$, the regression is
\begin{align}\label{eq:lin}
    Y_i = \alpha_g + \tau_g^{\text{Lin}} T_i + X_i^\top \beta_{0g} + (T_i \cdot X_i)^\top \beta_{1g} + \varepsilon_i,
\end{align}
where $X_i$ should be centered. The OLS coefficient $\widehat \tau^{\text{Lin}}_g$ is an unbiased estimator of $\tau_g$ and has asymptotic variance
\begin{align}
    \text{Var}(\widehat\tau^{\text{Lin}}_g) \approx \frac{1}{N_g}\left(\frac{\sigma^2_{1g}}{p_g} + \frac{\sigma^2_{0g}}{1-p_g}\right),
\end{align}
where $p_g=N^1_g/N_g$ is the treatment share and $\sigma^2_{tg}$ are residual variances from regressing $Y(t)$ on $X$ within group $g$. Robust standard errors (i.e., Huber--White HC2/HC3) consistently estimate this variance \citep{lin2013agnostic}, also under heteroskedasticity. Relative to contrast-in-means, Lin regression reduces variance by replacing raw outcome variance with residual variance, often substantially when $X$ is predictive of $Y$.

Another estimator is CUPED (``controlled experiments using pre-experiment data'') \citep{deng2013improving}. This is a regression-instantiation of the control variate technique from Monte Carlo methods, which leverages pre-treatment (e.g., lagged) outcomes, covariates, or other statistics $Z_i$ unaffected by treatment as control variates. 

CUPED works as follows. Define the adjusted outcome
\begin{align}
    Y_i^{\text{adj}} = Y_i - \theta_g Z_i,
\end{align}
with $\theta_g$ chosen to minimize variance. The CUPED estimator of GATE is then
\begin{align}\label{eq:cuped}
    \widehat\tau^{\text{CUPED}}_g = \overline{Y^{\text{adj}}}_{1g} - \overline{Y^{\text{adj}}}_{0g}
    = \Delta Y_g - \theta_g \Delta Z_g,
\end{align}
where $\overline{Y^{\text{adj}}}_{1g}$ and $\overline{Y^{\text{adj}}}_{0g}$ denote the within-group means of the adjusted outcome for treated and control observations, and $\Delta Y_g$ and $\Delta Z_g$ denote the corresponding treatment--control differences in outcomes and covariates. This estimator remains unbiased for $\tau_g$ because $\E[\Delta Z_g]=0$ under randomization. Its variance is
\begin{align}
    \textrm{Var}(\widehat\tau^{\text{CUPED}}_g) = \textrm{Var}(\Delta Y_g) + \theta_g^2 \textrm{Var}(\Delta Z_g) - 2\theta_g \textrm{Cov}(\Delta Y_g, \Delta Z_g),
\end{align}
which is minimized by setting $\theta_g^\star=\textrm{Cov}(\Delta Y_g, \Delta Z_g) / \textrm{Var}(\Delta Z_g)$. The resulting variance reduction in this CUPED estimator of GATE relative to unadjusted contrast-in-means is proportional to $1-\rho^2$, where $\rho$ is the correlation between $Y$ and $Z$ in group $g$ \citep{deng2013improving}. 

CUPAC (``Control Using Predictions as Covariates'') generalizes CUPED by using flexible ML models to construct augmentation terms from high-dimensional covariates, further reducing variance \citep{deng2023augmentation}. Variance can be estimated using plug-in sample covariances or by nonparametric bootstrap.

More generally, semiparametric ML estimators for the ATE (such as double ML \citep{Chernozhukov2018}, causal forests \citep{athey2019generalized}, and related meta-learners \citep{Kunzel2019}) can be applied within groups to estimate $\tau_g$ nonparametrically. These approaches combine models for outcomes and propensity scores as nuisance functions with cross-fitting to yield $\sqrt{N_g}$-consistent and asymptotically normal estimators of average treatment effects. Importantly, when used in this way, these causal ML methods do \emph{not} suffer from the aggregation-induced group bias studied in this paper, because they estimate the GATE directly rather than by first estimating individual-level CATEs and then aggregating them. As such, they provide valid and efficient estimators of GATEs that can be used as plug-ins for the group-bias estimator in Eq.~\eqref{eq:bias_decomp}.

\subsubsection{Relative-Scale Regression Adjustment}

Regression adjustment for ratio effects can be implemented either via generalized linear models (GLMs) with log links, providing natural analogs of \citet{lin2013agnostic} and CUPED/CUPAC, or via specialized estimators from biostatistics and epidemiology. As in the additive-scale case, these estimators preserve design-unbiasedness under randomization and reduce variance by conditioning on covariates that predict outcomes. The difference is that one must take an additional step to ensure that the regression recovers the relative-scale GATE.

We start by explaining log-link regression analogs of Lin’s regression and CUPED/CUPAC. For a group $g$, the relative-scale analogue of Lin’s interactive regression (cf. Eq.~\eqref{eq:lin}) is
\begin{align}
    \log \E[Y \mid T, X, G=g]
    = \nu_g + \theta_g T + X^\top \beta_{0g} + (T \times X)^\top \beta_{1g}.
\end{align}
A CUPED-style analogue includes the pre-treatment controls $Z$ additively in the linear predictor,
\begin{align}
    \log \E[Y \mid T, X, Z, G=g] = \nu_g + \theta_g T + X^\top \beta_g + \gamma_g Z,
\end{align}
whereas CUPAC replaces $Z$ with an ML-prediction $f(Z)$ of $Y$.

In both cases, the regression implies a multiplicative treatment effect on the conditional mean. The corresponding relative GATE
\begin{align}
    \tau_g = \frac{\E[Y(1)\mid G=g]}{\E[Y(0)\mid G=g]}
\end{align}
must therefore be obtained by collapsing the fitted conditional mean model, following Definition \ref{def:collapse} in Section~\ref{sec:estimands}. To do so, compute arm-specific predicted means
\begin{align}
    \widehat\mu_{1g} &= \frac{1}{N_g}\sum_{i\in g}\exp\!\big(\widehat\nu_g+\widehat\theta_g+X_i^\top\widehat\beta_{0g}+X_i^\top\widehat\beta_{1g}\big),\\
    \widehat\mu_{0g} &= \frac{1}{N_g}\sum_{i\in g}\exp\!\big(\widehat\nu_g+X_i^\top\widehat\beta_{0g}\big),
\end{align}
Taking the ratio then yields the GLM-estimate of the GATE:
\begin{align}\label{eq:glm_gate}
    \widehat\tau^{\text{GLM}}_g = \widehat\mu_{1g}/\widehat\mu_{0g}.
\end{align}
Robust (sandwich) standard errors combined with the delta method yield valid asymptotic inference for $\widehat\tau^{\text{GLM}}_g$. Nonparametric bootstrap provides a convenient alternative, particularly when the control variate $Z$ is estimated using ML.

The (bio)statistics and epidemiology literature literature has also developed specialized regression estimators for the relative risk estimand, which transfer directly to our framework when outcomes are binary and treatment effects are defined on a ratio scale. The common thread among these estimators is to fit a GLM with a log-link for the conditional relative risk. For instance, \citet{Greenland1999} describe a \emph{log-binomial regression model} that directly applies to binary outcomes, and \citet{Marschner2012} develop a \emph{quasi-Poisson (log-link) regression model}, using that a binary outcome can be modeled a quasi-count. Finally, \citet{Zhang1998} proposes a transformation from logistic regression estimates of the covariate-adjusted odds-ratio that approximates the conditional risk ratio, which is useful when the log-binomial regression fails to converge and the binary outcome should not be modeled as a quasi-count. All of these methods are developed for estimating the conditional relative risk in the population, and so estimating the GATE for use in our framework again forming the conditional means by $\widehat\mu_{1g}$  and $\widehat\mu_{0g}$ and then taking $\widehat\tau_g = \widehat\mu_{1g} / \widehat\mu_{0g}$, just like in Eq. \eqref{eq:glm_gate}.

\clearpage
\newpage
\section{Offline Evaluation on Historical Data}\label{app:evaluation}

We now show how the detection and mitigation be can evaluated ``offline'' on historical data. This is useful when a new experiment cannot be run but past data is available, or to test performance prior to running an eventual experiment. For instance, one may wish to test several debiasing strategies against each other and then select only the most promising to roll out or to test ``online''. 

Algorithm~\ref{alg:evaluation} summarizes the end-to-end procedure, enabling systematic comparison of multiple strategies under identical conditions. For assumption-lean inference, the data used to estimate CATE must be independent of the data used for bias detection, the and data used for mitigating the bias must further be independent of the detection data. 
In organizational settings, this independence holds by construction whenever a CATE model is first trained, its bias is later evaluated on some new observations, and the mitigation is finally applied to yet another set of observations. 

On historical data, however, we must take care to enforce independence. We do so simply by splitting the data per group into four mutually exclusive sets: (i) CATE predictions to collapse towards model-implied GATE; (ii) experimental data for estimating GATE; (iii) some CATE predictions to debias; and (iv) hold-out experimental data to re-estimate the GATE. This four-way split prevents information leakage and preserves valid inference.

To further avoid having to estimate covariances, we recommend a nonparametric bootstrap nested within this four-way split by group. The use of bootstrap lets us approximate the sampling distributions of all key estimates under minimal assumptions, irrespective of the CATE model or effect scale. As shown in Algorithm~\ref{alg:evaluation}, we also recommend computing all variance estimates once and reusing them across all steps. This improves computational efficiency by enabling us to execute the statistical tests both for the initial detection and the post-mitigation evaluation in a single pass at the end. As earlier noted, we may use a Bonferroni-adjusted significance level of $\alpha/(4|\mathcal G|)$, where the additional division by four accounts for the four-way split per group.

\begin{algorithm}[tbp]
\caption{End-to-end evaluation of bias detection and mitigation\label{alg:evaluation}}
\begin{algorithmic}[1]
    \State \textbf{Inputs:} 
    \begin{itemize}
        \item Pre-fitted CATE model $f \colon X \mapsto \widehat \tau^{f}(X)$
        \item RCT data $\mathcal D_{g} = \{(X_i, T_i, Y_i)\}^{N_g}_{i=1}$ and $\mathcal D^{\textrm{holdout}}_{g} = \{(X_{j}, T_{j}, Y_{j})\}^{M_g}_{j=1}$ per $g \in \mathcal G$.
        \item Collapsibility weights $(\widehat W_{ig})$ and $(\widehat W_{jg})$ for each $i \in [N_g]$, $j \in [M_g]$, $g \in \mathcal G$.
        \item Significance level $\alpha \in (0,1)$, possibly with Bonferroni correction.
        \item Set of mitigation strategies $\mathcal{S}$ defining correction factors $\gamma_g$; i.e, $s \colon \widehat B_g \mapsto \gamma_g$ for each strategy $s \in \mathcal S$.
    \end{itemize}
    \State \textbf{Outputs:} Group bias estimates $\widehat B_g$ and cross-group bias estimates $\widehat B_g - \widehat B_{-g}$ before and after mitigation, along with standard errors and test decisions.
    
    \For{$g \in \mathcal G$} \Comment{Bias detection}
        \State $\widehat \tau^{f}_g \leftarrow N^{-1}_g \sum_i \widehat W_i \widehat \tau^{f}(X_i)$.
        \State $\widehat \tau_g \leftarrow$ scale-appropriate contrast-in-means estimator on $\mathcal D_g$ (i.e., Eq. \eqref{eq:est_bias_diff} or Eq. \eqref{eq:est_bias_ratio}).
        \State $\widehat B_g \leftarrow \widehat \tau^{f}_g - \widehat \tau_g$.
        \State $\widehat{\sigma}^2_g \leftarrow \widehat{\textrm{Var}}(\widehat B_g)$.
        \State Compute $\widetilde B_g$ and $\widetilde{\sigma}^2_{\widetilde B_g}$ by repeating steps 4--7 with on $\mathcal D^{\textrm{holdout}}_g$.
        \Comment{Bias mitigation}
        \For{Strategy $s \in \mathcal S$} 
            \State Estimate correction factor: $\widehat \gamma_g^{(s)} \leftarrow s(\widehat B_g)$.
            \State $\widehat B^{\widehat \gamma}_g \leftarrow \widetilde B_g - \gamma_g^{(s)} \widehat B_g$.
            \State $\widehat{\sigma}^2_{\widehat B^{\widehat \gamma}_g} \leftarrow \widehat{\textrm{Var}}(\widehat B^{\widehat \gamma}_g)$.
        \EndFor

        \State Repeat steps 4--12 on $\mathcal D_{-g}$ and $\mathcal D^{\textrm{holdout}}_{-g}$, where $\mathcal D_{-g} = \cup_{k \neq g} \mathcal D_k$ and $\mathcal D^{\textrm{holdout}}_{-g} = \cup_{k \neq g} \mathcal D^{\textrm{holdout}}_k$.
        \State Compute $\widehat B_g - \widehat B_{-g}$ and $\widehat{\sigma}^2_{\widehat B_g - \widehat B_{-g}}$.
        \Comment{Execute statistical tests}
        \State Test the following null hypotheses at significance level $\alpha$:
        \begin{itemize}
            \item $H_0 \colon b_g=0$
            \item $H_0 \colon b^\gamma_g=0$
            \item $H_0 \colon b_g = b_{-g}$
            \item $H_0 \colon b^\gamma_g = b^\gamma_{-g}$ 
        \end{itemize}
    \EndFor
    \State Return $(\widehat B_g, \widehat B_g - \widehat B_{-g}, \widehat B^{\widehat \gamma}_g, \widehat B^{\widehat \gamma}_g - \widehat B^{\widehat \gamma}_{-g})$ with standard errors and test decisions for all $g \in \mathcal G$ and $s \in \mathcal S$.
\end{algorithmic}
\end{algorithm}

\clearpage
\newpage
\section{Simulation Study}\label{app:simulation}

We evaluate finite-sample performance via a simulation study. We consider relative treatment effects on binary outcomes, where collapsibility requires weighting.

\textbf{Data-generating process.}  
We generate covariates representing user-level features commonly observed in marketing (e.g., click-through rates, dwell times) as $X_{i1} \sim \text{Beta}(2,18)$, $X_{i2} \sim \text{Gamma}(2,0.2)$, and $X_{i3} \sim \text{TruncNormal}(0.05,0.1)$. Potential outcome success probabilities are modelled as $\Pr[Y_i(T_i)=1 \mid X_i] = \text{expit}(\eta_i(T_i; X_i))$, with treatment $T_i \sim \text{Bernoulli}(1/2)$, inverse logit $\text{expit}(\cdot)$, and linear predictors
\begin{align}
    \eta_i(0; X_i) &= 0.1 + \zeta_g (0.5 X_{i1} + 0.25 X_{i1}^2 + 0.3 X_{i2} + 0.2 X_{i2} X_{i3}), \\
    \eta_i(1; X_i) &= \eta_i(0) \cdot (1 + |\zeta_g(0.75 X_{i1} + 0.9 X_{i2} + 1.2 X_{i3})|).
\end{align}
The parameter $\zeta_g$ scales covariate-dependent heterogeneity in baseline success and treatment effects across groups.

We construct generic CATE predictions as $\widehat{\tau}^{f}(X_i) = \tau(X_i) + \beta_g + \varepsilon_{ig}$, where $\tau(X_i)=\E[Y_i(1)\mid X_i]/\E[Y_i(0)\mid X_i]$ is the true relative CATE, $\beta_g$ represents systematic bias, and $\varepsilon_{ig}\sim \mathcal{N}(0,\varrho_g)$ denotes estimation noise scaled by the within-group variance of observed outcomes. Observed binary outcomes are drawn as $Y_i\sim\text{Bernoulli}(\Pr[Y_i(T_i)=1])$.

\textbf{Sampling of data.}  
We generate a population of one million observations from the data-generating process and compute the true and predicted GATEs by aggregating CATEs within groups using the correct collapsible weights on the relative scale. We then sample $N \in \{5000, 50000\}$ observations with group proportions fixed at $(0.45, 0.20, 0.15, 0.12, 0.08)$, implying average group sizes $N_g \in \{1000, 10000\}$ across two sample-size scenarios. This design introduces three forms of heterogeneity: (i) groups differ in total sample size, (ii) treatment and control counts vary randomly within groups since treatment assignment is independent of group, and (iii) the proportion of data used for estimation versus prediction differs across groups. Together, these create realistic variation in detection precision across groups that challenges the methods.

\textbf{Estimation and evaluation.}
We follow the end-to-end procedure in Section~\ref{sec:evaluation}, using the four-way split of Algorithm~1. For each sample size, we randomly split the data within each group into detection and mitigation halves. Within both halves, we further split into prediction and estimation subsets to ensure independence between predicted and experimental GATEs. The proportion allocated to estimation varies by group as $(0.55, 0.35, 0.30, 0.25, 0.50)$, inducing heterogeneous estimation uncertainty and thereby unequal detection precision. The outer split by detection and mitigation ensures that we evaluate performance strictly out-of-sample, while the inner split cancels covariance in detection statistics and ensures unbiased post-mitigation evaluation.

Experimental GATEs are estimated by the ratio of mean outcomes among treated and control units, and predicted CATEs are collapsed using the estimated collapsible weights for the relative estimand. We also evaluate the alternative estimator that relies only on positive outcomes $(Y_i=1)$, described in Appendix~\ref{app:retrospective_estimator} and used in our empirical application at Booking.com. We use significance level $0.05$ for the detection test and $Z = 1000 - 1$ bootstrap resamples per group, where the minus 1 leads to exact inference \citep{Wilcox2010}. 

\textbf{Benchmarks.}
We include four regression-based calibration methods as benchmarks, adapting and extending the method of \citet{Leng2024} for use in our framework; see Appendix \ref{app:ld-to-gamma} for details. In doing so, we obtain four benchmarks: affine, log-affine, isotonic, and log-isotonic calibration.

\textbf{Performance metrics.}  
We evaluate mitigation performance using the root mean-square and mean absolute residual bias, and their cross-group differences, across groups:
\begin{align}\label{eq:rmse_rmsed}
    \textrm{RMSE} &= \sqrt{\frac{1}{|\mathcal G|} \sum_{g \in \mathcal G} (b^{\gamma_g}_g)^2}, &
    \quad
    \textrm{RMSED} &= \sqrt{\frac{1}{|\mathcal G|} \sum_{g \in \mathcal G} (b^{\gamma_g}_g - b^{\gamma_{-g}}_{-g})^2},\\
    \label{eq:mae_maed}
    \textrm{MAE}  &= \frac{1}{|\mathcal G|} \sum_{g \in \mathcal G} |b^{\gamma_g}_g|, &
    \quad
    \textrm{MAED} &= \frac{1}{|\mathcal G|} \sum_{g \in \mathcal G} |b^{\gamma_g}_g - b^{\gamma_{-g}}_{-g}|.
\end{align}
RMSE and RMSED capture centrality and dispersion of residual group bias (sensitive to large deviations), whereas MAE and MAED capture the average absolute magnitude of remaining bias (robust to outliers). 

We also compute absolute residual bias per group and report the min and max as a measure of the worst and best performance. We evaluate these in terms of the estimated residual bias, producing the observable $\min_{g\in \mathcal G} |\widehat B^{\gamma}_g|$ $\max_{g\in \mathcal G}|\widehat B^{\gamma}_g|$, and in terms of the true residual bias, yielding the population values $\min_{g\in \mathcal G} |b^{\gamma}_g|$ $\max_{g\in \mathcal G}|b^{\gamma}_g|$. Intuitively, these metrics tells us how the mitigation strategies impact the the tail behavior of the distribution of group bias. 

We compute each metric against both true and estimated GATEs post-mitigation, and additionally report the percentage change in each metric relative to no debiasing. This allows us to assess both the empirically observable and the true but unobservable bias reduction, and check their alignment. All performance metrics are computed using 5-fold cross-validation on the mitigation split to ensure stability.

\textbf{Results.}
We first assess performance visually before diving into detailed results on our metrics. Figure~\ref{fig:sim_N1000} and~\ref{fig:sim_N10000} show calibration performance across the two sample sizes and whether true bias is present or not. Overall, we see that all methods tend to reduce the bias for most groups, whether true bias is present or not, and across the smaller and larger sample sizes. Debiasing w.r.t. mean error leads to the least change in group bias when none is truly present; see Figure~\ref{fig:sim_N1000} and Figure~\ref{fig:sim_N10000}(f). The regression calibration methods perform well when the per-group sample sizes are small, as they pool information across (Figure~\ref{fig:sim_N1000}). When we increase the sample sizes per group, they instead perform worse than our per-group shrinkage methods. This follows from that each group then carry enough signal to direct the debiasing group-wise and the calibration mixes the heterogeneity via its pooling.  

We now comment detailed results on our performance metrics. Tables~\ref{tab:sim-collapsible} summarize the simulation for the across-group, distributional metrics. Two main patterns emerge. 

First, performance improves with sample size. Increasing $N_g$ from $1{,}000$ to $10{,}000$ reduces the remaining bias across all metrics and methods, reflecting that larger samples yield more precise bias detection and, therefore, more effective mitigation.

Second, the relative ranking of methods depends on whether bias is truly present. When bias truly exists, all methods achieve large reductions relative to no debiasing, with reductions exceeding 90\% for both RMSE and MAE at large $N_g$. Under smaller samples, the risk-minimizing methods (MSE$^+$, MSE$^-$) outperform the naïve and mean-error strategies, which tend to overcorrect under detection noise. The log-affine and isotonic regression-calibration benchmarks also perform competitively when bias is truly present but are outperformed by our group-wise risk-minimizing strategies in all scenarios when there is no bias. This reflects that the mean error strategy applies full debiasing when detection per group is statistically significant, that the MSE$^+$ and MSE$^-$ strategies adapt to group-wise statistical uncertainty, whereas the regression calibrators can inadvertently increase bias when no global miscalibration exists because they apply a common transformation across groups.

Table~\ref{tab:sim-bestworst-collapsible} show the best and worst debiasing performance among the groups, as quantified by the min and max of the absolute residual bias. The results largely confirm the former. For instance, when sample sizes are large, debiasing per group yields a lower maximum increase and greater maximum reduction in group bias than using regression calibration. 

Finally, across all sets of tables, the empirical and population metrics tell a consistent story. Methods that achieve the greatest reduction in true residual bias also minimize the empirical (hatted) metrics, and the discrepancy between the two shrinks with larger $N_g$. This confirms the theoretical guarantees for empirical applications.

\textbf{Extensions.} We also evaluate a practical extension of our framework. Specifically, we assess performance using a \emph{converted-only} estimator for collapsing predicted CATEs to GATEs. This estimator, described in Appendix~\ref{app:retrospective_estimator}, relies only on positive binary outcomes $(Y=1)$ and is the default method for collapsing CATE in our empirical application at Booking.com. It separately estimates mean CATE for treated and control observations and then combines these into an overall ATE. The method is efficient when outcomes are binary and positive cases are rare, difficult to record, or costly to record or retain. For instance, users on online platforms typically only engage with a small share of the items they are exposed to, and in many marketing and healthcare settings (e.g., scanner data and electronic health record), observations are only collected conditional on the ``positive'' outcome (e.g., in-store check-out or health concern). Tables \ref{tab:sim-converted} and \ref{tab:sim-bestworst-converted} show the results using when we apply this estimator per group during detection. In short, it yields results similar to those in Tables~\ref{tab:sim-collapsible} and \ref{tab:sim-bestworst-collapsible}, confirming that our framework generalizes to other methods for collapsing CATE.


\clearpage
\begin{figure}[htbp!]
  \FIGURE{
    \centering
    \begin{tabular}{cc}
    \subfigure[$N_g \approx 1000$ on avg., with true bias]{
      \includegraphics[height=2.25in]{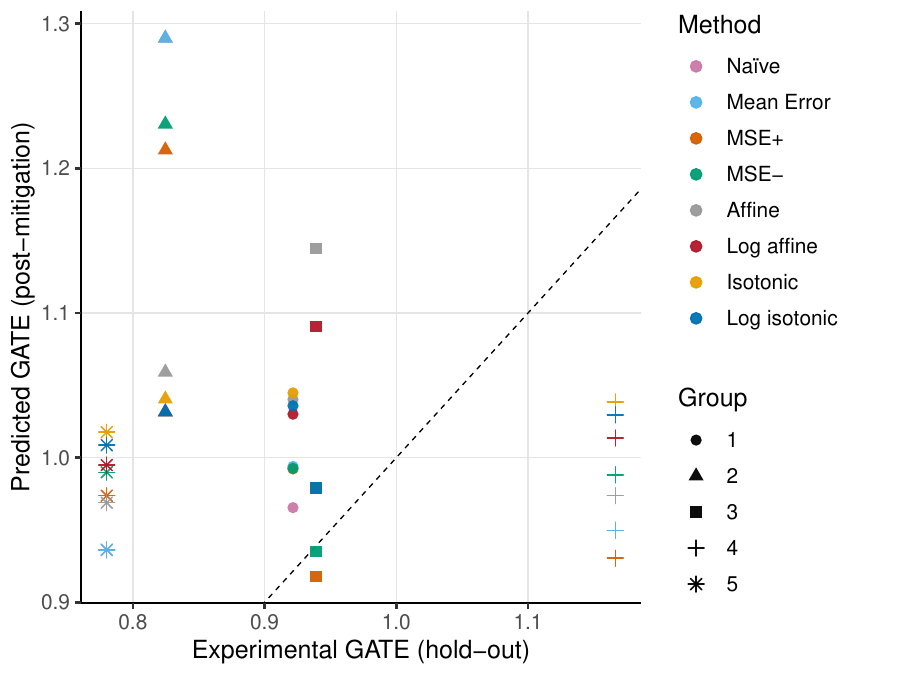}
    }
    &
    \subfigure[$N_g \approx 1000$ on avg., without true bias]{
      \includegraphics[height=2.25in]{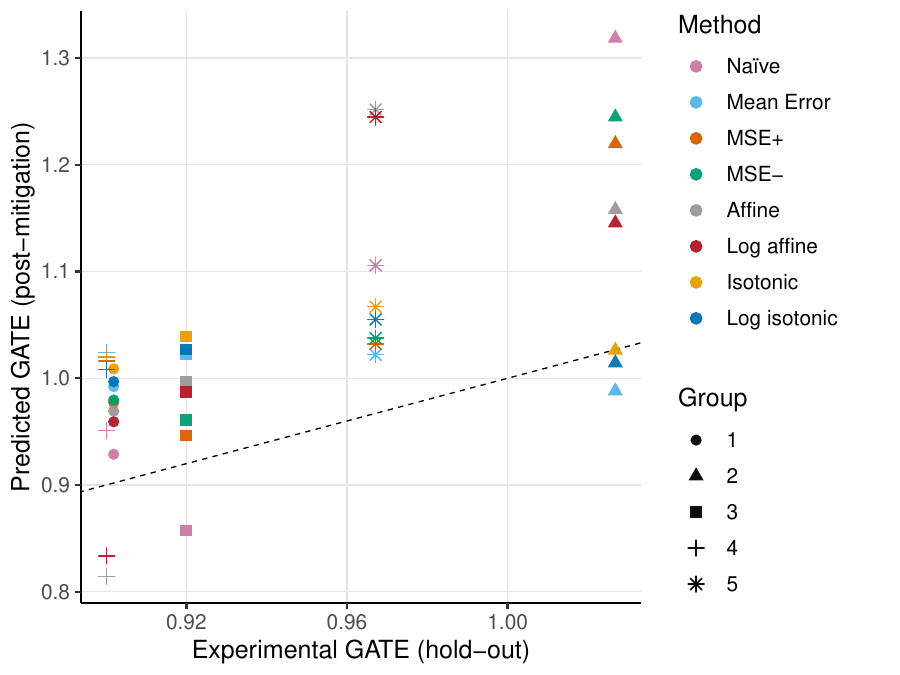}
    }
    \\[0.35cm]
    \subfigure[$N_g \approx 1000$ on avg., with true bias]{
      \includegraphics[height=2.25in]{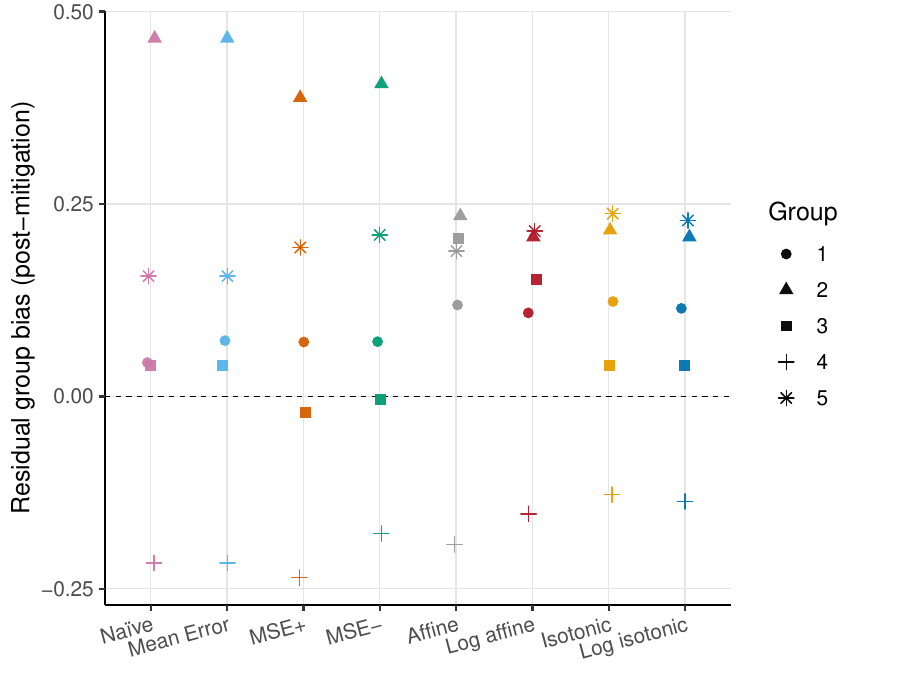}
    }
    &
    \subfigure[$N_g \approx 1000$ on avg., without true bias]{
      \includegraphics[height=2.25in]{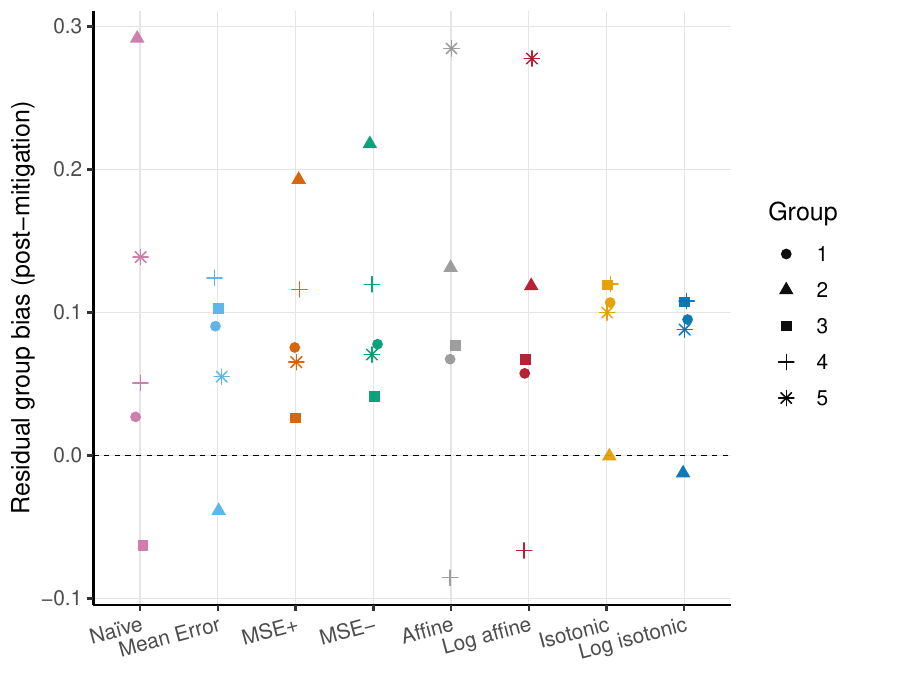}
    }
    \\[0.35cm]
    \subfigure[$N_g \approx 1000$ on avg., with true bias]{
      \includegraphics[height=2.25in]{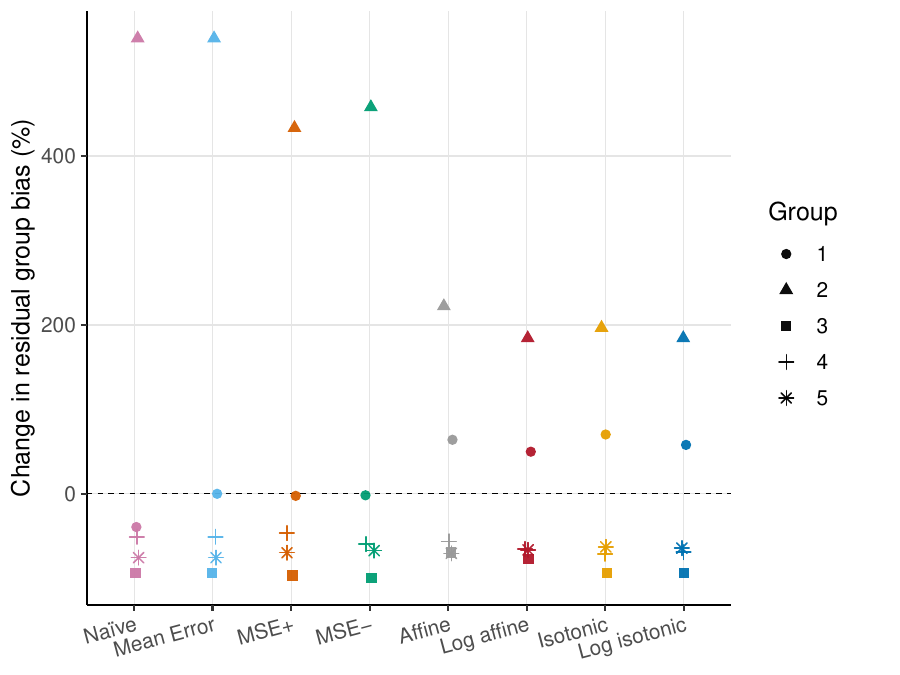}
    }
    &
    \subfigure[$N_g \approx 1000$ on avg., without true bias]{
      \includegraphics[height=2.25in]{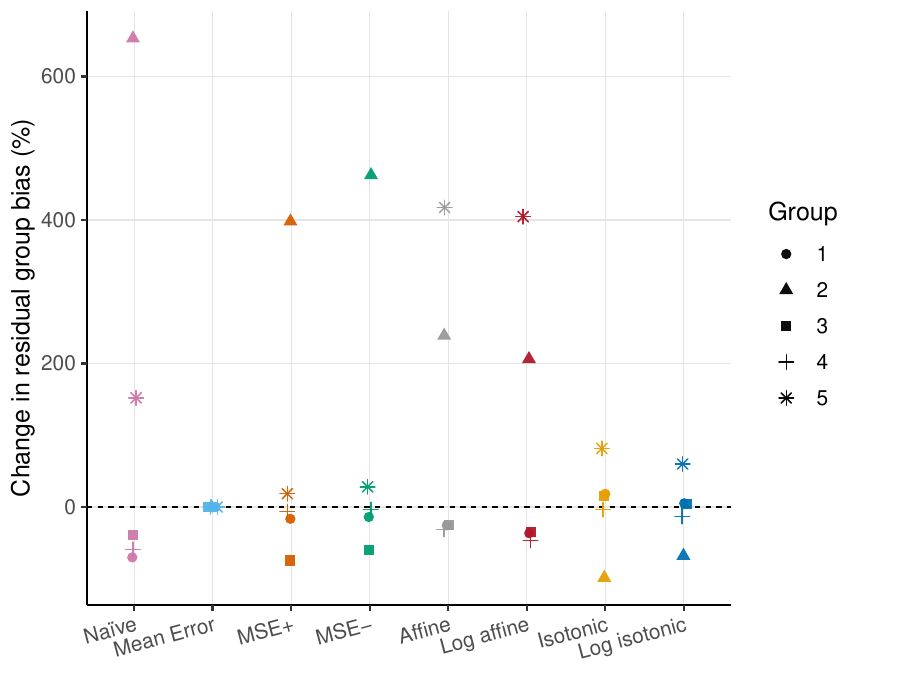}
    }
    \end{tabular}
    }
  {
  Simulation performance across debiasing strategies for small group sizes ($N_g \approx 1,000$ on average)
  \label{fig:sim_N1000}
  }
  {
  \scriptsize
  Left panels correspond to settings when true bias is present in the data-generating process; right panels show results when it is not. Panels (a) and (b) plot model-implied versus experimental GATEs estimated on independent hold-out data, with the dashed line indicating perfect alignment. Panels (c) and (d) directly plot the resulting residual group bias $\widehat B_g^\gamma$, defined as the difference between the model-implied and experimental GATEs. Panels (e) and (f) plot the percentage change in that residual group bias relative to no debiasing. Shapes denote groups and colors denote the mitigation strategy.
  }
\end{figure}

\clearpage
\begin{figure}[htbp!]
  \FIGURE{
\centering
\begin{tabular}{cc}
\subfigure[$N_g \approx 10000$ on avg., with true bias]{
  \includegraphics[height=2.25in]{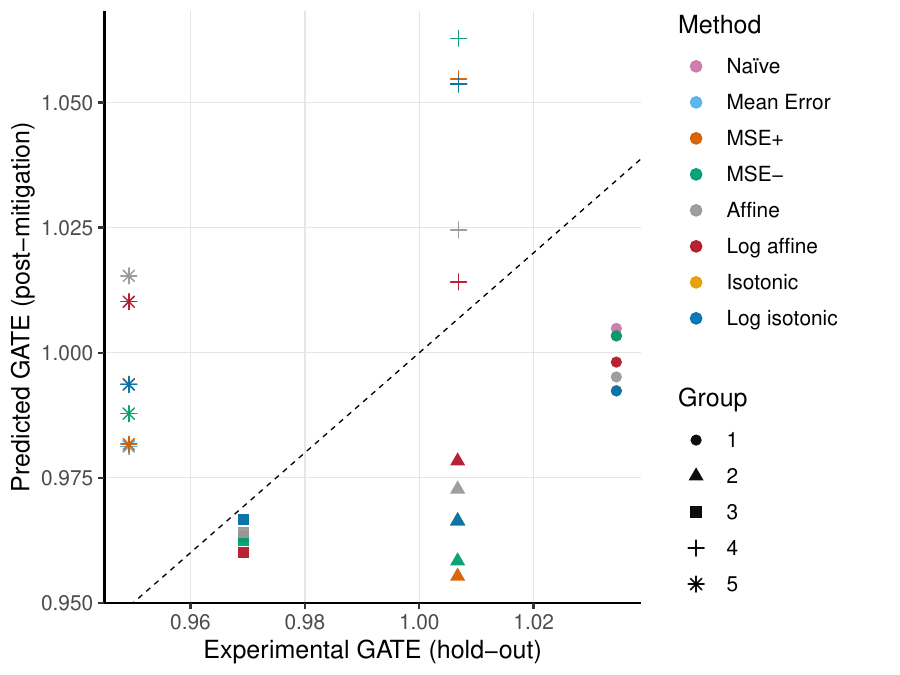}
}
&
\subfigure[$N_g \approx 10000$ on avg., without true bias]{
  \includegraphics[height=2.25in]{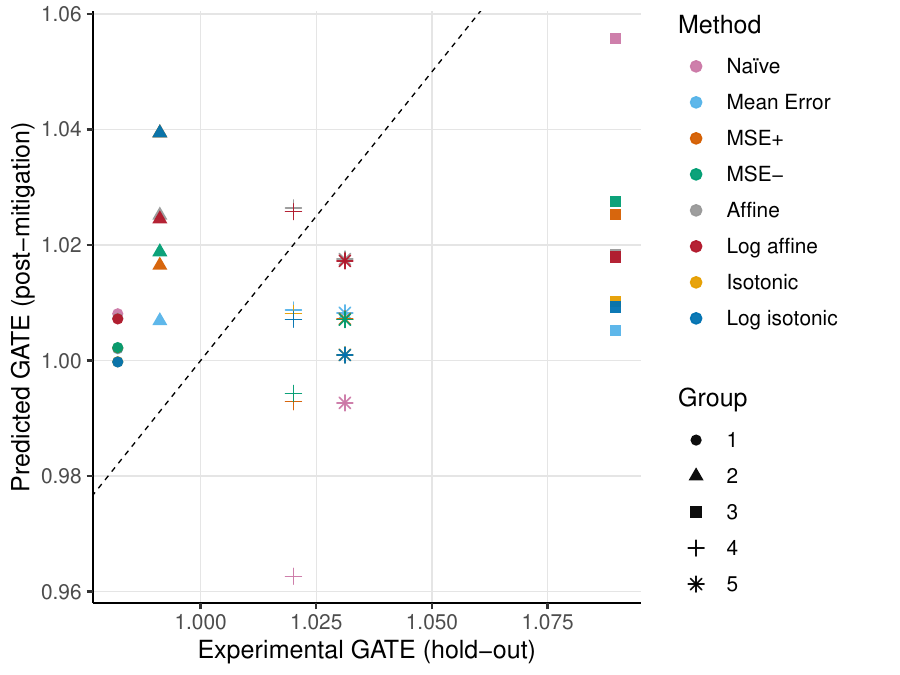}
}
\\[0.35cm]
\subfigure[$N_g \approx 10000$ on avg., with true bias]{
  \includegraphics[height=2.25in]{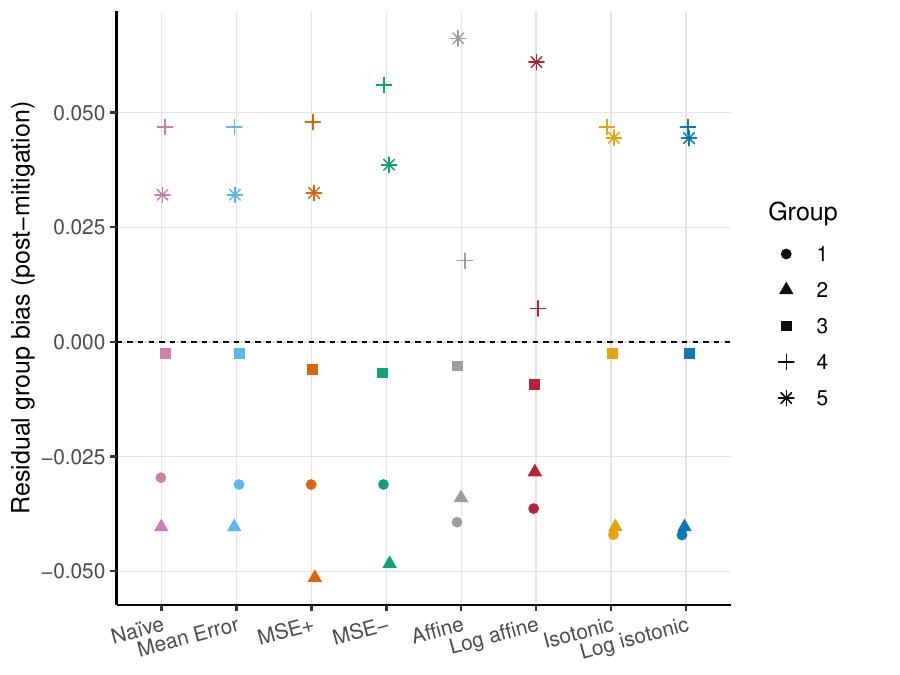}
}
&
\subfigure[$N_g \approx 10000$ on avg., without true bias]{
  \includegraphics[height=2.25in]{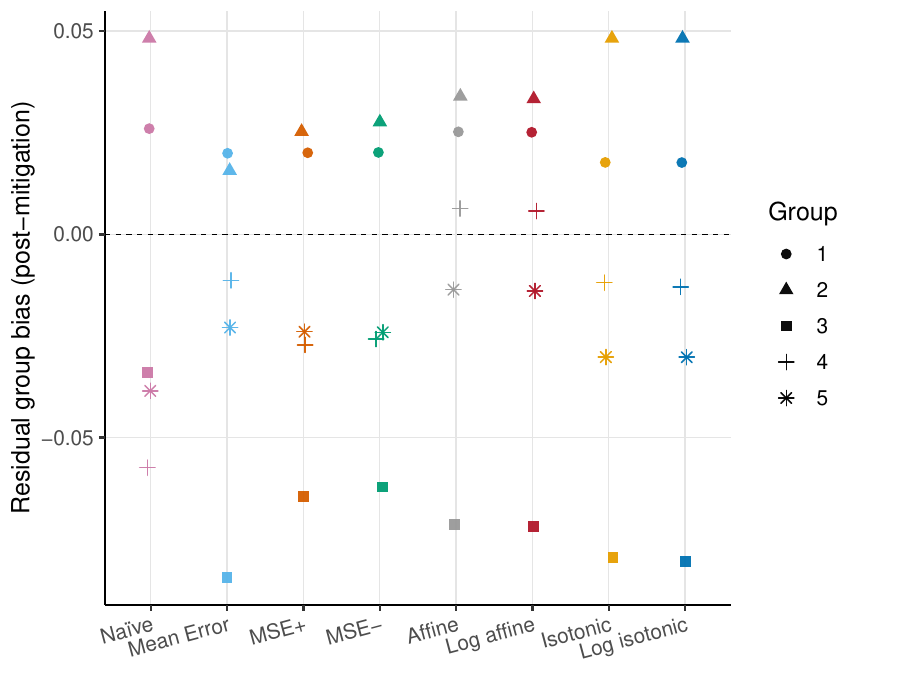}
}
\\[0.35cm]
\subfigure[$N_g \approx 10000$ on avg., with true bias]{
  \includegraphics[height=2.25in]{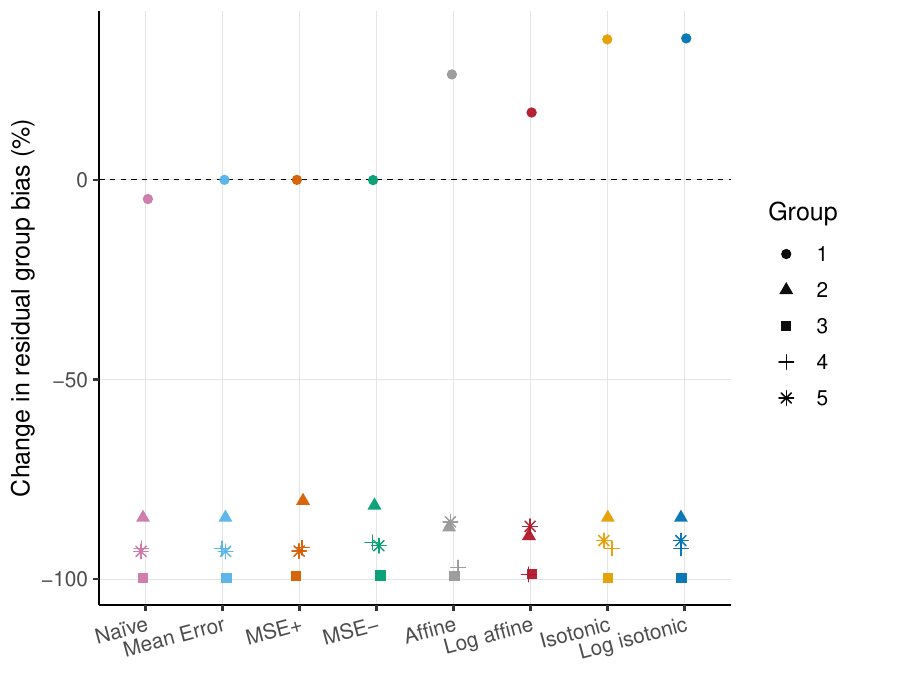}
}
&
\subfigure[$N_g \approx 10000$ on avg., without true bias]{
  \includegraphics[height=2.25in]{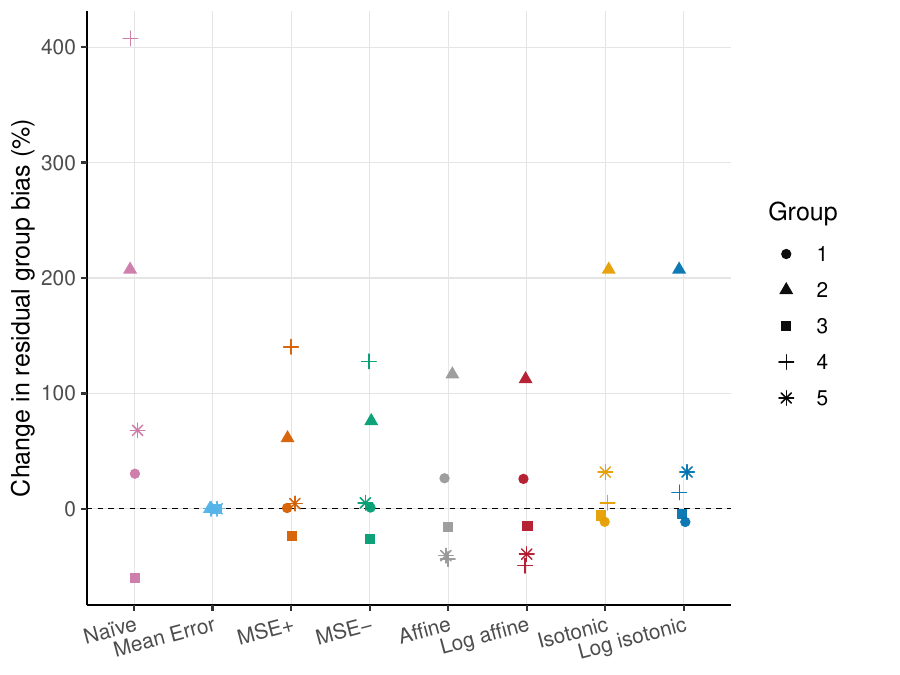}
}
\end{tabular}
}
  {
  Simulation performance across debiasing strategies for large group sizes ($N_g \approx 10,000$ on average)
    \label{fig:sim_N10000}
  }
    {
    \scriptsize
    Left panels correspond to settings when true bias is present in the data-generating process); right panels show results when it is not. Panels (a) and (b) plot model-implied versus experimental GATEs estimated on independent hold-out data, with the dashed line indicating perfect alignment. Panels (c) and (d) directly plot the resulting residual group bias $\widehat B_g^\gamma$, defined as the difference between the model-implied and experimental GATEs. Panels (e) and (f) plot the percentage change in that residual group bias relative to no debiasing. Shapes denote groups and colors denote the mitigation strategy.
    For each of the four regression benchmarks, we ``back out'' 
    }
\end{figure}

\begin{landscape}

\clearpage
\begin{table}[htbp!]
\centering
\TABLE
{Simulation results of distributional performance measured via root mean-square and mean absolute debiasing error, with the mean across groups.\label{tab:sim-collapsible}
}
{
\begin{tabular}{l c l l l l l l l l l}
\toprule
Bias & $N_g$ & Strategy  & RMSE  &  $\widehat{\textrm{RMSE}}$ & RMSED    &  $\widehat{\textrm{RMSED}}$ & MAE  &  $\widehat{\textrm{MAE}}$ & MAED    &  $\widehat{\textrm{MAED}}$ \\
\midrule
\multirow{16}{*}{Yes}
  & \multirow{8}{*}{1{,}000}
    & Na{\"i}ve & .144 (-69\%) & .269 (-43\%) & .195 (-67\%) & .332 (-43\%) & .097 (-76\%) & .184 (-52\%) & .093 (-66\%) & .157 (-49\%) \\
  & & Mean Error& .143 (-69\%) & .271 (-42\%) & .192 (-68\%) & .334 (-42\%) & .091 (-77\%) & .190 (-50\%) & .096 (-65\%) & .151 (-51\%) \\
  & & MSE$^+$   & .114 (-76\%) & .246 (-48\%) & .156 (-74\%) & .305 (-47\%) & .084 (-79\%) & .182 (-52\%) & .062 (-77\%) & .136 (-56\%) \\
  & & MSE$^-$   & .113 (-76\%) & .248 (-47\%) & .152 (75\%) & .299 (-48\%) & .070 (-83\%) & .174 (-54\%) & .077 (-72\%) & .136 (-56\%) \\
  & & Affine    & .085 (-82\%) & .194 (-62\%) & .099 (-83\%) & .199 (-69\%) & .062 (-85\%) & .188 (-51\%) & .039 (-86\%) & .035 (-89\%) \\
  & & Log Affine    & .055 (-88\%) & .186 (-64\%) & .061 (-90\%) & .186 (-71\%) & .031 (-92\%) & .167 (-56\%) & .028 (-90\%) & .044 (-86\%) \\
  & & Isotonic  & .038 (-92\%) & .189 (-63\%) & .042 (-93\%) & .192 (-70\%) & .027 (-93\%) & .149 (-61\%) & .011 (-96\%) & .078 (-75\%) \\
  & & Log Isotonic  & .035 (-93\%) & .185 (-64\%) & .040 (-93\%) & .190 (-70\%) & .020 (-95\%) & .145 (-62\%) & .011 (-96\%) & .072 (-76\%) \\
\cmidrule(lr){2-11}
  & \multirow{8}{*}{10{,}000}
    & Na{\"i}ve & .032 (-93\%) & .062 (-87\%) & .038 (-94\%) & .071 (-88\%) & .031 (1910\%) & .041 (32\%) & .020 (1297\%) & .012 (-55\%) \\
  & & Mean Error& .032 (-93\%) & .063 (-87\%) & .038 (-94\%) & .072 (-88\%) & .002 (0\%) & .031 (0\%) & .001 (0\%) & .027 (0\%) \\
  & & MSE$^+$   & .035 (-93\%) & .065 (-87\%) & .042 (-93\%) & .075 (-88\%) & .011 (586\%) & .032 (4\%) & .009 (563\%) & .016 (-40\%) \\
  & & MSE$^-$   & .036 (-92\%) & .066 (-86\%) & .045 (-93\%) & .077 (-87\%) & .011 (629\%) & .032 (4\%) & .009 (552\%) & .015 (-44\%) \\
  & & Affine    & .022 (-95\%) & .066 (-86\%) & .025 (-96\%) & .075 (-88\%) & .020 (-95\%) & .032 (-92\%) & .018 (-93\%) & .021 (-92\%) \\
  & & Log Affine    & .021 (-95\%) & .064 (-87\%) & .022 (-96\%) & .072 (-88\%) & .016 (-96\%) & .028 (-93\%) & .021 (-92\%) & .020 (-92\%) \\
  & & Isotonic  & .030 (-94\%) & .069 (-86\%) & .037 (-94\%) & .081 (-87\%) & .030 (-93\%) & .035 (-91\%) & .016 (-94\%) & .016 (-94\%) \\
  & & Log Isotonic  & .030 (-94\%) & .069 (-86\%) & .037 (-94\%) & .080 (-87\%) & .030 (-93\%) & .035 (-91\%) & .016 (-94\%) & .016 (-94\%) \\
\midrule
\multirow{16}{*}{No}
  & \multirow{8}{*}{1{,}000}
    & Na{\"i}ve & .171 (639\%) & .241 (25\%) & .229 (739\%) & .291 (35\%) & .139 (860\%) & .114 (39\%) & .091 (2191\%) & .101 (186\%) \\
  & & Mean Error& .023 (-1\%) & .193 (0\%) & .027 (-1\%) & .217 (1\%) & .014 (0\%) & .082 (0\%) & .004 (0\%) & .035 (0\%) \\
  & & MSE$^+$   & .103 (347\%) & .212 (10\%) & .135 (396\%) & .243 (13\%) & .065 (349\%) & .095 (16\%) & .074 (1755\%) & .059 (68\%) \\
  & & MSE$^-$   & .112 (387\%) & .219 (14\%) & .144 (427\%) & .250 (16\%) & .068 (372\%) & .105 (28\%) & .085 (2027\%) & .063 (79\%) \\
  & & Affine & .162 (693\%) & .221 (5\%) & .196 (733\%) & .245 (8\%) & .127 (780\%) & .129 (57\%) & .106 (2556\%) & .079 (123\%) \\
  & & Log Affine  & .153 (651\%) & .212 (1\%) & .185 (688\%) & .235 (4\%) & .123 (753\%) & .117 (43\%) & .092 (2214\%) & .081 (128\%) \\
  & & Isotonic & .036 (75\%) & .203 (-4\%) & .031 (33\%) & .220 (-3\%) & .025 (71\%) & .089 (9\%) & .021 (435\%) & .044 (25\%) \\
  & & Log Isotonic & .030 (45\%) & .209 (-1\%) & .031 (33\%) & .223 (-1\%) & .017 (20\%) & .082 (0\%) & .017 (321\%) & .035 (-1\%) \\
\cmidrule(lr){2-11}
  & \multirow{8}{*}{10{,}000}
    & Na{\"i}ve & .035 (455\%) & .070 (3\%) & .045 (542\%) & .083 (7\%) & .031 (1910\%) & .041 (32\%) & .020 (1297\%) & .012 (-55\%) \\
  & & Mean Error& .006 (0\%) & .068 (0\%) & .007 (0\%) & .077 (0\%) & .002 (0\%) & .031 (0\%) & .001 (0\%) & .027 (0\%) \\
  & & MSE$^+$   & .014 (116\%) & .065 (-4\%) & .017 (145\%) & .075 (-4\%) & .011 (586\%) & .032 (4\%) & .009 (563\%) & .016 (-40\%) \\
  & & MSE$^-$   & .014 (125\%) & .066 (-3\%) & .018 (154\%) & .076 (-2\%) & .011 (629\%) & .032 (4\%) & .009 (552\%) & .015 (-44\%) \\
  & & Affine    & .014 (121\%) & .072 (1\%) & .009 (38\%) & .083 (4\%) & .011 (627\%) & .030 (-2\%) & .005 (282\%) & .023 (-16\%) \\
  & & Log Affine & .014 (114\%) & .072 (1\%) & .009 (35\%) & .083 (4\%) & .011 (599\%) & .030 (-3\%) & .005 (266\%) & .023 (-16\%) \\  
  & & Isotonic  & .016 (157\%) & .076 (6\%) & .021 (202\%) & .088 (10\%) & .011 (602\%) & .037 (21\%) & .011 (644\%) & .026 (-2\%) \\
  & & Log Isotonic & .016 (155\%) & .076 (6\%) & .021 (204\%) & .088 (11\%) & .011 (603\%) & .038 (23\%) & .011 (643\%) & .026 (-1\%) \\
\bottomrule
\end{tabular}
}
{
\hspace{-0.6cm}
\begin{tabular}{p{21cm}}
\emph{Notes.}
Metrics without hats compute the residual bias with respect to true GATE and capture remaining true biases. Hatted metrics are computed with respect to estimates GATES and capture remaining empirical bias (cf. Eqs. \eqref{eq:rmse_rmsed})
Numbers in parentheses show the percentage change vs.\ if no debiasing was applied (negative = bias reduced, positive = bias increased).
Smaller values indicate more bias removed.
Percentages under the no-bias scenario appear big because the baseline was close to zero.
Rows labeled``Yes'' (``No'') correspond to data-generating processes with (without) systematic prediction bias.
$N_g$ is the average per-group sample size.
\end{tabular}
}
\end{table}

\clearpage
\newpage
\begin{table}[htbp!]
\centering
\TABLE{
Simulation results for most and least residual group bias across groups.\label{tab:sim-bestworst-collapsible}}
{
\setlength{\tabcolsep}{3pt}

\begin{tabular}{lllrrrrrrrrrrrrrrrr}
\toprule
Bias & $N_g$ & Strategy & $\max |\widehat B^\gamma|$ & $G$ & $\min |\widehat B^\gamma|$ & $G$ & $\max \Delta |\widehat B^\gamma| \%$  & $G$ & $\min \Delta |\widehat B^\gamma| \%$ & $G$ & $\max |b^\gamma|$ &$G$ & $\min |b^\gamma|$ & $G$ & $\max \Delta |b^\gamma| \%$ & $G$ & $\min \Delta |b^\gamma| \%$ &$G$ \\
\midrule
\multirow{16}{*}{Yes}
  & \multirow{8}{*}{1{,}000}
    & Na{\"i}ve & .282 & 2 & .026 & 3 & 310\% & 1 & -96\% & 3 & .465 & 2 & .040 & 3 & 539\% & 2 & -94\% & 3 \\
  & & Mean Error& .282 & 2 & .009 & 1 & 11\% & 2 & -96\% & 3 & .465 & 2 & .040 & 3 & 539\% & 2 & -94\% & 3 \\
  & & MSE$^+$   & .205 & 2 & .011 & 1 & 19\% & 1 & -91\% & 5 & .388 & 2 & .021 & 3 & 433\% & 2 & -97\% & 3 \\
  & & MSE$^-$   & .223 & 2 & .010 & 1 & 14\% & 1 & -96\% & 4 & .406 & 2 & .004 & 3 & 458\% & 2 & -99\% & 3 \\
  & & Affine    & .139 & 3 & .037 & 1 & 304\% & 1 & -93\% & 4 & .234 & 2 & .119 & 1 & 222\% & 2 & -71\% & 5 \\
  & & Log Affine& .086 & 3 & .001 & 4 & 193\% & 1 & -100\% & 4 & .215 & 5 & .109 & 1 & 184\% & 2 & -78\% & 3 \\
  & & Isotonic  & .042 & 1 & .007 & 5 & 353\% & 1 & -98\% & 5 & .237 & 5 & .040 & 3 & 197\% & 2 & -94\% & 3 \\
  & & Log Isotonic  & .033 & 1 & .002 & 5 & 256\% & 1 & -100\% & 5 & .229 & 5 & .040 & 3 & 184\% & 2 & -94\% & 3 \\
\cmidrule(lr){2-19}
  & \multirow{8}{*}{10{,}000}
    & Na{\"i}ve & .041 & 2 & .002 & 1 & 847\% & 1 & -95\% & 3 & .047 & 4 & .003 & 3 & -5\% & 1 & -100\% & 3 \\
  & & Mean Error& .041 & 2 & .000 & 1 & 0\% & 1 & -95\% & 3 & .047 & 4 & .003 & 3 & 0\% & 1 & -100\% & 3 \\
  & & MSE$^+$   & .052 & 2 & .000 & 1 & 1\% & 1 & -94\% & 3 & .051 & 2 & .006 & 3 & 0\% & 1 & -99\% & 3 \\
  & & MSE$^-$   & .050 & 4 & .000 & 1 & 11\% & 1 & -94\% & 5 & .056 & 4 & .007 & 3 & 0\% & 1 & -99\% & 3 \\
  & & Affine    & .041 & 3 & .005 & 5 & 4449\% & 1 & -99\% & 5 & .066 & 5 & .005 & 3 & 26\% & 1 & -99\% & 3 \\
  & & Log Affine& .045 & 3 & .000 & 5 & 2769\% & 1 & -100\% & 5 & .061 & 5 & .007 & 4 & 17\% & 1 & -99\% & 4 \\
  & & Isotonic  & .041 & 2 & .011 & 1 & 5993\% & 1 & -96\% & 5 & .047 & 4 & .003 & 3 & 35\% & 1 & -100\% & 3 \\
  & & Log Isotonic  & .041 & 2 & .011 & 1 & 6037\% & 1 & -96\% & 5 & .047 & 4 & .003 & 3 & 35\% & 1 & -100\% & 3 \\
\midrule
\multirow{16}{*}{No}
  & \multirow{8}{*}{1{,}000}
    & Na{\"i}ve & .311 & 2 & .062 & 4 & 1513\% & 2 & 456\% & 4 & .292 & 2 & .027 & 1 & 653\% & 2 & -70\% & 1 \\
  & & Mean Error& .019 & 2 & .011 & 1 & 0\% & 1 & 0\% & 1 & .124 & 4 & .039 & 2 & 0\% & 1 & 0\% & 1 \\
  & & MSE$^+$   & .212 & 2 & .003 & 4 & 1001\% & 2 & -72\% & 4 & .193 & 2 & .026 & 3 & 398\% & 2 & -75\% & 3 \\
  & & MSE$^-$   & .237 & 2 & .007 & 4 & 1130\% & 2 & -40\% & 4 & .218 & 2 & .041 & 3 & 463\% & 2 & -60\% & 3 \\
  & & Affine    & .243 & 5 & .009 & 3 & 1737\% & 5 & -51\% & 3 & .285 & 5 & .067 & 1 & 417\% & 5 & -31\% & 4 \\
  & & Log Affine& .236 & 5 & .018 & 3 & 1684\% & 5 & 4\% & 3 & .278 & 5 & .057 & 1 & 405\% & 5 & -46\% & 4 \\
  & & Isotonic  & .058 & 5 & .006 & 1 & 340\% & 5 & -48\% & 1 & .120 & 4 & .001 & 2 & 82\% & 5 & -99\% & 2 \\
  & & Log Isotonic  & .046 & 5 & .005 & 4 & 250\% & 5 & -64\% & 2 & .108 & 4 & .012 & 2 & 60\% & 5 & -68\% & 2 \\
\cmidrule(lr){2-19}
  & \multirow{8}{*}{10{,}000}
    & Na{\"i}ve & .051 & 3 & .005 & 1 & 160784\% & 3 & 328\% & 1 & .057 & 4 & .026 & 1 & 408\% & 4 & -60\% & 3 \\
  & & Mean Error& .004 & 4 & .000 & 3 & 0\% & 1 & 0\% & 1 & .084 & 3 & .011 & 4 & 0\% & 1 & 0\% & 1 \\
  & & MSE$^+$   & .020 & 3 & .001 & 1 & 63605\% & 3 & -11\% & 1 & .065 & 3 & .020 & 1 & 140\% & 4 & -24\% & 3 \\
  & & MSE$^-$   & .022 & 3 & .001 & 1 & 71089\% & 3 & -18\% & 1 & .062 & 3 & .020 & 1 & 128\% & 4 & -26\% & 3 \\
  & & Affine    & .018 & 2 & .004 & 1 & 41860\% & 3 & 235\% & 4 & .071 & 3 & .006 & 4 & 117\% & 2 & -44\% & 4 \\
  & & Log Affine& .017 & 2 & .004 & 1 & 40326\% & 3 & 220\% & 4 & .072 & 3 & .006 & 4 & 112\% & 2 & -49\% & 4 \\
  & & Isotonic  & .032 & 2 & .003 & 1 & 15978\% & 3 & 14\% & 4 & .079 & 3 & .012 & 4 & 207\% & 2 & -11\% & 1 \\
  & & Log Isotonic  & .032 & 2 & .003 & 1 & 12734\% & 3 & 40\% & 4 & .080 & 3 & .013 & 4 & 207\% & 2 & -11\% & 1 \\
\bottomrule
\end{tabular}
}
{
\hspace{-0.65cm}
\begin{tabular}{p{22cm}}
\emph{Notes.}
$|B^\gamma|$ denotes the absolute residual bias relative to the estimated GATE (measuring empirical residual bias), and $|b^\gamma|$ denotes the absolute residual bias relative to the true GATE (measuring true residual bias).
For each mitigation strategy, the minimum and maximum values are taken across groups; columns labeled $G$ report the group attaining that value.
$\Delta$ measures the percentage change in absolute bias relative to no debiasing (positive = increase/worse; negative = reduction/better).
Rows labeled``Yes'' (``No'') correspond to data-generating processes with (without) systematic prediction bias.
$N_g$ is the average per-group sample size.
\end{tabular}
}
\end{table}

\clearpage
\begin{table}[htbp!]
\centering
\TABLE{Simulation results of distributional performance measured via root mean-square and mean absolute debiasing error, with the mean across groups when using the converted-only estimator to collapse CATE.\label{tab:sim-converted}}
{
\begin{tabular}{l c l l l l l l l l l}
\toprule
Bias & $N_g$ & Strategy  & RMSE  &  $\widehat{\textrm{RMSE}}$ & RMSED    &  $\widehat{\textrm{RMSED}}$ & MAE  &  $\widehat{\textrm{MAE}}$ & MAED    &  $\widehat{\textrm{MAED}}$ \\
\midrule
\multirow{16}{*}{Yes}
  & \multirow{8}{*}{1{,}000}
    & Na{\"i}ve & .153 (-67\%) & .263 (-47\%) & .206 (-66\%) & .312 (-49\%) & .106 (-74\%) & .179 (-52\%) & .097 (-65\%) & .166 (-47\%) \\
  & & Mean Error& .152 (-68\%) & .264 (-47\%) & .203 (-66\%) & .311 (-49\%) & .100 (-75\%) & .184 (-51\%) & .100 (-63\%) & .161 (-49\%) \\
  & & MSE$^+$   & .122 (-74\%) & .242 (-51\%) & .167 (-72\%) & .288 (-53\%) & .093 (-77\%) & .177 (-53\%) & .064 (-77\%) & .139 (-56\%) \\
  & & MSE$^-$   & .121 (-74\%) & .245 (-50\%) & .163 (-73\%) & .281 (-54\%) & .079 (-80\%) & .169 (-55\%) & .080 (-71\%) & .137 (-56\%) \\
   & & Affine   & .088 (-81\%) & .213 (-56\%) & .103 (-83\%) & .232 (-61\%) & .064 (-84\%) & .182 (-51\%) & .038 (-86\%) & .045 (-86\%) \\
   & & Log Affine   & .060 (-87\%) & .197 (-59\%) & .065 (-89\%) & .210 (-65\%) & .036 (-91\%) & .161 (-57\%) & .025 (-91\%) & .046 (-85\%) \\
   & & Isotonic   & .051 (-89\%) & .194 (-60\%) & .059 (-90\%) & .210 (-65\%) & .031 (-92\%) & .143 (-62\%) & .010 (-96\%) & .081 (-74\%) \\
   & & Log Isotonic  & .048 (-90\%) & .192 (-60\%) & .057 (-91\%) & .210 (-65\%) & .027 (-93\%) & .139 (-63\%) & .007 (-98\%) & .075 (-76\%) \\
\cmidrule(lr){2-11}
  & \multirow{8}{*}{10{,}000}
    & Na{\"i}ve & .036 (-92\%) & .064 (-87\%) & .042 (-93\%) & .071 (-88\%) & .031 (-92\%) & .032 (-92\%) & .017 (-94\%) & .013 (-95\%) \\
  & & Mean Error& .036 (-92\%) & .064 (-87\%) & .042 (-93\%) & .071 (-88\%) & .031 (-92\%) & .033 (-92\%) & .017 (-94\%) & .013 (-95\%) \\
  & & MSE$^+$   & .039 (-92\%) & .066 (-86\%) & .045 (-93\%) & .073 (-88\%) & .034 (-92\%) & .036 (-91\%) & .020 (-93\%) & .015 (-95\%) \\
  & & MSE$^-$   & .040 (-92\%) & .067 (-86\%) & .048 (-92\%) & .075 (-87\%) & .034 (-92\%) & .038 (-91\%) & .023 (-91\%) & .016 (-94\%) \\
   & & Affine   & .027 (-94\%) & .069 (-86\%) & .029 (-95\%) & .077 (-87\%) & .022 (-94\%) & .035 (-92\%) & .017 (-94\%) & .020 (-93\%) \\
   & & Log Affine   & .026 (-95\%) & .066 (-86\%) & .025 (-96\%) & .073 (-88\%) & .018 (-95\%) & .031 (-93\%) & .020 (-93\%) & .020 (-93\%) \\
   & & Isotonic   & .034 (-93\%) & .071 (-85\%) & .041 (-93\%) & .081 (-87\%) & .031 (-92\%) & .037 (-91\%) & .018 (-94\%) & .015 (-94\%) \\
   & & Log Isotonic  & .035 (-93\%) & .071 (-85\%) & .041 (-93\%) & .082 (-86\%) & .031 (-92\%) & .037 (-91\%) & .018 (-94\%) & .015 (-94\%) \\
\midrule
\multirow{16}{*}{No}
  & \multirow{8}{*}{1{,}000}
    & Na{\"i}ve & .166 (401\%) & .264 (13\%) & .221 (457\%) & .316 (28\%) & .135 (593\%) & .106 (31\%) & .088 (657\%) & .103 (144\%) \\
  & & Mean Error& .033 (1\%) & .235 (0\%) & .040 (0\%) & .247 (0\%) & .002 (0\%) & .029 (0\%) & .001 (0\%) & .026 (0\%) \\
  & & MSE$^+$   & .099 (199\%) & .248 (6\%) & .128 (223\%) & .281 (14\%) & .063 (222\%) & .093 (15\%) & .074 (540\%) & .063 (49\%) \\
  & & MSE$^-$   & .108 (226\%) & .248 (6\%) & .136 (244\%) & .284 (15\%) & .066 (240\%) & .103 (28\%) & .085 (634\%) & .067 (58\%) \\
  & & Affine   & .161 (395\%) & .224 (2\%) & .191 (392\%) & .256 (9\%) & .126 (547\%) & .125 (55\%) & .098 (742\%) & .075 (78\%) \\
   & & Log Affine   & .153 (369\%) & .216 (-1\%) & .180 (366\%) & .246 (4\%) & .026 (36\%) & .088 (8\%) & .021 (80\%) & .043 (1\%) \\
   & & Isotonic   & .045 (38\%) & .215 (-2\%) & .041 (5\%) & .232 (-2\%) & .120 (513\%) & .114 (41\%) & .088 (661\%) & .077 (82\%) \\
   & & Log Isotonic & .039 (20\%) & .217 (-1\%) & .041 (6\%) & .234 (-1\%) & .020 (2\%) & .081 (0\%) & .017 (45\%) & .036 (-14\%) \\
\cmidrule(lr){2-11}
  & \multirow{8}{*}{10{,}000}
    & Na{\"i}ve & .036 (300\%) & .068 (-2\%) & .046 (368\%) & .079 (-1\%) & .031 (1166\%) & .039 (34\%) & .021 (1672\%) & .012 (-52\%) \\
  & & Mean Error& .009 (1\%) & .069 (0\%) & .010 (1\%) & .080 (0\%) & .002 (0\%) & .029 (0\%) & .001 (0\%) & .026 (0\%) \\
  & & MSE$^+$   & .015 (66\%) & .065 (-6\%) & .018 (89\%) & .075 (-6\%) & .011 (327\%) & .031 (5\%) & .011 (842\%) & .015 (-42\%) \\
  & & MSE$^-$   & .016 (73\%) & .066 (-4\%) & .019 (96\%) & .076 (-5\%) & .011 (353\%) & .030 (4\%) & .011 (826\%) & .014 (-46\%) \\
  & & Affine   & .014 (73\%) & .068 (3\%) & .009 (0\%) & .079 (3\%) & .013 (434\%) & .029 (-1\%) & .004 (273\%) & .021 (-18\%) \\
   & & Log Affine   & .014 (68\%) & .068 (3\%) & .009 (-1\%) & .079 (3\%) & .010 (307\%) & .036 (23\%) & .010 (773\%) & .025 (-2\%) \\
   & & Isotonic   & .015 (92\%) & .069 (4\%) & .019 (112\%) & .081 (6\%) & .013 (416\%) & .029 (-2\%) & .004 (253\%) & .021 (-18\%) \\
   & & Log Isotonic & .016 (93\%) & .070 (5\%) & .020 (116\%) & .082 (7\%) & .010 (308\%) & .036 (24\%) & .010 (772\%) & .025 (-1\%) \\
\bottomrule
\end{tabular}
}
{
\hspace{-0.6cm}
\begin{tabular}{p{21cm}}
\emph{Notes.}
Metrics without hats compute the residual bias with respect to true GATE and capture remaining true biases. Hatted metrics are computed with respect to estimates GATES and capture remaining empirical bias (cf. Eqs. \eqref{eq:rmse_rmsed})
Numbers in parentheses show the percentage change vs.\ if no debiasing was applied (negative = bias reduced, positive = bias increased).
Smaller values indicate more bias removed.
Percentages under the no-bias scenario appear big because the baseline was close to zero.
Rows labeled``Yes'' (``No'') correspond to data-generating processes with (without) systematic prediction bias.
$N_g$ is the average per-group sample size.
\end{tabular}
}
\end{table}

\clearpage
\begin{table}[htbp!]
\centering
\TABLE{
Simulation results for most and least residual group bias across groups, when using the converted-only estimator to collapse CATE.\label{tab:sim-bestworst-converted}}
{
\setlength{\tabcolsep}{3pt}

\begin{tabular}{lllrrrrrrrrrrrrrrrr}
\toprule
Bias & $N_g$ & Strategy & $\max |\widehat B^\gamma|$ & $G$ & $\min |\widehat B^\gamma|$ & $G$ & $\max \Delta |\widehat B^\gamma| \%$  & $G$ & $\min \Delta |\widehat B^\gamma| \%$ & $G$ & $\max |b^\gamma|$ &$G$ & $\min |b^\gamma|$ & $G$ & $\max \Delta |b^\gamma| \%$ & $G$ & $\min \Delta |b^\gamma| \%$ &$G$ \\
\midrule
\multirow{16}{*}{Yes}
  & \multirow{8}{*}{1{,}000}
    & Na{\"i}ve & .299 & 2 & .029 & 3 & 131\% & 1 & -96\% & 3 & .482 & 2 & .031 & 1 & 675\% & 2 & -95\% & 3 \\
  & & Mean Error& .299 & 2 & .022 & 1 & 22\% & 2 & -96\% & 3 & .482 & 2 & .037 & 3 & 675\% & 2 & -95\% & 3 \\
  & & MSE$^+$   & .221 & 2 & .023 & 1 & 8\% & 1 & -88\% & 3 & .404 & 2 & .024 & 3 & 549\% & 2 & -97\% & 3 \\
  & & MSE$^-$   & .240 & 2 & .016 & 4 & 6\% & 1 & -97\% & 4 & .422 & 2 & .007 & 3 & 578\% & 2 & -99\% & 3 \\
  & & Affine    & .137 & 3 & .024 & 1 & 13\% & 1 & -95\% & 4 & .249 & 2 & .106 & 1 & 300\% & 2 & -73\% & 5 \\
  & & Log Affine& .084 & 3 & .009 & 4 & -34\% & 1 & -98\% & 4 & .221 & 2 & .096 & 1 & 255\% & 2 & -78\% & 3 \\
  & & Isotonic  & .047 & 2 & .014 & 5 & 34\% & 1 & -97\% & 5 & .230 & 2 & .037 & 3 & 270\% & 2 & -95\% & 3 \\
  & & Log Isotonic  & .038 & 2 & .020 & 1 & -7\% & 1 & -96\% & 3 & .221 & 2 & .037 & 3 & 255\% & 2 & -95\% & 3 \\
\cmidrule(lr){2-19}
  & \multirow{8}{*}{10{,}000}
    & Na{\"i}ve & .045 & 4 & .000 & 1 & -64\% & 1 & -94\% & 3 & .051 & 4 & .007 & 3 & -5\% & 1 & -99\% & 3 \\
  & & Mean Error& .045 & 4 & .001 & 1 & 0\% & 1 & -94\% & 3 & .051 & 4 & .007 & 3 & 0\% & 1 & -99\% & 3 \\
  & & MSE$^+$   & .050 & 2 & .001 & 1 & 0\% & 1 & -94\% & 3 & .052 & 4 & .011 & 3 & 0\% & 1 & -99\% & 3 \\
  & & MSE$^-$   & .054 & 4 & .001 & 1 & -2\% & 1 & -95\% & 5 & .060 & 4 & .011 & 3 & 0\% & 1 & -98\% & 3 \\
  & & Affine    & .046 & 3 & .008 & 5 & 745\% & 1 & -98\% & 5 & .069 & 5 & .010 & 3 & 25\% & 1 & -99\% & 3 \\
  & & Log Affine& .050 & 3 & .003 & 5 & 476\% & 1 & -99\% & 5 & .064 & 5 & .012 & 4 & 16\% & 1 & -98\% & 4 \\
  & & Isotonic  & .045 & 4 & .012 & 1 & 993\% & 1 & -97\% & 5 & .051 & 4 & .007 & 3 & 34\% & 1 & -99\% & 3 \\
  & & Log Isotonic  & .045 & 4 & .012 & 1 & 1000\% & 1 & -97\% & 5 & .051 & 4 & .007 & 3 & 34\% & 1 & -99\% & 3 \\
\midrule
\multirow{16}{*}{No}
  & \multirow{8}{*}{1{,}000}
    & Na{\"i}ve & .310 & 2 & .059 & 4 & 7023\% & 5 & 234\% & 1 & .291 & 2 & .011 & 1 & 608\% & 2 & -86\% & 1 \\
  & & Mean Error& .033 & 3 & .001 & 5 & 0\% & 1 & 0\% & 1 & .128 & 4 & .041 & 2 & 0\% & 1 & 0\% & 1 \\
  & & MSE$^+$   & .211 & 2 & .007 & 4 & 874\% & 2 & -55\% & 4 & .192 & 2 & .042 & 3 & 366\% & 2 & -64\% & 3 \\
  & & MSE$^-$   & .236 & 2 & .010 & 4 & 1300\% & 5 & -31\% & 4 & .217 & 2 & .057 & 3 & 427\% & 2 & -52\% & 3 \\
  & & Affine    & .229 & 5 & .007 & 3 & 19285\% & 5 & -79\% & 3 & .270 & 5 & .051 & 1 & 530\% & 5 & -35\% & 4 \\
  & & Log Affine& .222 & 5 & .003 & 3 & 18696\% & 5 & -92\% & 3 & .264 & 5 & .041 & 1 & 513\% & 5 & -50\% & 4 \\
  & & Isotonic  & .049 & 3 & .010 & 4 & 3767\% & 5 & -61\% & 1 & .134 & 3 & .003 & 2 & 103\% & 5 & -93\% & 2 \\
  & & Log Isotonic  & .037 & 3 & .002 & 4 & 2768\% & 5 & -89\% & 4 & .123 & 3 & .015 & 2 & 76\% & 5 & -64\% & 2 \\
\cmidrule(lr){2-19}
  & \multirow{8}{*}{10{,}000}
    & Na{\"i}ve & .055 & 3 & .007 & 1 & 1503\% & 4 & 500\% & 1 & .056 & 4 & .028 & 1 & 447\% & 4 & -63\% & 3 \\
  & & Mean Error& .004 & 3 & .001 & 1 & 0\% & 1 & 0\% & 1 & .080 & 3 & .010 & 4 & 0\% & 1 & 0\% & 1 \\
  & & MSE$^+$   & .024 & 3 & .001 & 5 & 518\% & 4 & -63\% & 5 & .060 & 3 & .020 & 5 & 154\% & 4 & -25\% & 3 \\
  & & MSE$^-$   & .027 & 3 & .000 & 5 & 532\% & 3 & -72\% & 5 & .058 & 3 & .021 & 5 & 140\% & 4 & -28\% & 3 \\
  & & Affine    & .017 & 3 & .006 & 1 & 644\% & 2 & 313\% & 3 & .067 & 3 & .007 & 4 & 130\% & 2 & -48\% & 5 \\
  & & Log Affine& .017 & 3 & .006 & 1 & 614\% & 2 & 302\% & 3 & .068 & 3 & .007 & 4 & 125\% & 2 & -47\% & 5 \\
  & & Isotonic  & .030 & 2 & .001 & 1 & 1302\% & 2 & -14\% & 1 & .075 & 3 & .011 & 4 & 231\% & 2 & -10\% & 1 \\
  & & Log Isotonic  & .030 & 2 & .001 & 1 & 1302\% & 2 & -12\% & 1 & .076 & 3 & .012 & 4 & 231\% & 2 & -10\% & 1 \\
\bottomrule
\end{tabular}
}
{
\hspace{-0.65cm}
\begin{tabular}{p{22cm}}
\emph{Notes.}
$|B^\gamma|$ denotes the absolute residual bias relative to the estimated GATE (measuring empirical residual bias), and $|b^\gamma|$ denotes the absolute residual bias relative to the true GATE (measuring true residual bias).
For each mitigation strategy, the minimum and maximum values are taken across groups; columns labeled $G$ report the group attaining that value.
$\Delta$ measures the percentage change in absolute bias relative to no debiasing (positive = increase/worse; negative = reduction/better).
Rows labeled``Yes'' (``No'') correspond to data-generating processes with (without) systematic prediction bias.
$N_g$ is the average per-group sample size.
\end{tabular}
}
\end{table}

\end{landscape}

\clearpage
\newpage
\section{Additional Figures for the Booking.com Application}\label{app:distribution}

This appendix contains additional visualizations of group bias and cross-group bias that complement the scatterplots in Section~\ref{sec:booking_results}. Figure~\ref{fig:bias_dists} shows kernel densities of the empirical distributions of the bias estimates over the groups, where groups are defined by users’ country of origin. All estimates are standardized by the across-group mean and standard deviation to place the distributions on a comparable scale.

Three results stand out. First, the distribution of group bias is approximately symmetric and centered slightly below zero (Fig.~\ref{fig:bias_density}), indicating that, on average across countries, the model slightly underpredicts the experimental GATE. The dispersion is substantial, with nontrivial mass at two and even three standard deviations from zero. Second, although group bias is centered slightly negative, the distribution of cross-group bias is centered slightly positive (Fig.~\ref{fig:bias_disparity_density}). Thus, within-group bias and cross-group bias need not exhibit the same distributional properties. This shift reflects that a few countries are heavily overrepresented in the data. These countries contribute disproportionately to the pooled rest that forms the comparison group in the cross-group difference estimand, and typically exhibit more precise (and closer-to-zero) bias estimates due to their larger sample sizes. Consequently, for many countries the pooled-rest estimate $\widehat B_{-g}$ is smaller than their own bias estimate $\widehat B_g$, producing predominantly positive cross-group differences $\widehat B_g - \widehat B_{-g}$ and shifting the distribution upward.
Third, the cross-group bias estimates provide statistically significant evidence of systematic differences across groups. The empirical distribution of the test statistic closely tracks its theoretical null distribution, but with a positive shift and visible mass at conventional critical values (e.g., $\pm 1.64$, $\pm 1.96$, and $\pm 3$; Fig.~\ref{fig:t_stat_bias_dispariy_density}). This indicates that several countries exhibit significant cross-group bias at the 90\% or 95\% level, and a few at the 99.5\% level.

Figure~\ref{fig:bias_density_mitigated}--\ref{fig:bias_disparity_density_mitigated} shows the distributions after applying each debiasing strategy on the held-out data. All strategies improve calibration by shifting the distributions toward zero (corresponding to no average bias across groups) and reducing dispersion (corresponding to fewer groups with large bias). The risk-minimizing strategies outperform the naïve one, as they place more mass near zero and have less dispersion. Consistent with our theory, the MSE$^-$ and mean error strategies both outperform MSE$^+$ strategy, confirming that MSE$^+$ is a biased estimator of the optimal correction under its loss function. Nonetheless, MSE$^+$ still improves the distribution of bias more than the naïve strategy.

\begin{figure}[htbp]
\FIGURE
{%
  \centering
  \begin{tabular}{ccc}
    \subfigure[Group bias\label{fig:bias_density}]{
      \includegraphics[height=1.65in]{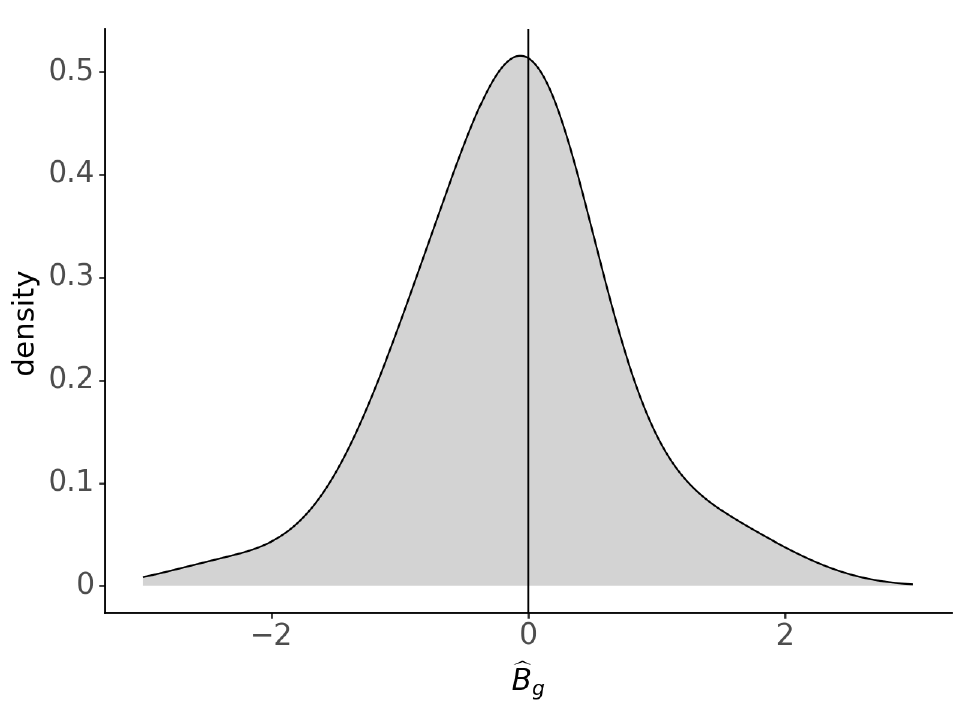}
    }
    &
    \subfigure[Cross-group bias\label{fig:bias_disparity_density}]{
      \includegraphics[height=1.65in]{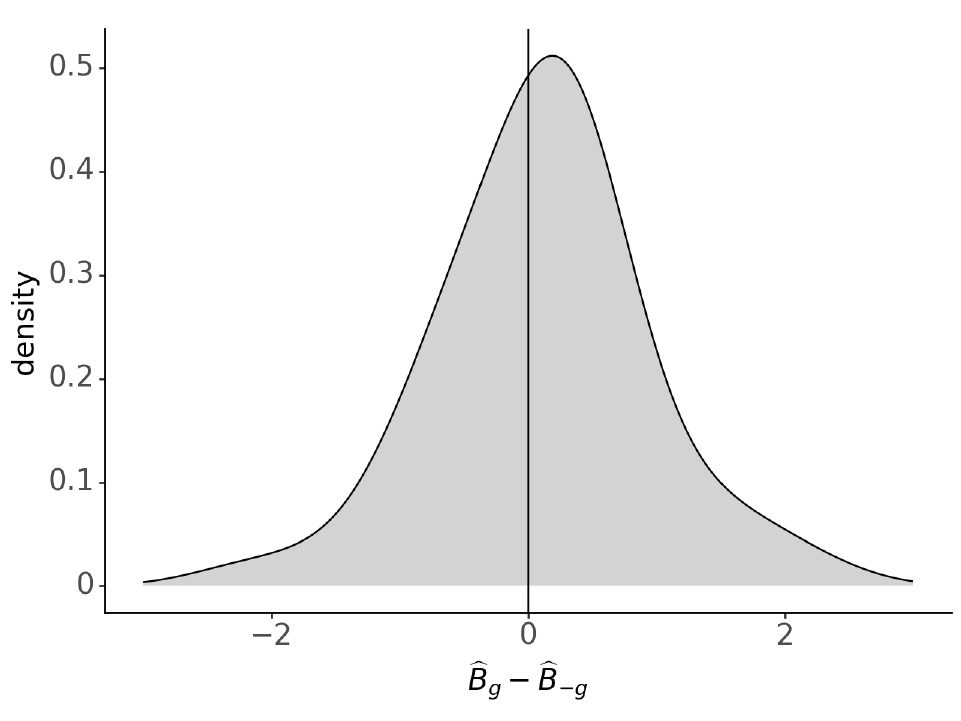}
    }
    &
    \subfigure[$t$-statistic of cross-group bias\label{fig:t_stat_bias_dispariy_density}]{
      \includegraphics[height=1.65in]{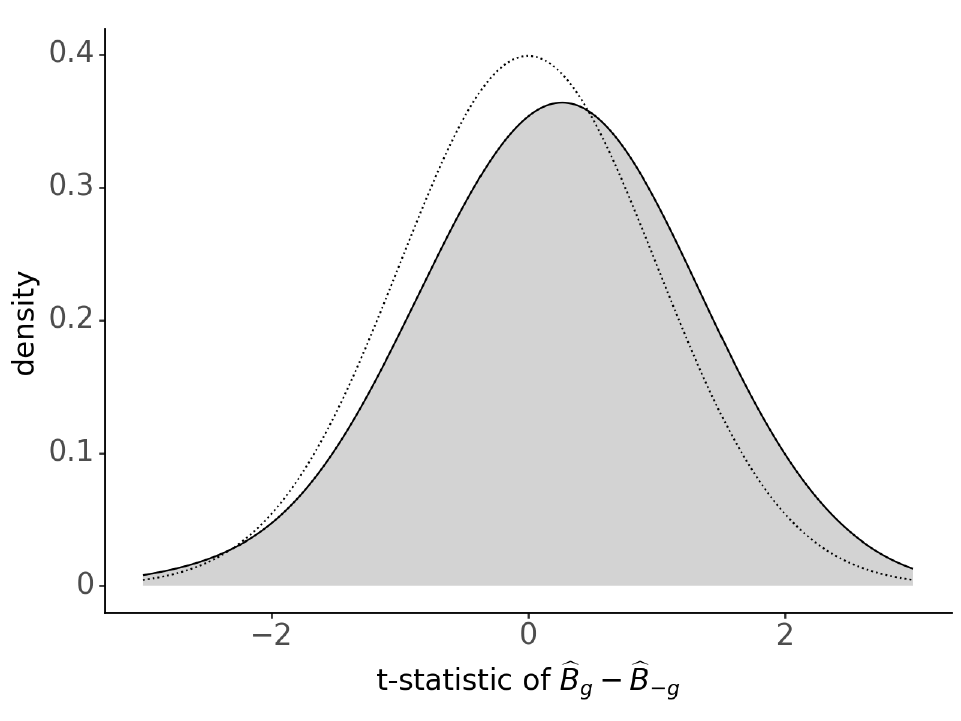}
    }
    \\
    \subfigure[Remaining group bias\label{fig:bias_density_mitigated}]{
      \includegraphics[height=1.9in]{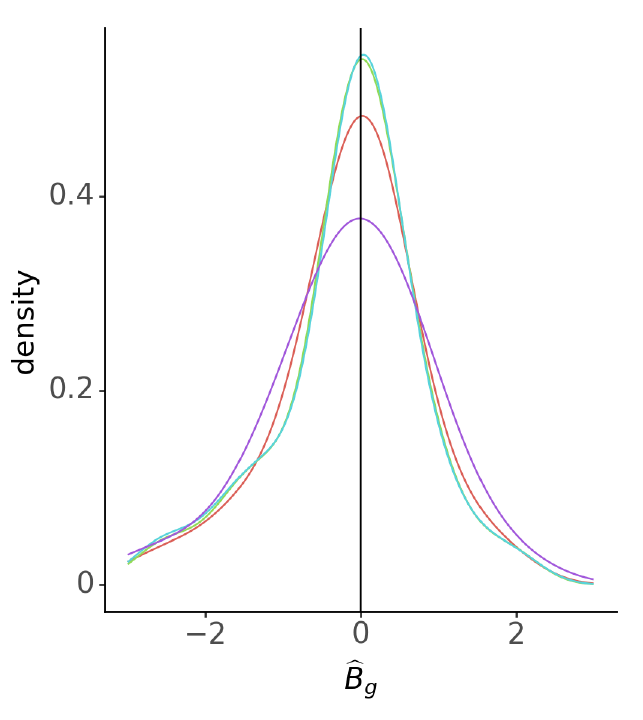}
    }
    &
    \subfigure[Remaining cross-group bias\label{fig:bias_disparity_density_mitigated}]{
      \includegraphics[height=1.9in]{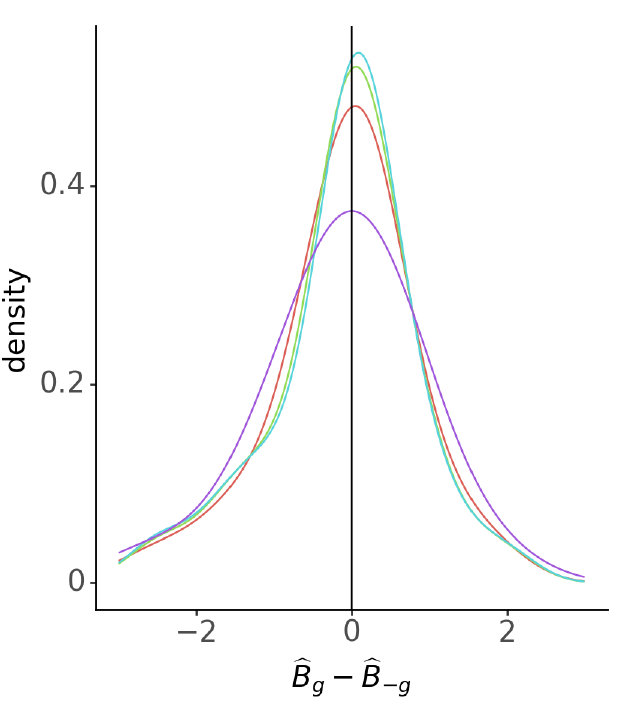}
    }
    &
    \raisebox{5em}[0pt][0pt]{%
      \includegraphics[width=0.1\linewidth]{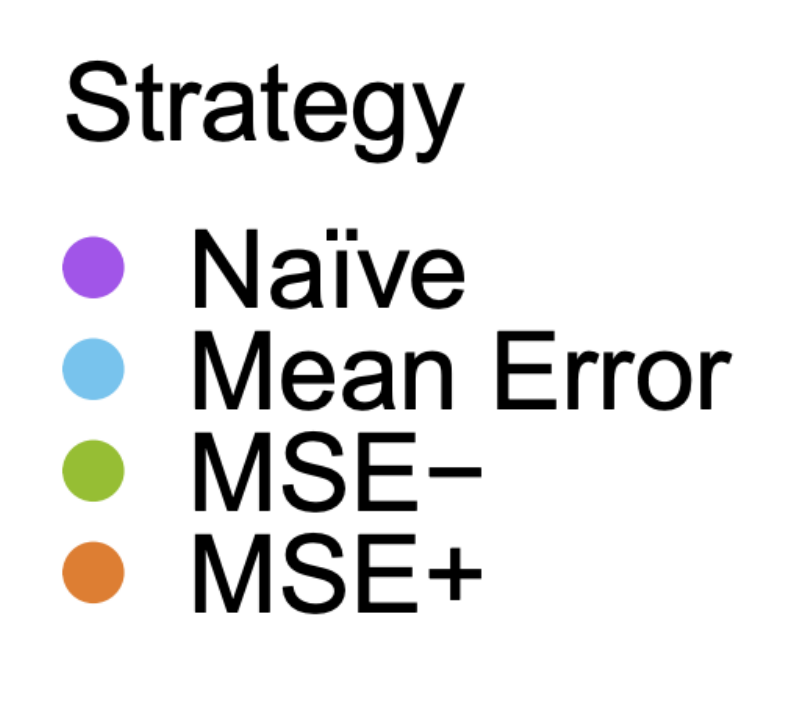}%
    }
  \end{tabular}
}
{
Empirical distributions of group bias and cross-group bias across countries of origin.
\label{fig:bias_dists}
}
{
Panels (a)--(b) show the distributions of group bias and cross-group bias during detection; panel (c) compares the empirical $t$-statistics for cross-group bias to their theoretical distribution under the null of no group bias (dashed line). Panels (d)--(e) report the corresponding distributions after applying each debiasing strategy. Kernel densities use a normal kernel with bandwidth chosen by Silverman’s rule-of-thumb. Estimates were standardized by the mean and standard deviation across the experimental estimated GATEs prior to estimating the densities to obtain an interpretable scale while preserving confidentiality.
}
\end{figure}

\clearpage
\newpage
\section{Differences from \citet{Leng2024} and extension to our framework}\label{app:ld-to-gamma}

This appendix translates the method of \citet{Leng2024} (henceforth: LD) into our notation and estimands in order to clarify how their method would operate if applied to the problem of our paper. We emphasize that the approach of LD is meant for a different objective, and that this appendix serves as a re-interpretation of their method when imported to our mitigation objective. 

Specifically, the original method of LD is to sort observations into bins formed by quantiles of the empirical distribution of CATE predictions, and then regress the average predicted CATE within each bin on experimental difference-in-means estimates of the corresponding bin-level treatment effects. The resulting regression coefficients are used to calibrate the model-implied bin effects towards the experimental estimates \emph{on average} across the CATE distribution. 

A major difference to our framework is that LD define ``groups'' as technical devices endogenously formed by the CATE model, whereas we consider groups to be of intrinsic interest and externally defined (e.g., markets, demographic groups, or geographic regions). Still, one can apply their basic regression approach and then mathematically ``back out'' through derivations what mitigation strategies that would imply in our framework, given our definition of groups. As such, we import the method of LD to our framework as follows: Compute model-implied GATEs by properly collapsing CATE predictions within groups, estimate the GATEs experimentally on the relevant effect scale (additive or relative), and then regress the former on the latter across groups. The key question is then what such a regression method can achieve for our mitigation objective of correcting bias in the model-implied GATE \emph{for each group}. This appendix addresses that question.

Throughout, let $\widehat{\tau}^{\text{pred}}_g \coloneqq \widehat{\tau}^f_g$ denote the model-implied (properly collapsed) predicted GATE for group $g$, and let $\widehat{\tau}^{\text{exp}}_g \coloneqq \widehat{\tau}_g$ denote the corresponding unbiased experimental GATE estimate. In what follows, we first derive the shrinkage factors implied by LD’s approach, also considering two extensions that we include as additional benchmarks in our simulation study (Appendix~\ref{app:simulation}). We then contrast the objective these pooled calibrations implicitly optimize with the group-wise bias correction proposed in our framework.

\subsection{Parametric Pooled-Regression Calibration}

The calibration method proposed by LD fits a pooled affine mapping between model-implied and experimentally estimated group-level effects. Translated into our notation, the method of LD is to fit
\begin{equation}
    \widehat{\tau}^{\text{exp}}_g
  = \alpha + \xi\,\widehat{\tau}^{\text{pred}}_g + \varepsilon_g,
\end{equation}
using weighted least squares across groups, with weights given by the inverse of the variance of the experimental GATEs. This yields the calibrated group-level prediction
\begin{equation}
  \widehat{\tau}^{\text{LD}}_g
  = \widehat{\alpha} + \widehat{\xi}\,\widehat{\tau}^{\text{pred}}_g.
\end{equation}
To compare LD’s approach to ours and back out the implied shrinkage factor, recall that in our framework a debiased group-level effect can be written as
\begin{equation}
    \widehat{\tau}^{\text{pred}}_g - \gamma_g \widehat{B}_g.
\end{equation}
Interpreting LD’s calibrated estimate $\widehat{\tau}^{\text{LD}}_g$ as such a debiased effect and equating the two expressions gives
\begin{equation}
    \widehat{\tau}^{\text{LD}}_g =  \widehat{\tau}^{\text{pred}}_g - \gamma_g \widehat{B}_g.
\end{equation}
Using the identity $\widehat B_g = \widehat{\tau}^{\text{pred}}_g - \widehat{\tau}^{\text{exp}}_g$, we can therefore back out the \emph{implied} shrinkage factor under LD’s method as
\begin{equation}\label{eq:gamma_LD}
\gamma^{\text{LD}}_g
= 1 + \frac{\widehat{\tau}^{\text{exp}}_g - \widehat{\tau}^{\text{LD}}_g}{\widehat B_g}.
\end{equation}
If the affine restriction holds exactly (i.e, $\widehat{\tau}^{\text{LD}}_g = \widehat{\tau}^{\text{exp}}_g$ for all groups), then $\gamma^{\text{LD}}_g = 1$ and LD’s method coincides with the naïve strategy in our framework. Moreover, $\gamma^{\text{LD}}_g$ may lie outside the interval $[0,1]$, whereas our framework restricts the shrinkage factor $\gamma_g$ to this range, corresponding to adjustments from no correction to full correction of the estimated group bias (cf. Section~\ref{sec:mitigation}).

\subsection{Isotonic and Multiplicative Extensions}

For completeness, and to provide additional benchmarks in our simulation study, we also consider two extensions. First, we replace the affine mapping with a monotone calibration function obtained via isotonic regression of $\widehat{\tau}^{\text{exp}}_g$ on $\widehat{\tau}^{\text{pred}}_g$. This yields calibrated effects
\begin{equation}
    \widehat{\tau}^{\text{ISO}}_g = \widehat m(\widehat{\tau}^{\text{pred}}_g),
\end{equation}
where $\widehat m(\cdot)$ is the least-squares monotone fit that maps predicted to experimental group effects while preserving their ordering, implementable with the \texttt{isoreg} function in \textsf{R}.\footnote{\SingleSpacedXI\footnotesize That is, $\widehat m$ is a piecewise-constant nondecreasing function fit on the ordered pairs} The associated shift is $\delta^{\text{ISO}}_g = \widehat{\tau}^{\text{ISO}}_g - \widehat{\tau}^{\text{pred}}_g$, and the implied shrinkage factor $\gamma^{\text{ISO}}_g$ is defined analogously to Eq.~\eqref{eq:gamma_LD}, replacing $\widehat{\tau}^{\text{LD}}_g$ by $\widehat{\tau}^{\text{ISO}}_g$. Because we have only a small number of groups, isotonic calibration serves as a parsimonious nonparametric benchmark: it enforces the natural requirement that groups with higher predicted GATEs do not receive lower calibrated GATEs, without the instability that would arise from fitting flexible unconstrained regressors to very few calibration points.

Second, we consider a multiplicative calibration map obtained by applying the affine calibration on the log scale. Specifically, we estimate
\begin{equation}
\log \widehat{\tau}^{\text{exp}}_g
= \alpha + \xi \log \widehat{\tau}^{\text{pred}}_g + \varepsilon_g,
\end{equation}
and then transform back to levels. Exponentiating both sides yields a power-form mapping,
\begin{equation}
\widehat{\tau}^{\text{cal}}_g
= \exp(\widehat{\alpha})\,\widehat{\tau}^{\text{pred}}_g{}^{\widehat{\xi}},
\end{equation}
so that calibration acts multiplicatively on the predicted group effects rather than as an additive shift. This extension is useful when deviations between predicted and experimental group effects are more naturally modeled as proportional rather than additive, for example when the errors between the model-implied and experimental GATE scale with the level of the former. As in the additive case, however, the calibration remains pooled across groups and therefore does not adapt to heterogeneous estimation precision of $\widehat B_g$.

\subsection{What Pooled Calibration Optimizes for}

The key distinction between LD’s approach and ours lies in the objective being optimized. LD solves
\begin{equation}
  \min_{\alpha,\xi}
  \sum_g \omega_g
  \big(
    \widehat{\tau}^{\text{exp}}_g
    - (\alpha + \xi\,\widehat{\tau}^{\text{pred}}_g)
  \big)^2,
\end{equation}
with weights $\omega_g$ proportional to the inverse variance of the experimental group-level estimates. Isotonic calibration replaces the affine mapping with a nondecreasing function but retains the same pooled objective. Rewriting this objective in our notation shows that LD choose pooled shifts $\delta_g$ to minimize
\begin{equation}
  \sum_g \omega_g (\widehat B_g + \delta_g)^2,
\end{equation}
thereby achieving mean calibration \emph{across} groups. In contrast, our method selects group-specific shrinkage factors by (in the case of MSE loss, for ease of exposition) minimizing the expected squared debiasing error,
\begin{equation}
  \min_{\gamma_g \in [0,1]}
  \E\!\left[(b_g - \gamma_g \widehat B_g)^2\right]
  = (1-\gamma_g)^2 b_g^2 + \gamma_g^2 \sigma_g^2,
\end{equation}
with oracle solution $\gamma^{\text{MSE}}_g = b_g^2/(b_g^2 + \sigma_g^2)$; see Proposition \ref{prop:opt_mse}. 

\subsection{When Do the Approaches Coincide?}

The two approaches coincide or diverge depending on the structure of group-level bias.
\begin{enumerate}
    \item \textbf{Exact pooled calibration.}  
    If experimental and predicted group-level effects are related by an exact affine (or monotone) mapping, LD’s calibration recovers the experimental effects exactly and implies $\gamma_g = 1$ for all groups, coinciding with naïve debiasing, which we show is generally suboptimal for the objective of recovering GATE using a model of CATE.
    
    \item \textbf{High signal-to-noise and approximate linearity.}  
    When pooled calibration error is small relative to $\widehat B_g$ and estimation noise is limited, LD’s implied $\gamma_g$ will be close to one. In contrast, when $\widehat B_g$ is noisy (due to small group size or weak effects) our risk-minimizing $\gamma_g$ shrinks toward zero, a behavior not shared by pooled calibration.
    
    \item \textbf{Heterogeneous or misspecified bias.}  
    When no single affine or monotone map captures the relationship between predicted and experimental effects across groups, pooled calibration may over- or under-correct individual groups, with implied $\gamma_g$ outside $[0,1]$. Our group-wise shrinkage remains well-defined in this setting.
\end{enumerate}

\newpage
\section{Estimator of Predicted GATE from Binary Positive Outcomes}
\label{app:retrospective_estimator}

We describe the estimator of the model-implied GATE used in the Booking.com application (Section \ref{sec:application}), which only uses observations with the positive-event binary outcome (i.e., from conversions). This estimator also applies in other contexts where it is preferable to only store or use data from the positive events, such as purchase scanner data or data from hospitals.

We first cover identification. Let $Y\in\{0,1\}$ denote a binary outcome. Units enter the estimation sample only when $Y=1$ (i.e., case-only sampling):
\begin{equation}
\Pr(\text{included}\mid Y,X,T,G=g)=\rho_g(Y),
\qquad
\rho_g(1)>0,\;\rho_g(0)=0.
\label{eq:caseonly}
\end{equation}
Under random treatment assignment and SUTVA, the model-implied GATE $\tau^{f}_g$ can then be recovered from positive outcomes using the predicted CATE $\widehat\tau^{f}(X_i)$.

Let $\widehat\lambda_g^{(1)}$ and $\widehat\lambda_g^{(0)}$ denote the average predicted CATE among positive-outcome treated and control units in group $g$; that is,
\begin{equation}
\widehat\lambda^{(1)}_g
=
\frac{1}{\sum_{i\in g}T_iY_i}\sum_{i\in g}T_iY_i\,\widehat\tau^{f}(X_i),
\qquad
\widehat\lambda^{(0)}_g
=
\frac{1}{\sum_{i\in g}(1-T_i)Y_i}\sum_{i\in g}(1-T_i)Y_i\,\widehat\tau^{f}(X_i).
\end{equation}
The predicted GATE is then estimated as
\begin{equation}\label{eq:pred_ATE_binary}
\widehat\tau^{f}_g
=
\frac{
\sum_{i\in g}(1-T_i)Y_i\,(\widehat\lambda^{(0)}_g)^{2}
+\sum_{i\in g}T_iY_i\,\widehat\lambda^{(1)}_g
}{
\sum_{i\in g}(1-T_i)Y_i\,\widehat\lambda^{(0)}_g
+\sum_{i\in g}T_iY_i
}.
\end{equation}
This estimator is a variant of the retrospective estimator proposed by \citet{Goldenberg2020}. 

It is straightforward to verify that the estimator in Eq.~\eqref{eq:pred_ATE_binary} satisfies collapsibility: if the true CATE is constant within a group, $\tau(X)=\tau_g$, then $\widehat\lambda_g^{(1)}=\widehat\lambda_g^{(0)}=\tau_g$, and the estimator reduces to $\widehat\tau_g^f=\tau_g$. Specifically, we have 
\begin{equation}
    \widehat\lambda^{(1)}_g 
    =
    \frac{1}{ \sum_{i\in g} T_iY_i}  \sum_{i\in g} T_iY_i \tau_g
    =
    \tau_g,
    \quad
    \widehat\lambda^{(0)}_g 
    =
    \frac{1}{\sum_{i\in g}(1-T_i)Y_i}\sum_{i\in g}(1-T_i)Y_i\,\tau_g
    =
    \tau_g
    .
\end{equation}
Plugging these into the estimator in Eq.~\eqref{eq:pred_ATE_binary} yields
\begin{align}
    \widehat\tau^{f}_g
    = \frac{ \sum_{i\in g} (1-T_i)Y_i \tau_g^2 + \sum_{i\in g} T_iY_i \tau_g}
    { \sum_{i\in g} (1-T_i)Y_i \tau_g +  \sum_{i\in g} T_iY_i } 
    =
    \tau_g
    \times
    \frac{ \sum_{i\in g} (1-T_i)Y_i \tau_g + \sum_{i\in g} T_iY_i }
    { \sum_{i\in g} (1-T_i)Y_i \tau_g + \sum_{i\in g} T_iY_i } 
    =
    \tau_g,
\end{align}
as we were supposed to show. Therefore, the estimator aggregates individual-level CATE predictions in a manner that recovers the (potentially biased) model-implied GATE under case-only sampling.

\end{APPENDICES}

\end{document}